\documentclass[a4paper,10pt]{article}
\usepackage[utf8x]{inputenc}
\usepackage{amsmath}
\usepackage{amsthm}
\usepackage{comment}
\usepackage{color}
\usepackage{epstopdf}
\usepackage{pdfpages}
\usepackage{amsfonts}
\usepackage{amssymb}
\usepackage{hyperref}

\usepackage{color}

\newcommand{\clu}{\textcolor[rgb]{0.00,0.00,0.00}}

\usepackage[margin=1.2in]{geometry}
\usepackage{cancel,amsthm,amssymb,amsmath,mathrsfs}
\usepackage{color,tikz}

\usepackage[toc,page]{appendix}

\usetikzlibrary{patterns, decorations.markings,arrows, calc}

\newtheorem{theorem}{Theorem}[section]
\newtheorem{lem}[theorem]{Lemma}
\newtheorem{remark}[theorem]{Remark}
\newtheorem{Def}[theorem]{Definition}
\newtheorem{prop}[theorem]{Proposition}

\newtheorem{corollary}[theorem]{Corollary}

\newtheorem{Properties}[theorem]{Properties}

\newtheorem{RH}{Riemann-Hilbert problem}

\newenvironment{proofof}{{\textit{Proof}}}{\hfill$\scriptstyle\square$\vskip1cm}

\renewcommand{\(}{\left(}
\renewcommand{\)}{\right)}
\newcommand{\ol}{\overline}
\renewcommand{\i}{\mathit{\mathbf\mathit{i}}}
\newcommand{\e}{\mathrm{e}}
\newcommand{\dsfrac}{\displaystyle\frac}
\renewcommand{\d}{\mathrm{d}}

\title{On a class of compact perturbations of the special pole-free joint solution of KdV and $P_I^2.$}
\author{B.~Dubrovin and A.~Minakov}

\begin{document}

\maketitle

\begin{abstract}
We consider perturbations of the special pole-free joint solution $U(x,t)$ of the Korteweg--de Vries equation $u_t+uu_x+\frac{1}{12}u_{xxx}=0$ and $P_I^2$ equation
$u_{xxxx}+10u_x^2+20uu_{xx}+40(u^3-6tu+6x)=0$
under the action of the KdV flow.
We show that if the perturbation is compact and of bounded variation, then the initial value problem for the KdV equation has 
a classical solution. Our method is the inverse scattering transform method in the form of the Riemann-Hilbert problem method.
Namely, we construct the corresponding spectral functions $a(\lambda), r(\lambda),$ and give characterization of the 
compact perturbations in terms of $a(\lambda), r(\lambda).$
\end{abstract}

\tableofcontents

\section{Introduction}
We construct a class of classical smooth solutions of the Korteweg--de Vries equation (KdV), which are unbounded as $x\to\pm\infty.$

There are several papers concerning unbounded classical solutions of KdV. 

\medskip
In \cite{Menikoff} the author showed that if an initial datum $u_0(x)$ satisfies $$\dfrac{\partial^n }{\partial x^n}u_0(x)={\ol{_\mathcal{O}}}(|x|^{1-n})\quad\mbox{ as }\quad x\to\pm\infty,\qquad 0\leq n\leq7,$$
where ${\ol{_\mathcal{O}}}$ is the standard Bachmann--Landau little-$o$ notation,
then the initial value problem for KdV has a unique solution that exists for all times $-\infty<t<+\infty.$
Solutions which admit asymptotics with power growth at $x\to\pm\infty$ were considered in \cite{BSh}. 
Initial data with a polynomial growth as $x\to\pm\infty$ were considered by \cite{KPV}.

The Cauchy problem for KdV with an initial datum, which is a perturbation of different (bounded) quasi-periodic solutions as $x\to\pm\infty$
was investigated by \cite{Egorova_Teschl}.  We refer also to \cite{BS}, \cite{MaS}, where KdV equation with vanishing datum is treated in different Sobolev spaces.

We develop formalism of direct and inverse scattering transform for the Schr\"{o}dinger operator with a potential from a class of unbounded at $x\to\pm\infty$ functions. Such a study for perturbations of finite-gap potentials (bounded at $x\pm\infty$) was done in \cite{BET08}.

The case of an initial datum, which decays exponentially as $x\to+\infty,$ and might be unbounded as $x\to-\infty,$ was studied by
\cite{Rybkin}.

In \cite{Claeys Vanlessen} the authors proved that one special solution $U(x,t)$ to the Korteweg--de Vries equation
\begin{equation}\label{KdV}u_t+uu_x+\frac{1}{12}u_{xxx}=0\end{equation}
 which behaves as $\sqrt[3]{-6x}$ as $x\to\pm\infty,$
and is simultaneously a solution to the second member of the first Painlevé hierarchy ($P_I^2$)
\begin{equation}\label{PI2} x=tu-\(\frac16 u^3+\frac{1}{24}(u_x^2+2uu_{xx})+\frac{1}{240}u_{xxxx}\)\end{equation}
 is pole free for all $(x,t)\in\mathbb{R}^2.$
This is a part of the so-called {\textit{Universality Conjecture}} \cite{Dubrovin}, that states that the function $U(x,t)$ describes the behavior of a generic solution to the general perturbed Hamiltonian equation
$$u_t+a(u)u_x+
\varepsilon^2\left[b_3(u)u_{xxx}+b_4(u)u_xu_{xx}+
b_5(u)u_x^3\right]+\ldots=0$$
 near the point of gradient catastrophe of the unperturbed solution
$$u_t+a(u)u_x=0.$$
\clu{In \cite{CG09} this was proved for the equations from the KdV hierarchy with some restrictions on the initial data.}
It also was conjectured by \cite{Dubrovin} that $U(x,t)$ is the {\textit{unique real smooth}} for all $x,t\in\mathbb{R}$ solution to \eqref{PI2}. To the best of our knowledge, the part of the Dubrovin's {\textit{Universality Conjecture}} that there are no other real smooth for all $x,t\in\mathbb{R}$ solutions to equation \eqref{PI2} remains open. Class of degenerate {\textit{tritronquee}} solutions of \eqref{PI2} was studied in \cite{Grava Kapaev Klein}.

Equation \eqref{PI2} for the particular value of the parameter $t=0$ was studied by \cite{Kapaev}, and it also appeared in the study of the double scaling limit for the matrix model with the multicritical index $m=3$ \cite{BrezinMarinariParisi}. 

In \cite{Claeys10} the long-time asymptotics as $t\to+\infty$ of $U(x,t)$ were studied, and the complete asymptotic expansion for large $x$ was obtained in \cite{Sul}.
Further, similar \clu{common} smooth solutions of KdV \eqref{KdV} and \clu{equations of} the first Painlevé hierarchy were constructed in \cite{Claeys12}. They behave like $c\sqrt[2m+1]{-x}$ as $x\to\pm\infty,$ $m\geq 1.$

In this paper we take a compactly supported perturbation of $U(x,t_0)$ at a given time $t_0,$ and study the Cauchy problem for KdV equation \eqref{KdV} with this initial datum.


To formulate our main results, we recall  that the special solution $U(x,t)$ of  the equations \eqref{KdV} and \eqref{PI2} can be constructed in the following way \cite{Claeys Vanlessen}, \cite{Grava Kapaev Klein}, 
 \cite{Dubrovin}, \cite{Kapaev}:
consider the Riemann-Hilbert problem (RHP) (see Figure \ref{Fig_E}):
\begin{RH}\label{RHE}
 Find a $2\times2$ matrix-valued function $\mathbb{E}(x,t;\lambda)$, which 
\begin{enumerate}
 \item 
is analytic in $\lambda\in\mathbb{C}\setminus\Sigma$ and continuous up to the boundary, where $\Sigma$ is the contour
\begin{equation}\label{Sigma}\Sigma = \mathbb{R} \cup \left(e^{6\pi i/7}\infty,0\right) \cup \left(e^{-6\pi i/7}\infty,0\right),\end{equation} with the orientation as is written (we denote by $(0, \e^{\i\alpha})$  the ray emanating from the origin and coming to infinity at an angle $\alpha,$ and $(\e^{\i\alpha},0)$ means the ray with the opposite orientation);
\item satisfies the jump conditions $\mathbb{E}_+ =\mathbb{E}_-  J_{\mathbb{E}}$ on the contour $\Sigma,$
 where
$$J_{\mathbb{E}}=\begin{pmatrix}
     1&i\\0&1
    \end{pmatrix},\ \lambda\in \gamma_3:=\left(e^{6\pi i/7}\infty,0\right) \mbox{ and }\ \lambda\in\gamma_{-3}:=\left(e^{-6\pi i/7}\infty,0\right),
$$$$
J_{\mathbb{E}}=\begin{pmatrix}
     1&0\\-i&1
    \end{pmatrix},\ \lambda\in \gamma_0:=(0,\infty),
\quad
J_{\mathbb{E}}=\begin{pmatrix}
     0&-i\\-i&0
    \end{pmatrix},\ \lambda\in \rho:=(-\infty,0),
$$
Here $\mathbb{E}_{\pm}$ stands for the limiting values of $\mathbb{E}$ on the contour $\Sigma.$ The positive side of the contour is from the left, the negative one is from the right; these relations define segments $\gamma_{0,3,-3},\rho;$
\item has the following asymptotics as $\lambda\to\infty:$
 \begin{equation}\label{eq_3a_E_asymp}\mathbb{E}(\lambda)=
\frac{\lambda^{-\sigma_3/4}}{\sqrt{2}}\begin{pmatrix}
                                      1&1\\1&-1
                                     \end{pmatrix}(I+\frac{\widetilde b\sigma_3}{\sqrt{\lambda}}+\mathcal{O}(\lambda^{-1}))
\e^{\theta \sigma_3},\textrm{ where } \theta(x,t;\lambda):=\frac{1}{105}\lambda^{7/2}-\frac{t}{3}\lambda^{3/2}+x\lambda^{1/2},
\end{equation}
where $\widetilde b=\widetilde b(x,t)$ is a scalar (which is not fixed, but is introduced to fix the form of the asymptotics). 
We take the standard branch along $-\infty,0$ for roots of $\lambda.$
Here $$I=\begin{pmatrix}1&0\\0&1\end{pmatrix},\quad \sigma_3=\begin{pmatrix}1&0\\0&-1\end{pmatrix}\quad \mbox{ and }\quad \lambda^{-\sigma_3/4}=\begin{pmatrix}\lambda^{-1/4}&0\\0&\lambda^{1/4}\end{pmatrix},\quad \e^{\theta\sigma_3}=\begin{pmatrix}\e^{\theta}&0\\0&\e^{-\theta}\end{pmatrix}$$ are the matrix exponents.
\end{enumerate}
\end{RH}

\noindent {\textbf{Reconstruction of $U(x,t)$ from $\mathbb{E}(x,t;\lambda).$}} Having the solution $\mathbb{E}(x,t;\lambda)$ of the RHP \ref{RHE}, the function $U(x,t)$ can be defined as 
\begin{equation}\label{U_RHE}U(x,t):=2\lim\limits_{\lambda\to\infty}\lambda\left[\frac{1}{\sqrt{2}}(\sigma_3+\sigma_1)\lambda^{\sigma_3/4}\mathbb{E}(x,t;\lambda)\e^{-\theta\sigma_3}\right]_{12},\quad\mbox{where}\ \sigma_1=\begin{pmatrix}0&1\\1&0\end{pmatrix},\end{equation}
and the subscript $_{12}$ denotes the element staying in the intersection of the first row and the second column.

\noindent {\textbf{Jost solutions associated with $U(x,t).$}}
Function $\mathbb{E}(x,t;\lambda)$ satisfies the differential equation
$$\mathbb{E}_x(x,t;\lambda)=\begin{pmatrix}0&1\\\lambda-2U(x,t)&0\end{pmatrix}\mathbb{E}(x,t;\lambda),$$
where the subscript $_x$ denotes the differentiation  with respect to (w.r.t.) $x.$ Hence, the elements of the matrix-valued function  $\mathbb{E}(x,t;\lambda)$ have the following structure in different parts of the complex $\lambda$ plane:
\begin{equation}\label{Eluldr_from_E}\mathbb{E}(x,t;\lambda)=:\begin{cases}
(E_{lu}(x,t;\lambda),E_{r}(x,t;\lambda))=:\begin{pmatrix}e_{lu}(x,t;\lambda)&e_{r}(x,t;\lambda)
\\
e_{lu,x}(x,t;\lambda)&e_{r,x}(x,t;\lambda)\end{pmatrix},\arg\lambda\in(0,\frac{6\pi}{7}),
\\
(E_{ld}(x,t;\lambda),E_{r}(x,t;\lambda))=:\begin{pmatrix}e_{ld}(x,t;\lambda)&e_{r}(x,t;\lambda)
\\
e_{ld,x}(x,t;\lambda)&e_{r,x}(x,t;\lambda)\end{pmatrix},\arg\lambda\in(-\frac{6\pi}{7},0),
\\
(E_{lu}(x,t;\lambda),-\i
E_{ld}(x,t;\lambda))=:\begin{pmatrix}e_{lu}(x,t;\lambda)&-\i
e_{ld}(x,t;\lambda)
\\
e_{lu,x}(x,t;\lambda)& -\i
e_{ld,x}(x,t;\lambda)\end{pmatrix},\arg\lambda\in(\frac{6\pi}{7},\pi),
\\
(E_{ld}(x,t;\lambda),\i
E_{lu}(x,t;\lambda))=:\begin{pmatrix}e_{ld}(x,t;\lambda)&\i
e_{lu}(x,t;\lambda)
\\
e_{ld,x}(x,t;\lambda)&\i
e_{lu,x}(x,t;\lambda)\end{pmatrix},\arg\lambda\in(-\pi,-\frac{6\pi}{7}).
\end{cases}\end{equation}
The above formula \eqref{Eluldr_from_E} defines scalar functions $e_{lu},$ $e_{ld},$ $e_r,$ and their vector counterparts $E_{lu,} $ $E_{ld,}$ $E_r,$ 
where the subscripts $_r,$ $_{lu},$ $_{ld}$ mean "right", "left, upper half-plane", "left, lower half-plane", 
and the subscript $_x$ means the differentiation w.r.t. $x.$
Further, out of functions $E_{lu},$ $E_{ld}$ we construct the piece-wise analytic functions $E_l,$ $e_l$
$$
E_l(x,t;\lambda)=\begin{pmatrix}e_l(x,t;\lambda)\\e_{l,x}(x,t;\lambda)\end{pmatrix}:=\begin{cases}E_{lu}(x,t;\lambda),\ \Im\lambda>0,
\\
E_{ld}(x,t;\lambda),\ \Im\lambda<0,
\end{cases}$$
which are discontinuous across the real line $\lambda\in\mathbb{R}$ and continuous up to the boundary.  We call the functions $E_l(x,t;\lambda),$ $e_l(x,t;\lambda)$
the vector and scalar \textit{left Jost solutions,} associated with $U(x,t),$ respectively,
and the functions $E_r(x,t;\lambda),$ $e_r(x,t;\lambda)$ the vector and scalar \textit{ right Jost solutions,} associated with $U(x,t).$
Functions $e_r, e_l$ are solutions to the associated spectral problem
\begin{equation}\label{exx}e_{xx}(x,t;\lambda)+2U(x,t)e(x,t;\lambda)=\lambda e(x,t;\lambda),\end{equation}
whence the definition (here the subscript $_{xx}$ denotes the second derivative w.r.t. to $x$).


\noindent Our first preliminary result describes the properties of the associated with $U(x,t)$ Jost solutions:

\begin{lem}\label{lem_prel}
For any $\lambda\in\mathbb{C},$ $t\in\mathbb{R},$
function $E_r(x,t;\lambda)$ vanishes exponentially as $x\to+\infty.$

\noindent For any $\lambda\in\mathbb{C}\setminus\mathbb{R},$ $t\in\mathbb{R},$ function $E_l(x,t;\lambda)$ vanishes exponentially as $x\to-\infty.$

\noindent  For $\lambda\in\mathbb{R}\pm\i 0,$ $t\in\mathbb{R},$ function $E_l(x,t;\lambda)$ might have at most polynomial growth in $x$ as $x\to-\infty.$
\end{lem}
This lemma justifies the names ''left'', ''right'' in the definitions of the Jost solutions. More detailed behavior of Jost \clu{solutions is given} in Lemma \ref{lemma_E_rl} below.

Define the following classes of functions.
\begin{Def}\label{def_BVloc}
Denote by $BV_{loc}=BV_{loc}(\mathbb{R})$ the class of functions $u:\mathbb{R}\mapsto\mathbb{R},$ which on every compact subset of $\mathbb{R}$ are functions of bounded variation.

Denote by $BV_{loc}^{(n)}=BV_{loc}^{(n)}(\mathbb{R}),$ $n=0,1,2,\ldots,$ the class of functions $u:\mathbb{R}\mapsto\mathbb{R},$ which are $n$ times differentiable, and the $n^{th}$ derivative $u^{(n)}$ belongs to $BV_{loc}.$ We set $BV_{loc}^{(0)}\equiv BV_{loc}.$
Finally, the class $BV_{loc}^{(\infty)}(\mathbb{R})$ of functions, which belong to $BV_{loc}^{(n)}$ for any $n\geq0,$ is $C^{\infty}(\mathbb{R}).$
\end{Def}

Now we are ready to formulate our first main result.

\begin{theorem}\label{teor_main} (forward scattering)
Let $t_0,$ $A<B$ are real, and let function $u_{t_0}(x)$ be a function of the class $BV_{loc}(\mathbb{R})$, which equals $U(x,t_0)$ for $x\in(-\infty,A)\cup(B,+\infty).$ Then

\textbf{I.} The associated spectral problem $$\partial_x^2{f_{t_0}}(x;\lambda)+2u_{t_0}(x) f_{t_0}(x;\lambda)=\lambda f_{t_0}(x;\lambda)$$
has two solutions $f_r,$ $f_l$ (we call them Jost solutions associated with $u_{t_0}(x)$), which are determined by the conditions
$$f_r(x,t_0;\lambda)=e_r(x,t_0;\lambda),\ \forall x>B,\qquad f_l(x,t_0;\lambda)=e_l(x,t_0;\lambda),\ \forall x<A$$
on their domains of definition, and which are analytic respectively in $\lambda\in\mathbb{C}$ and $\lambda\in\mathbb{C}\setminus\mathbb{R}.$

\textbf{II.} Define functions $a(\lambda)\equiv a(\lambda;t_0),$ $b(\lambda)\equiv b(\lambda; t_0),$ $r(\lambda)\equiv r(\lambda; t_0)$ (we call them the {\textbf{spectral functions}} associated with $u_{t_0}(x)$) as follows:
$$a(\lambda; t_0):=\left\{f_r(x,t_0;\lambda), f_l(x,t_0;\lambda)\right\},$$
where the bracket $\left\{g,h\right\}\equiv W\left\{g,h\right\}=gh_x-g_xh$ denotes the Wronskian of two functions, and
$$b(\lambda;t_0):=\(f_l(x,t_0;\lambda)-a(\lambda;t_0)e_l(x,t_0;\lambda)\)\cdot (e_r(x,t_0;\lambda))^{-1},x>B, \quad r(\lambda;t_0)=b(\lambda;t_0)a^{-1}(\lambda;t_0).$$ The spectral functions possess the following properties:
\begin{enumerate}\item\label{prop_ruan} the restrictions of $a(\lambda), b(\lambda), r(\lambda)$ to the upper half-plane $\Im\lambda>0,$ we call them $a_u(\lambda), $  $b_u(\lambda),$ $r_u(\lambda),$ can be extended analytically to $\mathbb{C};$ 
\\
the restrictions of $a(\lambda),$ $b(\lambda),$ $r(\lambda)$ to the lower half-plane $\Im\lambda<0$, we call them $a_d(\lambda),$ $b_d(\lambda),$ $r_d(\lambda),$ are related with $a_u, b_u, r_u$ as 
$\ol{a_u(\ol{\lambda})}=a_d(\lambda),$ $\ol{b_u(\ol{\lambda})}=b_d(\lambda),$ $\ol{r_u(\ol{\lambda})}=r_d(\lambda);$
\item \label{prop_Imru} $\Im r_u(s)=\mathcal{O}(|s|^{-1})$ for $s\in\mathbb{R},$ $s\to\pm\infty$ (here $\mathcal{O}$ is the  Bachmann--Landau big-O notation);
\item\label{prop_Imru12} $\Im r_u(s)<\frac12$ and $\Im r_d(s)>-\frac12$ for $s\in\mathbb{R};$
\item \label{prop_rurdneqi}$r_u(\lambda)-\ol{r_u(\ol{\lambda})}\neq \i$ for all $\lambda\in\mathbb{C};$
\item \label{prop_aasymp} $a(\lambda)=1+\frac{1}{\sqrt{\lambda}}\int\limits_A^B(U(x,t_0)-u_{t_0}(x))\d x+{\ol{_\mathcal{O}}}(\frac{1}{\sqrt{\lambda}})$ as $\lambda\to\infty$ uniformly w.r.t. $\arg\lambda\in[-\pi,\pi];$
\item \label {prop_ru_vanishing}$r_u(\lambda)=\mathcal{O}(\frac{1}{\lambda})\cdot\e^{2\theta(B,t_0;\lambda)}$ as $\lambda\to\infty$ uniformly w.r.t. $\arg\lambda\in[-\pi,\pi];$
\\ $6^{(N)}$. Furthermore, if $u_{t_0}\in BV^{(N)}_{loc},$ then
$r_u(\lambda)=\mathcal{O}({\lambda}^{-\frac{N}{2}-1})\cdot\e^{2\theta(B,t_0;\lambda)}.$
\\$6^{(\infty)}$. Furthermore, if $u_{t_0}\in BV^{(\infty)}_{loc},$ then
for any $N,$
$r_u(\lambda)=\mathcal{O}({\lambda}^{-\frac{N}{2}-1})\cdot\e^{2\theta(B,t_0;\lambda)}.$
\item \label{prop_auad}
$a_u(\lambda;t_0)-\ol{a_u(\ol{\lambda};t_0)}=\mathcal{O}(\frac{1}{\lambda})\cdot \e^{-2\theta(A,t_0;\lambda)}$ as $\lambda\to\infty$ uniformly w.r.t. $\arg\lambda\in[-\pi,\pi];$
\item \label{prop_ra}$r_u(\lambda)-r_d(\lambda)=\i\(1-\dfrac{1}{a_u(\lambda)a_d(\lambda)}\);$
\item function $a(\lambda)$ does not vanish nowhere, i.e. $a_u(\lambda)\neq 0$ for $\Im\lambda\geq 0;$
\item \label{prop_ainr} functions $a_u,$ $a_d$ can be expressed in terms of $r_u$ in the following way:
\begin{equation}\label{au_ru}\begin{split}a_u(\lambda)&=\exp\left\{\frac{-1}{2\pi\i}
\int\limits_{-\infty}^{+\infty}\frac{\ln(1-2\Im r_u(s)) \d s}{s-\lambda}\right\},
\quad\mbox{ for } \Im\lambda>0,
\\
&=\frac{1}{1+\i \(r_u(\lambda)-\ol{r_u(\ol\lambda)}\)}\cdot
\exp\left\{\frac{-1}{2\pi\i}\int\limits_{-\infty}^{+\infty}\frac{\ln(1-2\Im r_u(s)) \d s}{s-\lambda}\right\},
\ \mbox{ for } \Im\lambda<0,
\end{split}\end{equation}
and 
\begin{equation}\label{ad_ru}\begin{split}
a_d(\lambda)&=\frac{1}{1+\i \(r_u(\lambda)-\ol{r_u(\ol\lambda)}\)}\cdot
\exp\left\{\frac{1}{2\pi\i}\int\limits_{-\infty}^{+\infty}\frac{\ln(1-2\Im r_u(s)) \d s}{s-\lambda}\right\},
\mbox{ for } \Im\lambda>0,
\\
&=\exp\left\{\frac{1}{2\pi\i}\int\limits_{-\infty}^{+\infty}\frac{\ln(1-2\Im r_u(s)) \d s}{s-\lambda}\right\},
\mbox{ for } \Im\lambda<0,\end{split}\end{equation}
and hence, not only $a_u(\lambda)\neq 0$ for $\Im\lambda\geq 0,$ but $a_u(\lambda)\neq 0$ for all $\lambda\in\mathbb{C}$.
\end{enumerate}
\end{theorem}
\vskip5mm

Let us observe, that the property \ref{prop_ru_vanishing} of the part II of Theorem \ref{teor_main} implies that there is no more than finite number of points $\lambda$  in
the \clu{sector} $\arg\lambda\in(\frac{6\pi}{7},\pi),$  at which $r_u(\lambda)=\i;$
and 
the property \ref{prop_Imru12} implies that $r_u(s)\neq \i$ for real $s.$ 
 Define the \clu{ray} $\gamma_{3}^{\mathbb{F}}$ as follows: if there are no roots of $r_u(\lambda)=\i$ on $\gamma_3=(\e^{6\pi\i/7}\infty,0),$ then $\gamma_3^{\mathbb{F}}:=\gamma_3.$ If there are some roots of $r_u(\lambda)=\i$ on $\gamma_3,$ then we move locally the \clu{ray} $\gamma_3$ in such a way, that it would not contain roots of $r_u(\lambda)=\i$ anymore, and call the resulting contour by $\gamma_{3}^{\mathbb{F}}.$ The contour $\gamma_{-3}^{\mathbb{F}}=\ol{\gamma_3^{\mathbb{F}}}$ is symmetric to $\gamma_3^{\mathbb{F}}$ w.r.t the real line, with orientation imposed by this symmetry. Let us call the domain \clu{included} between $\gamma_0$ and $\gamma_3^{\mathbb{F}}$ by $\Omega_I,$ the one between $\gamma_{3}^{\mathbb{F}}$ and $\rho$ by $\Omega_{II},$ the one between $\rho$ and $\gamma_{-3}^{\mathbb{F}}$ by $\Omega_{III},$ and the one between $\gamma_{-3}^{\mathbb{F}}$ and $\gamma_0$ by $\Omega_{IV}.$ Denote by $\lambda_j,$ $j=1,\ldots,J$ the points in $\Omega_{II}$ at which $r_u(\lambda_j)=\i.$

Given the spectral functions we can define the following RHP (see Figure \ref{Fig_F}):

\begin{RH}\label{RHFhat}
To find a $2\times2$ matrix-valued function $\widehat{\mathbb{F}}(x,t;\lambda),$ which
\begin{enumerate}
 \item is meromorphic in $\lambda\in\mathbb{C}\setminus\Sigma_{\mathbb{F}},$ with finite number of poles at $\lambda_j,$ $\ol{\lambda_j},$ $j=i,\ldots, J.$ Here $\Sigma_{\mathbb{F}}=\gamma_0\cup\gamma_3^{\mathbb{F}} \cup \gamma_{-3}^{\mathbb{F}} \cup \rho;$
\item has the following jump $\widehat{\mathbb{F}}_+=\widehat{\mathbb{F}}_-J_{\mathbb{F}}$ across $\Sigma:$
$$J_{\mathbb{F}}=\begin{cases}
                  \begin{pmatrix}
                   \frac{-\i r_u}{1-\i r_d} & \frac{-\i}{(1+\i r_u)(1-\i r_d)} \\ \frac{-\i}{a_ua_d} & \frac{\i r_d}{1+\i r_u}
                  \end{pmatrix},\ \lambda\in\rho,
\qquad
\begin{pmatrix}
 1 & 0 \\ \frac{-\i}{a_ua_d} & 1
\end{pmatrix}, \ \lambda\in\gamma_0,
\\\\
\qquad
\begin{pmatrix}
 1 & \frac{\i}{1+\i r_u} \\ 0 & 1
\end{pmatrix}, \ \lambda\in\gamma_3^{\mathbb{F}},
\qquad 
\begin{pmatrix}
 1 & \frac{\i}{1-\i r_d} \\ 0 & 1
\end{pmatrix}, \ \lambda\in\gamma_{-3}^{\mathbb{F}},
\end{cases},
$$
\item \label{pole_cond} has the following pole conditions at the points $\lambda_j,$ $\ol{\lambda_j},$ $j=1,\ldots,J:$
$$\widehat{\mathbb{F}}_{[2]}(\lambda)+\frac{\i}{1+\i r_u(\lambda)}\widehat{\mathbb{F}}_{[1]}(\lambda)=\mathcal{O}(1)\quad \mbox{ and }\quad  \widehat{\mathbb{F}}_{[1]}(\lambda)=\mathcal{O}(1) \quad \mbox{ for }\quad \lambda\to\lambda_j,$$
$$\widehat{\mathbb{F}}_{[2]}(\lambda)-\frac{\i}{1-\i r_d(\lambda)}\widehat{\mathbb{F}}_{[1]}(\lambda)=\mathcal{O}(1) \quad \mbox{ and }\quad \widehat{\mathbb{F}}_{[1]}(\lambda)=\mathcal{O}(1) \quad \mbox{ for }\quad\lambda\to\ol{\lambda_j},$$
\item has the following asymptotics as $\lambda\to\infty,$ which are uniform w.r.t. $\arg\lambda\in[-\pi, \pi]:$
  $$\widehat{\mathbb{F}}(x,t;\lambda)=\(I+{\ol{_\mathcal{O}}}(1)\)\dfrac{\lambda^{-\sigma_3/4}}{\sqrt{2}}
  \begin{pmatrix}
   1&1\\1&-1
  \end{pmatrix}\e^{\theta\sigma_3}.$$
\end{enumerate}
\end{RH}

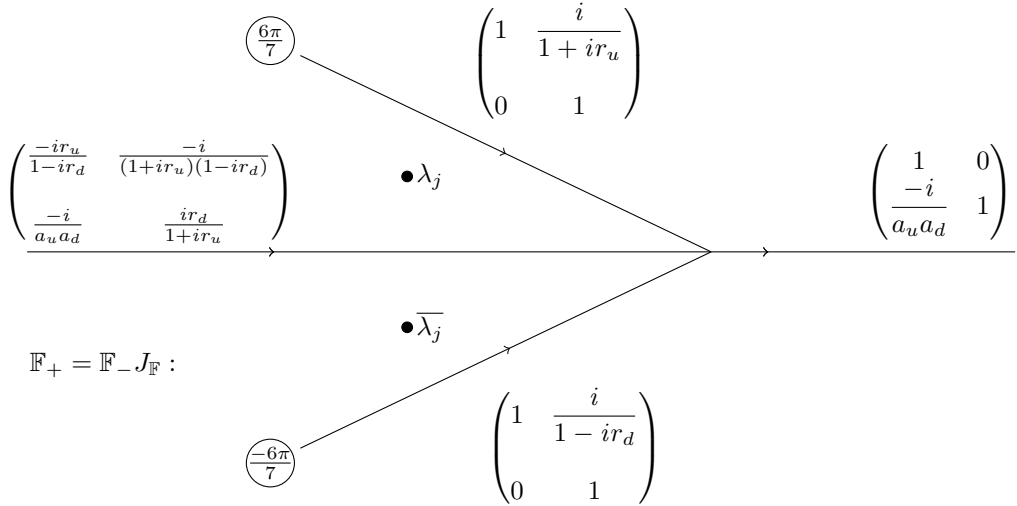
\begin{figure}[ht]
\center
\begin{tikzpicture}
\node at (-8,-1.5) {$\mathbb{F}_+=\mathbb{F}_-J_{\mathbb{F}}:$};

\draw 
[decoration={markings, mark=at position  0.5 with {\arrow{<}}}, postaction={decorate}]
(0,0) -- (-5.40581, -2.6033);
\draw [decoration={markings, mark=at position  0.5 with {\arrow{>}}}, postaction={decorate}] 
(-5.40581, 2.6033) -- (0,0);

\draw[decoration={markings, mark=at position  0.75 with {\arrow{>}}}, postaction={decorate}]
[decoration={markings, mark=at position  0.25 with {\arrow{>}}}, postaction={decorate}](-9,0) -- (4,0);

\node at (-5.8,2.8) {$\frac{6\pi}{7}$};
\draw (-5.8, 2.8) circle [radius=9pt];
\node at (-5.8,-2.8) {$\frac{-6\pi}{7}$};
\draw (-5.8,-2.8) circle [radius=9pt];

\node at (3,.8){$\begin{pmatrix}1 & 0 \\ \dsfrac{-\i }{a_ua_d} & 1\end{pmatrix}$};
\node at (-1.8,-2.6){$\begin{pmatrix}1 & \dsfrac{\i }{1-\i r_d}\\\\ 0 & 1\end{pmatrix}$};
\node at (-2.,2.5){$\begin{pmatrix}1 & \dsfrac{\i}{1+\i r_u}\\\\ 0 & 1\end{pmatrix}$};
\node at (-7.4,0.8){$\begin{pmatrix}\frac{-\i r_u }{1-\i r_d} & \frac{-\i}{(1+\i r_u)(1-\i r_d)} \\\\ \frac{-\i}{a_ua_d} & \frac{\i r_d}{1+\i r_u}\end{pmatrix}$};

\filldraw (-4,1) circle [radius=2pt];
\node at (-3.7,1) {$\lambda_j$};
\filldraw (-4,-1) circle [radius=2pt];
\node at (-3.7,-1) {$\ol{\lambda_j}$};

\end{tikzpicture}

\caption{RHP \ref{RHFhat} for the function $\mathbb{F}(x,t;\lambda).$}
\label{Fig_F}
\end{figure}

\begin{remark}
Let us notice that the set of data of the RHP \ref{RHFhat} is the function $r_u(\lambda),$ which is an entire function. Thus, it contains also all the information about the points $\lambda_j,$ $\Im\lambda_j>-0,$ where $r_u(\lambda_j)=\i.$ 

\end{remark}

In the theorem below we assume more regularity on the initial function $u_{t_0}(x),$ namely, we will take it from $C^{\infty}(\mathbb{R}).$
\begin{theorem}\label{teor_main2}
Let $t_0,$ $A<B,$ $u_{t_0}(x)$ be a $C^{\infty}(\mathbb{R})$ function, which coincides with $U(x,t_0)$ for $x\in(-\infty,A)\cup(B,+\infty).$ Then RHP \ref{RHFhat} has a unique solution $\widehat{\mathbb{F}}(x,t;\lambda),$ and $\widehat{\mathbb{F}}(x,t;\lambda)$ satisfies \clu{a} stronger asymptotic condition
  $$\widehat{\mathbb{F}}(x,t;\lambda)=\(I+\mathcal{O}(\frac{1}{\lambda})\)\dfrac{\lambda^{-\sigma_3/4}}{\sqrt{2}}
  \begin{pmatrix}
   1&1\\1&-1
  \end{pmatrix}\e^{\theta\sigma_3},
\lambda\to\infty,\quad\mbox{ uniformly w.r.t. }\arg\lambda\in[-\pi, \pi].
$$
\clu{Furthermore}, there exists \clu{a} limit when $\lambda\to\infty$ along non-transverse directions with the \clu{rays} $\gamma_{3}, \gamma_{-3}, \rho$
$$\mathbb{A}(x,t):=\lim\limits_{\lambda\to\infty}\(\widehat{\mathbb{F}}(x,t;\lambda)\e^{-\theta(x,t;\lambda)\sigma_3}\frac{1}{\sqrt{2}}(\sigma_3+\sigma_1)\lambda^{\sigma_3/4}-I\),$$
and the function \begin{equation}\label{u_F}u(x,t):=2\mathbb{A}_{11}(x,t)-\mathbb{A}_{12}^2(x,t)\end{equation}
satisfies the following properties:
\begin{itemize}\item $u(x,t)$ is $C^{\infty}$ smooth in $x,t$ for all $x,t\in\mathbb{R};$ 
\item $u(x,t)$ satisfies KdV equation \eqref{KdV} for all $t\in\mathbb{R};$
\item $u(x, t_0)=u_{t_0}(x).$ 
\end{itemize}
\end{theorem}

\noindent The following theorem is in some sense inverse to Theorems \ref{teor_main}, \ref{teor_main2}.
\begin{theorem}\label{teor_main3}
Let a function $r_u(.)$ be any function, that satisfies the following set of conditions:
\begin{enumerate}
\item $r_u(\lambda)$ is an entire function in $\lambda\in\mathbb{C};$ define $r_d(\lambda)=\ol{r_u(\ol\lambda)};$
\item properties \ref{prop_Imru}--\ref{prop_rurdneqi} of the part II of Theorem \ref{teor_main} are satisfied;
\item define functions $a_u,$ $a_d$ by formulas \eqref{au_ru}, \eqref{ad_ru}. Then property \ref{prop_ra} of the part II of Theorem \ref{teor_main} is satisfied;
\item there exist real $t_0,$ $A<B$ such that properties \ref{prop_ru_vanishing}$^{(\infty)}$, \ref{prop_auad} of the part II of Theorem \ref{teor_main} are satisfied.
\end{enumerate}
Then RHP \ref{RHFhat} has a unique solution $\widehat{\mathbb{F}}(x,t;\lambda),$ and $\widehat{\mathbb{F}}(x,t;\lambda)$ satisfies \clu{a} stronger asymptotic condition
  $$\widehat{\mathbb{F}}(x,t;\lambda)=(I+\mathcal{O}(\frac{1}{\lambda}))\dfrac{\lambda^{-\sigma_3/4}}{\sqrt{2}}
  \begin{pmatrix}
   1&1\\1&-1
  \end{pmatrix}\e^{\theta\sigma_3},
\lambda\to\infty,\mbox{ uniformly w.r.t. }\arg\lambda\in[-\pi, \pi].
$$
\clu{Furthermore}, there exist limit when $\lambda\to\infty$ along non-transverse directions with the segments $\gamma_{3}, \gamma_{-3}, \rho$
$$\mathbb{A}(x,t):=\lim\limits_{\lambda\to\infty}\(\widehat{\mathbb{F}}(x,t;\lambda)\e^{-\theta(x,t;\lambda)\sigma_3}\frac{1}{\sqrt{2}}(\sigma_3+\sigma_1)\lambda^{\sigma_3/4}-I\),$$
and the function \begin{equation}\label{u_F}u(x,t):=2\mathbb{A}_{11}(x,t)-\mathbb{A}_{12}^2(x,t)\end{equation}
satisfies the following properties:
\begin{itemize}\item $u(x,t)$ is $C^{\infty}$ smooth in $x,t$ for all $x,t \in\mathbb{R};$
\item $u(x,t)$ satisfies KdV equation \eqref{KdV} for all $x,t\in\mathbb{R};$
\item $u(x, t_0)=U(x, t_0)$ for any $x\leq A$ and any $x\geq B.$
\end{itemize}
\clu{Furthermore}, functions $r_u, r_d, a_u, a_d$ are the spectral functions associated with $u(x, t_0).$
\end{theorem}

The structure of the paper is as follows. In Section \ref{sect_U} we study the properties of the special pole-free joint solution  $U(x,t),$ and associated with it \textit{Jost} solutions of the associated Schr\"{o}dinger equation.
In Section \ref{section_Jost} we study the properties of perturbed Jost solutions, associated with a perturbation of $U(x,t_0)$ at a given time $t_0.$ In Section \ref{sect_abr} we introduce the spectral functions $a, b, r,$ associated with a perturbation of $U(x,t_0),$ and study their properties.
The statements of \clu{Theorem 1.3, i.e. the properties of the spectral functions, follow from the material given in Sections 3, 4, namely Lemmas 3.1, 3.2, 3.3, 3.5, 3.6, 4.1, 4.2, 4.4, 4.7, 4.11.}
In Sections \ref{sect_RH}, \ref{sect_RHP_exist} we solve an \textit{inverse scattering problem}, by constructing an appropriate Riemann-Hilbert problem (Section \ref{sect_RH}), and proving its solvability (Section \ref{sect_RHP_exist}). The latter also proves the solvability of the initial value problem for the KdV and give a way to link the solution of KdV with a solution of a corresponding RHP. 
\clu{Theorems 1.4 follows from Theorems 6.3, 6.8, and Theorem 1.5 follows from Theorems 6.3, 6.8, 6.9}.

Finally, in Appendix we prove a spectral decomposition of the Schr\"{o}dinger operator with the potential $U(x,t)$. In other words we show that the corresponding Jost solutions are orthogonal to each other. Let us notice that for bounded potentials (finite-gap or vanishing) a spectral decomposition (which is orthogonality of the exponents, or Fourier inversion formula) plays a central role in constructing of the integral representation for the Jost solutions, and studying their properties. In our work, since our potential is a compactly supported perturbation of $U(x, t_0),$ we managed to study the properties of the Jost solutions \textit{by hands}, without referring to a spectral theorem. However, \clu{should} one try to construct Jost solutions for more general class of perturbations, she or he or they might need to use that theorem.

The properties at $x\to\pm\infty$ or $t\to\infty$ of the contrsucted solution $u(x,t)$ of the KdV equation will be studied somewhere else.


\noindent Schematically, we can represent the contents of Theorems \ref{teor_main}, \ref{teor_main2}, \ref{teor_main3} by the following diagrams:
\[\begin{split}&\mbox{Theorem \ref{teor_main}}: \quad u_{t_0}(x)\mapsto \left\{a(\lambda), b(\lambda), r(\lambda)\right\}.
\\
&\mbox{Theorem \ref{teor_main2}}: \quad u_{t_0}(x)\mapsto \left\{a(\lambda), b(\lambda), r(\lambda)\right\}\ \ \, \mapsto 
\mbox{RHP} \ref{RHFhat} \ \mbox{ for } \  \widehat{\mathbb{F}}(x,t;\lambda)\mapsto u(x,t).
\\
&\mbox{Theorem \ref{teor_main3}}: \qquad \qquad \quad r(\lambda)\mapsto\left\{a(\lambda), b(\lambda)\right\}\mapsto 
\mbox{RHP} \ref{RHFhat}\ \mbox{ for } \ \widehat{\mathbb{F}}(x,t;\lambda)\mapsto u(x,t)\mapsto u_{t_0}(x).
\end{split}
\]

\section{Special pole-free joint solution $U(x,t)$ of the KdV equation and the $P_I^2$ equation.}\label{sect_U}
Here we list some further properties of the function $U(x,t)$ and the corresponding RHP \ref{RHE}. 
Along with the RHP \ref{RHE} we consider another RHP for a function $\widehat{\mathbb E}(x,t;\lambda),$ with the same analyticity and jump conditions as in RHP \ref{RHE}, but with the condition \clu{at} $\lambda\to\infty$ altered.
\begin{RH}\label{RHEhat}
 Find a $2\times2$ matrix-valued function $\mathbb{E}(x,t;\lambda)$, which satisfies analyticity and jump conditions of RHP \ref{RHE}, and has the following asymptotics: \begin{enumerate}
\setcounter{enumi}{2} \item[3a.] asymptotics as $\lambda\to\infty$ $$\widehat{\mathbb{E}}(\lambda)=(I+\mathcal{O}(\lambda^{-1}))
\frac{\lambda^{-\sigma_3/4}}{\sqrt{2}}\begin{pmatrix}
                                      1&1\\1&-1
                                     \end{pmatrix}
\e^{\theta \sigma_3}.
$$

\end{enumerate}
\end{RH}
\begin{remark}Both conditions at infinity fix uniquely the solution of the RHP. This is obvious for the condition 3a. in RHP \ref{RHEhat}. For the condition 3. in RHP \ref{RHE}, we first conclude that the ratio of any 2 solutions $\widetilde{\mathbb{E}}$ and $\mathbb{E}$  is constant in $\lambda$ and lower triangular:
$$\widetilde{\mathbb{E}}\cdot \mathbb{E}^{-1}=\begin{pmatrix}1 & 0 \\ c(x,t) & 1\end{pmatrix},$$ and then multiplying the above ratio by $\mathbb{E}$ from the right and checking the 
asymptotics for $\widetilde{\mathbb{E}}$ at $\lambda\to\infty,$
we see that this triangular matrix should be the identity.
\end{remark}

There is a simple relation between the solutions of the RHPs \ref{RHE} and \ref{RHEhat}, and each admits a full asymptotic expansion at infinity of the following form:
\begin{figure}[ht]
\center
\begin{tikzpicture}
\node at (-9.2,1){$\mathbb{E}=\begin{pmatrix}1&0\\b_1&1\end{pmatrix}\widehat{\mathbb{E}},$};
\node at (-9,0) {$\mathbb{E}_+=\mathbb{E}_-J_{\mathbb{E}}:$};

\draw 
[decoration={markings, mark=at position  0.5 with {\arrow{<}}}, postaction={decorate}]
(0,0) -- (-5.40581, -2.6033);
\draw [decoration={markings, mark=at position  0.5 with {\arrow{>}}}, postaction={decorate}] 
(-5.40581, 2.6033) -- (0,0);
\draw[decoration={markings, mark=at position  0.75 with {\arrow{>}}}, postaction={decorate}]
[decoration={markings, mark=at position  0.25 with {\arrow{>}}}, postaction={decorate}](-6,0) -- (4,0);

\node at (1,1) {\color{blue}$(E_{lu},E_r)$};
\node at (1,-1) {\color{blue}$(E_{ld},E_r)$};
\node at (-3.5,0.5) {\color{blue}$(E_{lu}, -\i E_{ld}=E_r-\i E_{lu})$};
\node at (-3.5,-0.5) {\color{blue}$(E_{ld}, \i E_{lu}=E_{r}+\i E_{ld})$};
\node at (-5.8,2.8) {$\frac{6\pi}{7}$};
\draw (-5.8, 2.8) circle [radius=9pt];
\node at (-5.8,-2.8) {$\frac{-6\pi}{7}$};
\draw (-5.8,-2.8) circle [radius=9pt];

\node at (3,0.5){$\begin{pmatrix}1 & 0 \\ -\i & 1\end{pmatrix}$};
\node at (-3.2,-2.4){$\begin{pmatrix}1 & \i \\ 0 & 1\end{pmatrix}$};
\node at (-3,2.3){$\begin{pmatrix}1 & \i \\ 0 & 1\end{pmatrix}$};
\node at (-6.5,0.5){$\begin{pmatrix}0 & -\i \\ -\i & 0\end{pmatrix}$};
\end{tikzpicture}
\caption{RHP \ref{RHE} for the function $\mathbb{E}(x,t;\lambda).$ }
\label{Fig_E}
\end{figure}
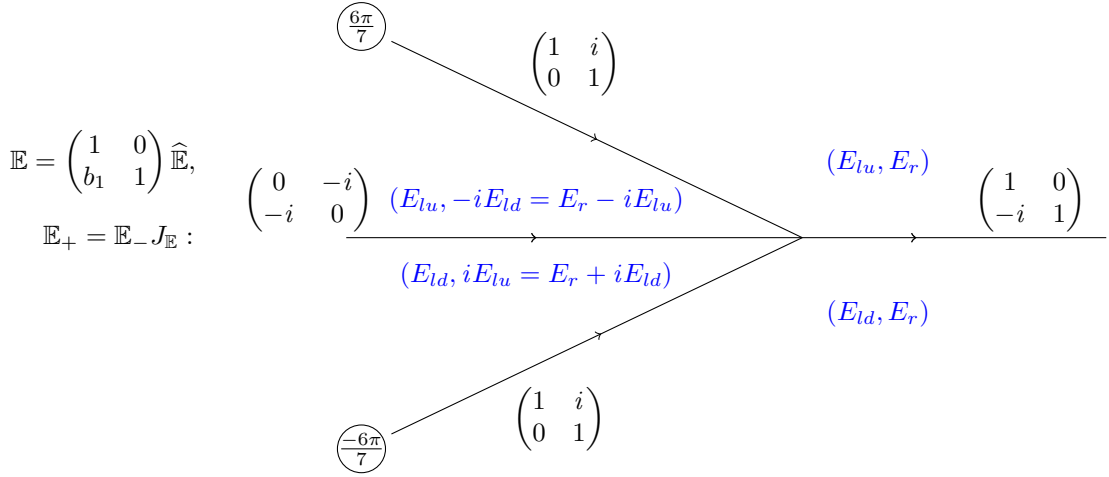

\begin{theorem}[\cite{Claeys Vanlessen}]
 For any $x,t\in\mathbb{R}$ there exists a unique solution $\widehat{\mathbb{E}}(x,t;\lambda)$ to the RH problem \ref{RHEhat}, p.\pageref{RHEhat}.
It is smooth (infinitely many times differentiable) w.r.t. $x,t,$ and 
admits the following uniform asymptotic expansion as $\lambda\to\infty:$ for any integer $J\geq 1$
\begin{equation}\label{asserEhat}\widehat{\mathbb{E}}(x,t;\lambda)=\(I+\sum\limits_{j=1}^{J-1}\begin{pmatrix}a_j(x,t)&b_j(x,t)\\c_j(x,t)&d_j(x,t)\end{pmatrix}\lambda^{-j}+\mathcal{O}(\lambda^{-J})\)\frac{\lambda^{-\sigma_3/4}}{\sqrt{2}}\begin{pmatrix}
                                      1&1\\1&-1
                                     \end{pmatrix}\e^{\theta(x,t;\lambda)\sigma_3}.\end{equation}
\end{theorem}

\begin{corollary}
The solution $\mathbb{E}$ of the RHP \ref{RHE}, p. \pageref{RHE} is related to $\widehat{\mathbb{E}}$ as follows:
\begin{equation}\label{E_Ehat_relation}\mathbb{E}(x,t;\lambda)=\begin{pmatrix}1 & 0 \\ b_1(x,t) & 1\end{pmatrix}\widehat{\mathbb{E}}(x,t;\lambda),\end{equation}
and hence admits the uniform asymptotic expansion
\begin{equation}\label{asserE}\begin{split}&\mathbb{E}(x,t;\lambda)\hskip-0.5mm=\hskip-0.5mm\frac{\lambda^{-\sigma_3/4}}{\sqrt{2}}\begin{pmatrix}1 & 1\\1 & -1 \hskip-0.5mm \end{pmatrix}
\hskip-2mm\(\hskip-0.5mmI\hskip-0.5mm+\hskip-0.5mm\frac{b_1(x,t)\sigma_3}{\sqrt{\lambda}}+
\sum\limits_{j\geq 1}\frac{a_j\hskip-0.5mm+\hskip-0.5mmd_j\hskip-0.5mm+\hskip-0.5mm
b_1b_j}{2\lambda^j}\begin{pmatrix}1 & 0\\0 & 1\end{pmatrix}
\hskip-0.5mm+\hskip-0.5mm
\sum\limits_{j\geq 1}\frac{a_j\hskip-0.5mm-\hskip-1mmd_j\hskip-1mm-\hskip-1mmb_1b_j}{2\lambda^j}\begin{pmatrix}0 & 1\\1 & 0\end{pmatrix}
\right.
\\
&
\qquad \qquad \left.
+\sum\limits_{j\geq 1}\frac{b_1a_j+c_j+b_{j+1}}{2\lambda^{j+\frac12}}\begin{pmatrix}1 & 0\\0 & -1\end{pmatrix}
+\sum\limits_{j\geq 1}\frac{b_1a_j+c_j-b_{j+1}}{2\lambda^{j+\frac12}}\begin{pmatrix}0 & 1\\-1 & 0\end{pmatrix}
\)\cdot\e^{\theta(x,t;\lambda)\sigma_3}=
\\
&
\frac{\lambda^{-\frac{\sigma_3}{4}}}{\sqrt{2}}\hskip-0.5mm\begin{pmatrix}
                                          1+\dfrac{1}{\sqrt{\lambda}}\sum\limits_{j\geq1}\dfrac{b_j}{\lambda^{j-1}}+\sum\limits_{j\geq 1}\dfrac{a_j}{\lambda^j} & 
1-\dfrac{1}{\sqrt{\lambda}}\sum\limits_{j\geq 1}\dfrac{b_j}{\lambda^{j-1}}+\sum\limits_{j\geq 1}\dfrac{a_{j}}{\lambda^{j}} 
\\ 
\\
\hskip-1mm1\hskip-1mm+\hskip-1mm\dfrac{1}{\sqrt{\lambda}}
\hskip-1mm
\(
\hskip-1mm
b_1
\hskip-1mm+
\hskip-1mm
\sum\limits_{j\geq 1}
\hskip-1mm\dfrac{b_1a_j\hskip-0.5mm+\hskip-0.5mmc_j}{\lambda^j}
\hskip-1mm
\)\hskip-1mm+
\hskip-1mm
\sum\limits_{j\geq 1}
\hskip-1mm
\dfrac{b_1b_j\hskip-0.5mm+\hskip-0.5mm
d_j}{\lambda^j} 
& 
-1\hskip-0.5mm+\hskip-1mm\dfrac{1}{\sqrt{\lambda}}\hskip-0.5mm
\(b_1\hskip-0.5mm+\hskip-1mm
\sum\limits_{j\geq 1}\hskip-1mm\dfrac{b_1a_j+c_j}{\lambda^j}
\hskip-1mm\)
\hskip-1mm-\hskip-1mm
\sum\limits_{j\geq 1}\hskip-1mm\dfrac{d_j+b_1b_j}{\lambda^j}
\end{pmatrix}\hskip-0.5mm\e^{\theta\sigma_3}.
\end{split}\end{equation}

\end{corollary}

\begin{remark}
{\textbf{Symmetry.}} Let us notice that the jump matrix $J_{\mathbb{E}}$ satisfies the symmetry
$$J_{\mathbb{E}}(\lambda) = \ol{J^{-1}_{\mathbb{E}}(\ol\lambda)},$$
which implies
\begin{equation}\label{E_sym} \mathbb{E}(x,t;\lambda)=\ol{\mathbb{E}(x,t;\ol\lambda)},
 \quad 
 \widehat{\mathbb{E}}(x,t;\lambda)=\ol{\widehat{\mathbb{E}}(x,t;\ol\lambda)},
\end{equation}
and hence all the coefficients $a_j(x,t),$ $b_j(x,t),$ $c_j(x,t),$ $d_j(x,t)$ in expansion \eqref{asserEhat}
are real.
\end{remark}

 \vskip 0.5 cm
\noindent \clu{Furthermore}, since the jumps of $\mathbb{E}$ are independent of $x,t,\lambda,$
we have  
\begin{equation}\label{x-eq_t-eq_lambda-eq}{\mathbb{E}}_x{\mathbb{E}}^{-1}=:\mathfrak{U}(x,t;\lambda),
\qquad {\mathbb{E}}_t{\mathbb{E}}^{-1}=:\mathfrak{V}(x,t;\lambda),
\qquad {\mathbb{E}}_{\lambda}{\mathbb{E}}^{-1} =\mathfrak{W}(x,t;\lambda),
\end{equation}
where $\mathfrak{U}(x,t;\lambda),$ $\mathfrak{V}(x,t;\lambda),$ $\mathfrak{W}(x,t;\lambda)$ are polynomials in $\lambda.$
Substituting asymptotic series \eqref{asserEhat}, \eqref{E_Ehat_relation} into \eqref{x-eq_t-eq_lambda-eq}, and taking into account that 
$\det\mathbb{E}=-1,$ one obtains
\begin{equation}\label{UfrakVfrak}
 \mathfrak{U}(x,t;\lambda)=:\begin{pmatrix}0 & 1 \\ \lambda - 2 U(x,t) & 0\end{pmatrix},
 \quad
\mathfrak{V}(x,t;\lambda) = \begin{pmatrix}
 \frac{U_x}{6} & \frac{-\lambda-U}{3} \\
 \frac{-\lambda^2}{3}+\frac{U\lambda}{3}+\frac{4U^2+U_{xx}}{6} & \frac{-U_x}{6}
\end{pmatrix}
\end{equation}
\begin{equation}\label{Wfrak}
 \mathfrak{W}(x,t;\lambda)=\begin{pmatrix}
  \frac{-U_x\lambda}{60}-\(\frac{1}{20}UU_x+\frac{1}{240}U_{xxx}\) & \hskip-3mm\frac{\lambda^2}{30}+\frac{U\lambda}{30}+
  \frac{-60t+6U^2+U_{xx}}{120}
  \\\\
  \frac{\lambda^3}{30}-\frac{U\lambda^2}{30}+\frac{\(-60t-2U^2-U_{xx}\)\lambda}{120}+\(x+\frac{U^3}{15}-\frac{U_x^2}{120}
  +\frac{UU_{xx}}{60}\) & \hskip-3mm\frac{U_x\lambda}{60}+\frac{12 UU_x+U_{xxx}}{240}
 \end{pmatrix}
\end{equation}
and some of the coefficients $a_j, b_j, c_j, d_j$ in \eqref{asserEhat}  are related to the function $U=U(x,t)$ in the following way:
\begin{equation}\label{U_a1b1}U(x,t)=2a_1-b_1^2=-b_{1x},\quad U_x:=2(3a_1b_1-b_1^3-b_2+c_1),\end{equation}
\begin{equation}\label{uvpq2}
x+\frac{U^3}{15}-\frac{U_x^2}{120}+\frac{UU_{xx}}{60}=\frac{-2a_1^2+a_2+4a_1b_1^2-b_1^4-2b_1b_2+b_1c_1-d_2+3b_{1t}}{3},
\end{equation}
\begin{equation}\label{b1}
 b_1=\frac{1}{480}\(240t U^2-20U^4-20U(24x+U_x^2)+U_{xx}^2-2U_xU_{xxx}\).
\end{equation}

\noindent 
Hence, asymptotic series \eqref{asserEhat}, \eqref{asserE} develop into
\begin{equation}\label{asserEhat_dev}
\widehat{\mathbb{E}}(x,t;\lambda)=\(I+\begin{pmatrix}\frac12\(U+b_1^2\)&b_1(x,t)\\c_1(x,t)&-\frac12\(U+b_1^2\)
\end{pmatrix}\lambda^{-1}+\mathcal{O}(\lambda^{-2})\)\frac{\lambda^{-\sigma_3/4}}{\sqrt{2}}\begin{pmatrix}
                                      1&1\\1&-1
                                     \end{pmatrix}\e^{\theta(x,t;\lambda)\sigma_3},\end{equation}
\begin{equation}\label{asserE_dev}\mathbb{E}(x,t;\lambda)=\frac{\lambda^{-\sigma_3/4}}{\sqrt{2}}\begin{pmatrix}1 & 1\\1 & -1\end{pmatrix}
\(I+\frac{b_1(x,t)\sigma_3}{\sqrt{\lambda}}+
\frac{1}{2\lambda}\begin{pmatrix}b_1^2 & U\\U & b_1^2\end{pmatrix}
+\mathcal{O}(\lambda^{-3/2})\)\cdot\e^{\theta(x,t;\lambda)\sigma_3},
\end{equation}
which allows one to reconstruct $U(x,t),$ once $\mathbb{E}$ or $\widehat{\mathbb{E}}$ are known.

The consistency condition of the system
$$\begin{cases}
 \mathbb{E}_x=\mathfrak{U}\mathbb{E},\\\mathbb{E}_t=\mathfrak{V}\mathbb{E},
\end{cases}
$$ i.e. $$\mathfrak{U}_t-\mathfrak{V}_x+[\mathfrak{U},\mathfrak{V}]=0,$$
gives that $U(x,t)$ satisfies the KdV equation \eqref{KdV},
and the 
consistency condition of the system
$$\begin{cases}
 \mathbb{E}_x=\mathfrak{U}\mathbb{E},\\\mathbb{E}_{\lambda}=\mathfrak{W}\mathbb{E},
\end{cases}
$$ i.e. $$\mathfrak{U}_{\lambda}-\mathfrak{W}_x+[\mathfrak{U},\mathfrak{W}]=0,$$
gives that $U(x,t)$ satisfies the P$_I^2$ equation \eqref{PI2}.

\vskip1cm 
\begin{remark}
In the notations of \cite{Grava Kapaev Klein}, $H_1=-b_1,$ $(H_1(x,t))_x=U(x,t),$
\[\mathbb{E}(\lambda) = \begin{pmatrix}1&0\\-H_1&1\end{pmatrix}\left(I+\frac{1}{\lambda}\begin{pmatrix}
                   \frac{H_1^2+U}{2}&-H_1\\-\frac{H_0}{3}\hskip-1mm+\hskip-1mm\frac{H_1^3}{3}\hskip-1mm+\hskip-0.5mm
                   H_1 U \hskip-0.5mm +\hskip-0.5mm \frac{U_x}{4}&\frac{H_1^2+U}{-2}
                  \end{pmatrix}\hskip-1mm+\hskip-1mm
                  \mathcal{O}(\lambda^{-2})\right)
\hskip-1mm\lambda^{-\frac{\sigma_3}{4}}
\hskip-1mm\begin{pmatrix}
 1&1\\1&-1
\end{pmatrix}\hskip-1mm
\frac{e^{\theta\sigma_3}}{\sqrt{2}}.
\]
\end{remark}
\vskip1cm 
\noindent It was obtained in \cite{Claeys Vanlessen} that
\begin{theorem}\label{theor_U_xlarge}\cite{Claeys Vanlessen}
 \begin{itemize}
  \item $U(x,t)$ is real-valued and pole-free for $x,t\in\mathbb{R},$
  \item for fixed $t\in\mathbb{R},$ $U(x,t)$ has the following asymptotic behavior:
  $$U(x,t)=\frac{z_0(x,t)}{2}\sqrt[3]{|x|}
  \hskip-0.5mm
  +\hskip-0.5mm
  \mathcal{O}(x^{-2})\hskip-0.5mm
  =\hskip-0.5mm
  \sqrt[3]{-6x}
  \hskip-0.5mm+\hskip-0.5mm\dfrac{2t}{\sqrt[3]{-6x}}
  \hskip-0.5mm+
  \hskip-0.5mm
  \dfrac{8}{3\hskip-0.5mm\cdot\hskip-0.5mm\(6x\)^{5/3}}
  \hskip-0.5mm
  +
  \hskip-0.5mm
  \mathcal{O}(|x|^{-2}),\  x\to\pm\infty,$$
 where $z_0(x,t)$ is the real solution of 
 $$z_0^3=-48\mathrm{sgn}(x)+24\frac{z_0 t}{|x|^{2/3}}.$$
  \end{itemize}
\end{theorem}

\subsection{Jost solutions associated with $U(x,t).$}\label{sect_Elr}
It follows from the RHP \ref{RHE} that the Jost solutions $E_l, E_r$ defined in \eqref{Eluldr_from_E} satisfy the relation
\begin{equation}\label{Eluldr_relation}E_{lu}(x,t;\lambda)=E_{ld}(x,t;\lambda)-\i E_r(x,t\lambda),\quad E_{r}(x,t;\lambda)=\i E_{lu}(x,t;\lambda)-\i E_{ld}(x,t\lambda).\end{equation}

 Large $\lambda$ behavior of $E_l,$ $E_r$ follows from formulas \eqref{asserEhat}, \eqref{E_Ehat_relation}, \eqref{asserE}, \eqref{asserEhat_dev}, \eqref{asserE_dev}.
The following lemma, which \clu{refines} Lemma \ref{lem_prel}, gives us large $x$ behavior of $E_l, E_r.$


\begin{lem}\label{lemma_E_rl}{[Properties of $E_r$, $E_l$ as $x\to\pm\infty.$]}
\begin{enumerate}\item Let $t\in\mathbb{R}, \lambda\in\mathbb{C}$ be fixed, and $x\to+\infty.$ Then 
$$(E_l, E_r)=\frac{1}{\sqrt{2}}\begin{pmatrix}1+\mathcal{O}(|x^{\frac{-7}{3}}|) & \mathcal{O}(|x^{\frac{-4}{3}}|) \\ \mathcal{O}(|x^{-1}|) & 1+\mathcal{O}(|x^{\frac{-7}{3}}|)\end{pmatrix}
(\lambda-\lambda_0)^{\frac{-\sigma_3}{4}}\begin{pmatrix}1 & 1 \\ 1 & -1\end{pmatrix}\e^{g\sigma_3}$$
$$=\frac{1}{\sqrt{2}}(\lambda-\lambda_0)^{\frac{-\sigma_3}{4}}\begin{pmatrix}1 & 1 \\ 1 & -1\end{pmatrix}\(I+\mathcal{O}(|x|^{\frac{-7}{6}})\)\e^{g\sigma_3}.$$

\item Let $t\in\mathbb{R}, \lambda\in(\mathbb{C}\setminus\mathbb{R})\cup(\mathbb{R}+\i 0)\cup(\mathbb{R}-\i 0)$ be fixed, and $x\to-\infty.$ Then 

$$(E_l, -\i\cdot \mathrm{sgn}\Im\lambda\cdot\ol{E_l(\ol{\lambda})})=
\left\{\begin{array}{ccc}(E_{lu}, -\i E_{ld}),\ \Im\lambda>0,\\(E_{ld}, \i E_{lu}), \Im\lambda<0\end{array}\right\}
$$
$$=
\frac{1}{\sqrt{2}}\begin{pmatrix}1+\mathcal{O}(|x^{\frac{-7}{3}}|) & \mathcal{O}(|x^{\frac{-4}{3}}|) \\ \mathcal{O}(|x^{-1}|) & 1+\mathcal{O}(|x^{\frac{-7}{3}}|)\end{pmatrix}
(\lambda-\lambda_0)^{\frac{-\sigma_3}{4}}\begin{pmatrix}1 & 1 \\ 1 & -1\end{pmatrix}\e^{g\sigma_3}$$
$$
=
\frac{1}{\sqrt{2}}(\lambda-\lambda_0)^{\frac{-\sigma_3}{4}}\begin{pmatrix}1 & 1 \\ 1 & -1\end{pmatrix}\(I+\mathcal{O}(|x|^{\frac{-7}{6}})\)\e^{g\sigma_3}.$$

\item $\ol{E_{l}(x,t;\ol\lambda)}=E_l(x,t;\lambda);$
$\ol{E_{r}(x,t;\ol\lambda)}=E_r(x,t;\lambda).$
\item $\det(E_l,E_r)=W\left\{e_l, e_r\right\}=-1,\quad \det(E_{lu}, E_{ld})=W\left\{e_{lu}, e_{ld}\right\}=-\i.$

\end{enumerate}
\noindent Here
\begin{equation}\label{g} g=g(x,t;\lambda) =
\frac{1}{105}\left(\lambda-\lambda_0\right)^{7/2}+\frac{\lambda_0}{30}\(\lambda-\lambda_0\)^{5/2}+
\frac{\lambda_0^2-8t}{24}\(\lambda-\lambda_0\)^{3/2},\end{equation}
where $\lambda_0=\lambda_0(x,t)$ is the solution of the equation
\begin{equation}\label{lambda0}\lambda_0^3-24t\lambda_0+48x=0,\end{equation} which is fixed for
$x\to\pm\infty$ by the condition that $\lambda_0(x,t)$ is real for
real $x,t.$

\end{lem}

\begin{remark}
It is convenient to expand $g(x,t;\lambda)$ for large $x\to\pm\infty.$ We have:
\begin{enumerate}
 \item as $x\to-\infty,$ $\lambda_0(x,t)\to+\infty,$ $g=\Re g + \i \Im g,$  where
\[\begin{split}&\Re g=-\mathrm{sgn}\,\Im\lambda\cdot\(\frac{1}{80}\Im(\lambda) \lambda_0^{5/2} +
\frac{1}{192}\Im(\lambda^2)
\lambda_0^{3/2}+\frac{1}{128}\Im(\lambda^3-64t\lambda)\lambda_0^{1/2}+\mathcal{O}(\lambda_0^{\frac{-1}{2}})\),
\\
&\Im g
\hskip-1mm
=
\hskip-1mm
 -\mathrm{sgn}\,\Im\lambda
 \hskip-1mm
 \cdot\hskip-1mm
  \(\frac{1}{56}\lambda_0^{7/2}
  \hskip-1mm
  -
  \hskip-1mm
  \frac{1}{80}\Re\lambda
  \hskip-1mm
  \cdot \hskip-1mm
  \lambda_0^{5/2}
  \hskip-1mm
   -
   \hskip-1mm
    \frac{1}{192}\Re(\lambda^2+64t) \lambda_0^{3/2}\hskip-1mm+\hskip-1mm
\frac{1}{128}\Re(64t\lambda-\lambda^3) 
\lambda_0^{1/2}
\hskip-1mm
+\hskip-1mm
\mathcal{O}(\lambda_0^{-\frac{1}{2}})\),\end{split}\]
where we set $\mathrm{sgn}\Im\lambda=\pm1$ for $\lambda\in\mathbb{R}\pm\i 0,$ 
since $g$ has discontinuity across $\lambda\in(-\infty,\lambda_0];$
\item as $x\to+\infty,$  $\lambda_0(x,t)\to-\infty$ and
$$g= {\frac{(-\lambda_0)^{\frac72}}{56}+\frac{\lambda}{80}\(-\lambda_0\)^{5/2}-\(\frac{t}{3}+
\frac{\lambda^2}{192}\)(-\lambda_0)^{\frac32}+\frac{(\lambda^3-64t\lambda)}{128}(-\lambda_0)^{\frac12}}+\mathcal{O}(\lambda_0^{\frac{-1}{2}}).$$
\end{enumerate}
Lemma \ref{lemma_E_rl} implies Lemma \ref{lem_prel} and shows  that $E_r(x,t;\lambda)$ is rapidly vanishing for any $\lambda\in\mathbb{C}$ as $x\to+\infty.$ On the other hand, 
$$E_{l}=\begin{cases}E_{lu},\Im\lambda>0,\\E_{ld},\Im\lambda<0\end{cases}$$ is rapidly vanishing 
for all $\lambda\in\mathbb{C}\setminus\mathbb{R}$ as $x\to-\infty,$ while 
for $\lambda\in\mathbb{R}\pm\i0,$ the function $E_{l}$ has oscillatory behavior of finite amplitude. 
This observation justifies referring to $E_{l},$ $E_r$ as Jost solutions, 
and the subscripts $_r,$ $_l,$
$_{lu},$ $_{ld}.$
\end{remark}

\begin{proofof}{ \textit{of Lemmas \ref{lem_Fl}, \ref{lem_Fl_exp}, \ref{lem_Fl_comp}.}}
Large $x$ behavior of $E_l,$ $E_r$ can be obtained in the same manner, as large $x$ behavior of $U(x,t),$ see
\cite{Claeys Vanlessen}, \cite{Claeys10}, \cite{Grava Kapaev Klein}). However, for convenience of the reader, we will 
sketch the derivation.

In order to study the asymptotics RHP \ref{RHE}'s solution  for fixed $\lambda$, $t$, and
$x\to\pm\infty,$ first we shift the contours 
$$\gamma_3=(\e^{\frac{6\pi\i}{7}}\infty,0), \ \gamma_{-3}=(\e^{\frac{-6\pi\i}{7}}\infty,0) \quad  \mbox{ to }\quad 
\gamma_3+\lambda_0=(\e^{\frac{6\pi\i}{7}}\infty,\lambda_0),
\ 
\gamma_{-3}+\lambda_0=(\e^{-6\frac{\pi\i}{7}}\infty, \lambda_0),$$ where $\lambda_0$ is defined in \eqref{lambda0}. We call the solution of such a shifted RHP by ${\mathbb{E}}_{\lambda_0}(x,t;\lambda).$

In other words, by $\mathbb{E}_{\lambda_0}$ we denote the solution of
the Riemann-Hilbert problem, which can be obtained from the
Riemann-Hilbert problem for $\mathbb{E}$ by moving the intersection
point of the contour $\Sigma$ from $0$ to $\lambda_0$.

\noindent Next we make a scaling change of the variables:
$$\lambda=:\zeta\sqrt[3]{|x|},\quad \lambda_0=:\zeta_0\sqrt[3]{|x|},
\qquad\zeta_0^3-24\frac{t}{\sqrt[3]{x^2}}\zeta_0+48\mathrm{sgn}(x)=0,$$
then the function $g$ defined in \eqref{g} and the function $\theta$ can be written as
\[g(\lambda) = |x|^{7/6}\(\frac{1}{105}\left(\zeta-\zeta_0\right)^{7/2}+\frac{\zeta_0}{30}\(\zeta-\zeta_0\)^{5/2}+
\frac{\zeta_0^2-\frac{8t}{|x|^{2/3}}}{24}\(\zeta-\zeta_0\)^{3/2}\),
\]
\[\theta(\lambda)=|x|^{\frac76}\(\frac{1}{105}\zeta^{\frac72}-\frac{t}{3|x|^{\frac23}}\zeta^{\frac32}+\mathrm{sgn}\,x\cdot\zeta^{\frac12}\).\]
The signature table (distribution of signs) of the function $\Im g$ is shown in Figure \ref{Fig_Im_g_lambda}.



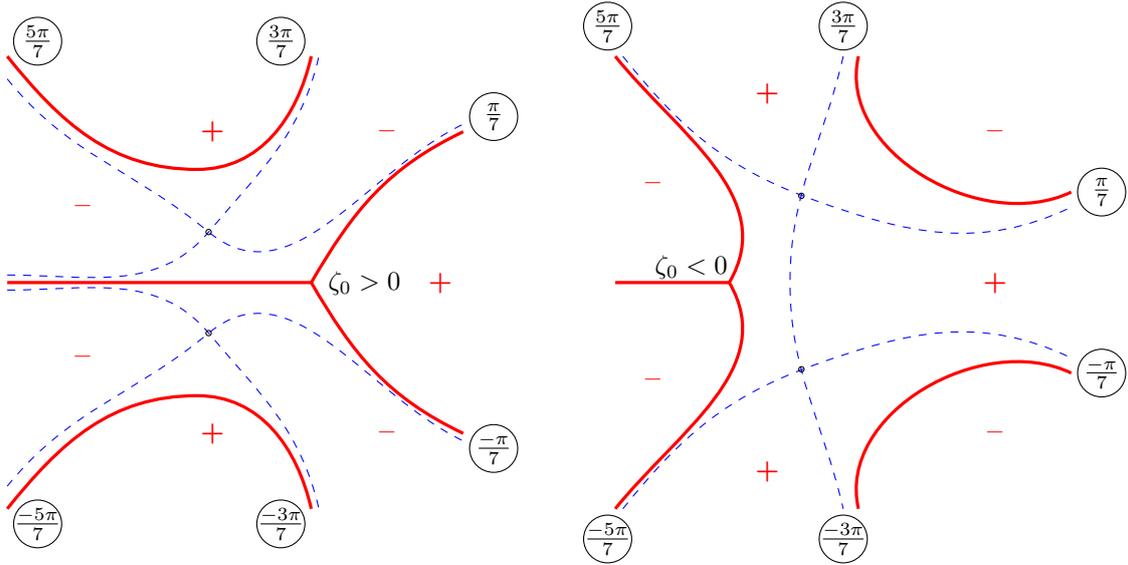
\begin{figure}[ht]
\center
 \begin{tikzpicture}
  \draw[red,very thick] (-7,0) -- (-3,0);
  \draw[red,very thick] (-3,0) to [out=60, in=205.714] (-1,2);
  \draw[red,very thick] (-7,3) to [out=308.571,in =180](-4.5,1.5) to [out=0, in=257.143](-3,3);
  \draw[blue, dashed] [out=308.571](-7,2.7) to [in=140](-4.35, 0.67) to [out=-40,in=205.714](-1,2.1);
  \draw[blue, dashed] (-7,0.1) to [out=0, in=230](-4.35, 0.67) to [out=50, in=257.143](-2.9,3);
  \draw (-4.35, 0.67) circle [radius=1pt];
  \node at (-2.3,0) {$\zeta_0>0$};  
  \node at (-1.3,0) {{\color{red}\textbf{+}}};  \node at (-4.3,2) {{\color{red}\textbf{+}}};  \node at (-6,1) {{\color{red}\textbf{--}}};  \node at (-2,2) {{\color{red}\textbf{--}}};
\node at (-0.6,2.2){$\frac{\pi}{7}$};\draw (-0.6,2.2) circle[radius=9pt];
\node at (-3.4,3.2){$\frac{3\pi}{7}$};\draw (-3.4,3.2) circle[radius=9pt];
\node at (-6.6,3.2){$\frac{5\pi}{7}$};\draw (-6.6,3.2) circle[radius=9pt];

\draw[red,very thick] (-3,0) to [out=-60, in=-205.714] (-1,-2);
  \draw[red,very thick] (-7,-3) to [out=-308.571,in =-180](-4.5,-1.5) to [out=0, in=-257.143](-3,-3);
  \draw[blue, dashed] [out=-308.571](-7,-2.7) to [in=-140](-4.35, -0.67) to [out=40,in=-205.714](-1,-2.1);
  \draw[blue, dashed] (-7,-0.1) to [out=0, in=-230](-4.35, -0.67) to [out=-50, in=-257.143](-2.9,-3);
  \draw (-4.35, -0.67) circle [radius=1pt];
  \node at (-4.3,-2) {{\color{red}\textbf{+}}};  \node at (-6,-1) {{\color{red}\textbf{--}}};  \node at (-2,-2) {{\color{red}\textbf{--}}};
\node at (-0.6,-2.2){$\frac{-\pi}{7}$};\draw (-0.6,-2.2) circle[radius=9pt];
\node at (-3.4,-3.2){$\frac{-3\pi}{7}$};\draw (-3.4,-3.2) circle[radius=9pt];
\node at (-6.6,-3.2){$\frac{-5\pi}{7}$};\draw (-6.6,-3.2) circle[radius=9pt];

  \draw[red, very thick] (1,0) -- (2.5,0);
\draw [red, very thick](2.5,0) to [out=60, in =308.571] (1,3);
\draw[blue, dashed] (1.1,3) [out=308.571, in =160 ] to (3.6,1.1) [out=-20, in=205.714] to (7,1);
\draw (3.45,1.15) circle[radius=1pt];
\draw [blue, dashed] (4,3) [out=257.143, in=90] to (3.3,0) [out=-90, in=102.857]to (4,-3);
\draw[red, very thick] (4.2,3) [out=257.143, in=205.714] to (7,1.2);
\node at (2,0.2) {$\zeta_0<0$};
\node at (6,2) {{\color{red}\textbf{--}}}; \node at (6,0) {{\color{red}\textbf{+}}};\node at (3,2.5) {{\color{red}\textbf{+}}}; \node at (1.5,1.3) {{\color{red}\textbf{--}}};
\node at (7.4,1.2){$\frac{\pi}{7}$};\draw (7.4,1.2) circle[radius=9pt];
\node at (4,3.4){$\frac{3\pi}{7}$};\draw (4,3.4) circle[radius=9pt];
\node at (0.9, 3.4){$\frac{5\pi}{7}$};\draw (0.9,3.4) circle[radius=9pt];

\draw [red, very thick](2.5,0) to [out=-60, in =-308.571] (1,-3);
\draw[blue, dashed] (1.1,-3) [out=-308.571, in =-160 ] to (3.6, -1.1) [out=20, in=-205.714] to (7,-1);
\draw (3.45,-1.15) circle[radius=1pt];
\draw[red, very thick] (4.2, -3) [out=-257.143, in=-205.714] to (7, -1.2);
\node at (6,-2) {{\color{red}\textbf{--}}}; \node at (3, -2.5) {{\color{red}\textbf{+}}}; \node at (1.5,-1.3) {{\color{red}\textbf{--}}};
\node at (7.4, -1.2){$\frac{-\pi}{7}$};\draw (7.4, -1.2) circle[radius=9pt];
\node at (4, -3.4){$\frac{-3\pi}{7}$};\draw (4, -3.4) circle[radius=9pt];
\node at (0.9, -3.4){$\frac{-5\pi}{7}$};\draw (0.9, -3.4) circle[radius=9pt];

 \end{tikzpicture}



\caption{
Distribution of signs of $\Re g(\lambda):$
for $x\to-\infty$ (on the left),
for $x\to+\infty$ (on the right).
}
Blue dashed lines are separatrices of $\Re g=const\neq0.$ Red lines correspond to $\Re g=0.$
\label{Fig_Im_g_lambda}
\end{figure}

\noindent We see that the position of $\zeta_0=\zeta_0(x,t)$ tends to fixed limits when $x\to\pm\infty.$ 

The next step is to introduce a new local variable $z(\zeta)$ in the vicinity of the point $\zeta_0,$
\[\frac{2}{3}z^{3/2}=g(\lambda),\quad |\zeta-\zeta_0|<R,\quad \mbox{ where } R>0 \quad \mbox{is sufficiently small,}\]
which allows us to define an approximation for 
$$\chi\mathbb{E}_{\lambda_0},$$ where the matrix $\chi$ is to be determined.
Namely, define
\[\mathbb{E}_{\infty}(\zeta) = \begin{cases}
                           \(\zeta-\zeta_0\)^{-\sigma_3/4}|x|^{-\sigma_3/12}\frac{\sigma_3+\sigma_1}{\sqrt{2}}e^{g(\lambda)\sigma_3},\quad & 
                          |\zeta-\zeta_0|>R,\\
               B(z)Z(z(\zeta)),\quad & |\zeta-\zeta_0|<R,
                          \end{cases}
\]
where the matrix-valued function $Z(z)$ is constructed as follows: denote
$$v_1(z)=\sqrt{2\pi}Ai(z),\quad v_2(z)=\sqrt{2\pi}e^{\pi
i/6}Ai(e^{2\pi i/3}\ z),\quad v_3(z)=\sqrt{2\pi}e^{-\pi
i/6}Ai(e^{-2\pi i/3}\ z),$$ then \begin{equation}\label{Z} Z(z)=\begin{cases}
            \begin{pmatrix}
             v_3&v_1\\v_3'&v_1'
            \end{pmatrix},\quad \arg z\in(0,2\pi i/3),\ \qquad
            \begin{pmatrix}
             v_3&-iv_2\\v_3'&-iv_2'
            \end{pmatrix},\quad \arg z\in(2\pi i/3,\pi i),\\
            \begin{pmatrix}
             v_2&v_1\\v_2'&v_1'
            \end{pmatrix},\quad \arg z\in(0,-2\pi i/3),\qquad
            \begin{pmatrix}
             v_2&iv_3\\v_2'&iv_3'
            \end{pmatrix},\quad \arg z\in(-2\pi i/3,-\pi i).
           \end{cases}
\end{equation}
The function $Z$ has the following asymptotics as $z\to\infty:$
$$Z(z)\hskip-0.5mm
=\hskip-0.5mm\begin{pmatrix}1\hskip-0.5mm+\hskip-0.5mm\mathcal{O}(\hskip-0.5mmz^{-3}\hskip-0.5mm) & \mathcal{O}(z^{-2}) \\ \mathcal{O}(z^{-1}) & 1\hskip-0.5mm +\hskip-0.5mm \mathcal{O}(\hskip-0.5mmz^{-3}\hskip-0.5mm)\end{pmatrix}
\hskip-0.5mm z^{-\sigma_3/4}\frac{\sigma_3+\sigma_1}{\sqrt{2}}e^{\frac23 z^{3/2}\sigma_3}
\hskip-0.5mm=\hskip-0.5mm
z^{-\sigma_3/4}
\hskip-0.5mm
\frac{\sigma_3\hskip-0.5mm+\hskip-0.5mm\sigma_1}{\sqrt{2}}
\hskip-0.5mm
\(
\hskip-0.5mm
I\hskip-0.5mm
+\hskip-0.5mm
\mathcal{O}(z^{-3/2})
\hskip-0.5mm
\)\hskip-0.5mm
e^{\frac23 z^{3/2}\sigma_3}
,$$ and $Z$ has the same jumps as $\mathbb{E}.$

\noindent The function $\chi$ is introduced in order to ensure identical
asymptotics at infinity of
$$\mathbb{E}_{err}(\zeta):=\chi\mathbb{E}_{\lambda_0}(\lambda(\zeta))\mathbb{E}_{\infty}^{-1}(\zeta),$$
$$\mathbb{E}_{err}\hskip-0.5mm:=\hskip-0.5mm\chi\hskip-0.5mm
\begin{pmatrix}1 & 0 \\b_1 & 1\end{pmatrix}
\hskip-0.5mm
(I\hskip-0.5mm+\hskip-0.5mm\mathcal{O}\hskip-0.5mm(\hskip-0.5mm
\lambda^{\hskip-0.5mm-1}\hskip-0.5mm
)\hskip-0.5mm)
\hskip-0.5mm
\underbrace{\lambda^{\frac{-\sigma_3}{4}}\frac{\sigma_3+\sigma_1}{\sqrt{2}}\e^{\theta\sigma_3}
\e^{-g\sigma_3}\frac{\sigma_3+\sigma_1}{\sqrt{2}}}_{=\begin{pmatrix}1 & 0\\-h_1 & 1\end{pmatrix}\cdot(I+\mathcal{O}(\lambda^{-1}))}(\lambda\hskip-0.5mm-\hskip-0.5mm\lambda_0)^
{\frac{\sigma_3}{4}}
\hskip-0.5mm=\hskip-0.5mm
\chi\hskip-0.5mm
\begin{pmatrix}
     1 & 0 \\ b_1 \hskip-0.5mm  - \hskip-0.5mm h_1 & 1
    \end{pmatrix}\hskip-0.5mm (\hskip-0.5mm I\hskip-0.5mm+\hskip-0.5mm
    \mathcal{O}\hskip-0.5mm
    (\hskip-0.5mm \lambda^{-1}\hskip-0.5mm)\hskip-0.5mm).
$$
Hence, we take 
\begin{equation}\label{chi}\chi:=\begin{pmatrix}
        1&0\\h_1-b_1 & 1
       \end{pmatrix},
\end{equation}
where $h_1=h_1(x,t)$ is determined from the expansion \clu{at} $\lambda\to\infty$ of
$
g(\lambda)-\theta(\lambda) =
\dfrac{h_1}{\sqrt{\lambda}}+\mathcal{O}(\lambda^{-3/2}),$ i.e.
$$h_1= \frac{\lambda_0^4}{128}-\frac{t\lambda_0^2}{8}=
\frac{t\lambda_0^2}{16}-\frac{3x\lambda_0}{8}=|x|^{\frac43}\(\frac{\zeta_0^4}{128}-\frac{t}{8|x|^{\frac23}}\)
=|x|\(\frac{t\lambda_0^2}{16}-\frac{3x\lambda_0}{8}\).
$$
The jump for $\mathbb{E}_{err}(\zeta)$ on $|\zeta-\zeta_0|=R$ is $\mathbb{E}_{err,+}=\mathbb{E}_{err,-}\cdot J_{\mathbb{E}_{err}},$ where
$$J_{\mathbb{E}_{err}}\hskip-1mm
=
\hskip-0.5mm
(\lambda-\lambda_0)^{\frac{-\sigma_3}{4}}\frac{\sigma_3+\sigma_1}{\sqrt{2}}\e^{g\sigma_3}Z^{-1}(z(\zeta))B^{-1}(z)
\hskip-0.5mm=\hskip-0.5mm
\(\hskip-0.5mm
\frac{\lambda\hskip-0.5mm-\hskip-0.5mm\lambda_0}{z}\hskip-0.5mm
\)^{\frac{-\sigma_3}{4}}
\hskip-0.5mm
\begin{pmatrix}
 1\hskip-0.5mm +\hskip-0.5mm \mathcal{O}(z^{-3}) & \mathcal{O}(z^{-2}) \\ \mathcal{O}(z^{-1}) & 1\hskip-0.5mm +\hskip-0.5mm \mathcal{O}(z^{-3})
\end{pmatrix}
\hskip-0.5mm
B(z)^{\hskip-0.5mm-\hskip-0.5mm 1}.
$$
Hence, in order to make this jump close to the identity matrix on the circle $|\zeta-\zeta_0|=R,$ we choose an analytic in $|\zeta-\zeta_0|<R$ matrix $B(z)$ as
$$B(z):=\(\frac{\lambda-\lambda_0}{z}\)^{\frac{-\sigma_3}{4}}.$$
Then the jump on $|\zeta-\zeta_0|=R$ satisfies 
\begin{equation}\label{JEerr}J_{\mathbb{E}_{err}} = \begin{pmatrix}
 1 + \mathcal{O}(z^{-3}) & \mathcal{O}(z^{-2}) \(\frac{z}{\lambda-\lambda_0}\)^{\frac12} \\ \mathcal{O}(z^{-1})\(\frac{\lambda-\lambda_0}{z}\)^{\frac12} & 1 + \mathcal{O}(z^{-3})
\end{pmatrix} = 
\begin{pmatrix}
 1 + \mathcal{O}(|x|^{-\frac73}) & \mathcal{O}(|x|^{-\frac43}) \\ \mathcal{O}(|x|^{-1}) & 1 + \mathcal{O}(|x|^{-\frac73})
\end{pmatrix}.
\end{equation}

\noindent The RHP for $\mathbb{E}_{err}$ is equivalent to the following singular integral equation (SIE):
\begin{equation}\label{SIE_Eerr}\mathbb{E}_{err,-}=I+\mathcal{C}_-\left[\mathbb{E}_{err,-}(J_{\mathbb{E}_{err}}-I)\right],\end{equation}
where $\mathcal{C}_{-}=\mathcal{C}_{\Sigma_{\mathbb{E}_{err}}}^{-}$ is the operator defined \clu{by} 
$$[\mathcal{C}_{\pm}f](\lambda)=\frac{1}{2\pi\i}\lim\limits_{\scriptsize{\begin{array}{ccc}\lambda'\to\lambda,\\
\lambda'\hskip-1mm\in\hskip-1mm{\pm}\mathrm{side}\end{array}}}\int\limits_{\Sigma_{\mathbb{E}_{err}}}\dfrac{f(s)\d s}{s-\lambda'}\equiv
\frac{1}{2\pi\i}\int\limits_{\Sigma_{\mathbb{E}_{err}}}\dfrac{f(s)\d s}{(s-\lambda)_{\pm}},\quad \mathcal{C}_+f-\mathcal{C}_-f=f,$$
and $$\Sigma_{\mathbb{E},err}=(\zeta_0+|x|^{-1/3}\Sigma)\cup\left\{\zeta:\ |\zeta-\zeta_0|=R\right\}\setminus\left\{\zeta: \ |\zeta-\zeta_0|<R\right\}$$ 
is the contour for $\mathbb{E}_{err}(\zeta).$ 
Once the solution of the above SIE \eqref{SIE_Eerr} is known, the solution to the RHP is given by the formula 
$$\mathbb{E}_{err}(\zeta)=I+\frac{1}{2\pi\i}\int\limits_{\Sigma_{\mathbb{E}_{err}}}\dfrac{\mathbb{E}_{err,-}(s)
(J_{\mathbb{E}_{err}}(s)-I) \d s}{s-\zeta}=:I+\(\mathcal{C}\left[\mathbb{E}_{err,-}(J_{\mathbb{E}_{err}}-I)\right]\)(s).$$
Analyzing the SIE \eqref{SIE_Eerr}, taking into account formula \eqref{JEerr} for the jump $J_{\mathbb{E}_{err}}$ on the circle $|\zeta-\zeta_0|=R,$ and also that on the other parts of the contour $\Sigma_{\mathbb{E}_{err}}$ the jump matrix $J_{\mathbb{E}_{err}}$ is exponentially close to $I,$ 
we conclude that the entries for $\mathbb{E}_{err}$ have the following asymptotics as $|x|\to\infty$, which are uniform w.r.t. $\zeta\in\mathbb{C}\cup\left\{\infty\right\}$:
$$\mathbb{E}_{err}(\zeta)=\begin{pmatrix}
 1 + \mathcal{O}(|x|^{-\frac73}) & \mathcal{O}(|x|^{-\frac43}) \\ \mathcal{O}(|x|^{-1}) & 1 + \mathcal{O}(|x|^{-\frac73})
\end{pmatrix}.
$$
Moreover, the entries in the large $\zeta$ expansion of $\mathbb{E}_{err}$ are well controlled in $x.$
Now, to obtain the large $x$ asymptotics for $\mathbb{E}_{\lambda_0},$ we recall that $\mathbb{E},$ and hence 
$\mathbb{E}_{\lambda_0},$ admit asymptotic expansion of the form
$$\mathbb{E}_{\lambda_0}=\begin{pmatrix}
              1 & 0 \\ b_1 & 1
             \end{pmatrix}\(I+\sum\limits_{j\geq 1}\begin{pmatrix}a_j & b_j\\c_j & d_j\end{pmatrix}\lambda^{-j}+\mathcal{O}(\lambda^{-\infty})\)\lambda^{\frac{-\sigma_3}{4}}\cdot\frac{\sigma_3+\sigma_1}{\sqrt{2}}\e^{\theta\sigma_3},
$$
and hence 
$$\mathbb{E}_{err}\cdot \mathbb{E}_{\infty} = \chi\mathbb{E}_{\lambda_0}=\begin{pmatrix}
              1 & 0 \\ h_1 & 1
             \end{pmatrix}\(I+\sum\limits_{j\geq 1}\begin{pmatrix}a_j & b_j\\c_j & d_j\end{pmatrix}\lambda^{-j}+\mathcal{O}(\lambda^{-\infty})\)\lambda^{\frac{-\sigma_3}{4}}\cdot\frac{\sigma_3+\sigma_1}{\sqrt{2}}\e^{\theta\sigma_3}.
$$
From here (after some computations) we first obtain
\begin{equation}\label{U_a1_b1}b_1= h_1+\mathbb{E}_{err}^{1, 12},\quad a_1=\frac{h_1^2}{2}+h_1 \mathbb{E}_{err}^{1, 12}+
\frac{\lambda_0}{4}+\mathbb{E}_{err}^{1, 11}, \quad \Big(\mbox{hence}\quad U=2a_1-b_1^2=\frac{\lambda_0}{2}+2\mathbb{E}_{err}^{1, 11}-
\(\mathbb{E}_{err}^{1, 12}\)^2\Big),\end{equation}
where we denoted
$$\mathbb{E}_{err}=I+\sum\limits_{j\geq 1}
\begin{pmatrix}
 \mathbb{E}_{err}^{j,11} & \mathbb{E}_{err}^{j,12} \\ \mathbb{E}_{err}^{j,21} & \mathbb{E}_{err}^{j,22}
\end{pmatrix}\lambda^{-j}
=
I+\sum\limits_{j\geq 1}
\begin{pmatrix}
 \mathbb{E}_{err}^{j,11} & \mathbb{E}_{err}^{j,12} \\ \mathbb{E}_{err}^{j,21} & \mathbb{E}_{err}^{j,22}
\end{pmatrix}|x|^{-j/3}\zeta^{-j},
$$
and thus,
since 
$$\mathbb{E}_{err}^{1,11} =\mathcal{O}(|x|^{-7/3}|x|^{1/3})=\mathcal{O}(|x|^{-2}),\quad  
\mathbb{E}_{err}^{1,12} =\mathcal{O}(|x|^{-4/3}|x|^{1/3})=\mathcal{O}(|x|^{-1}),$$
we obtain, from \eqref{U_a1_b1},
$$b_1-h_1=\mathbb{E}_{err}^{1, 12}=\mathcal{O}(|x|^{-1}),
\qquad \mbox{ and }\quad U=\frac{\lambda_0}{2}+\mathcal{O}(|x|^{-2}).
$$

\noindent This yields Theorem \ref{theor_U_xlarge} and secondly, due to definition \eqref{chi} of $\chi$ , 
$$\mathbb{E}_{\lambda_0}=\chi^{-1}\cdot \mathbb{E}_{err}\cdot\mathbb{E}_{\infty}=
\begin{pmatrix}
 1 & 0 \\ \mathcal{O}(|x|^{-1}) & 1
\end{pmatrix}
\cdot
\begin{pmatrix}
 1 + \mathcal{O}(|x|^{-\frac73}) & \mathcal{O}(|x|^{-\frac43}) \\ \mathcal{O}(|x|^{-1}) & 1 + \mathcal{O}(|x|^{-\frac73})
\end{pmatrix}
(\lambda-\lambda_0)^{\frac{-\sigma_3}{4}}\cdot \frac{\sigma_3+\sigma_1}{\sqrt{2}}\cdot \e^{g\sigma_3}.
$$

To obtain the first two statements of the lemma we use the definition of $E_{l},$ $E_{r}$ \eqref{Eluldr_from_E}, and track how $\mathbb{E}_{\lambda_0}$ is related to them.
As to the latter, for 
any fixed $\lambda\in\mathbb{C},$ $t\in\mathbb{R}$ and $x\to+\infty$ we have
$$\mathbb{E}_{\lambda_0}(x,t;\lambda) = (E_l(x,t;\lambda),E_r(x,t;\lambda)),\quad x\to+\infty$$
(here we incorporated both cases $\Im\lambda\gtrless 0,$ since
$E_{l}$ has cut along $\mathbb{R}$).
On the other hand, for any fixed $\lambda,$ $t\in\mathbb{R}$ and $x\to-\infty$ we have
$$\mathbb{E}_{\lambda_0}(x,t;\lambda) = (E_l(x,t;\lambda),-\i\cdot \mathrm{sgn}\Im\lambda \cdot \ol{E_{l}(x,t;\ol\lambda)}).$$

%

\noindent Hence, for any $\lambda,$ as $x\to-\infty,\ \lambda_0\to+\infty,$ we have
$$e_l(x,t;\lambda)=\frac{1}{\sqrt{2}}(\lambda-\lambda_0)^{-1/4}(1+\mathcal{O}(|x|^{-7/6}))e^{g(x,\lambda)},$$
$$e_r(x,t;\lambda)=\frac{i\ \mathrm{sign}\Im\lambda}{\sqrt{2}}(\lambda-\lambda_0)^{-1/4}(1+\mathcal{O}(|x|^{-7/6}))e^{g(x,\lambda)}
+
\frac{1}{\sqrt{2}}(\lambda-\lambda_0)^{-1/4}(1+\mathcal{O}(|x|^{-7/6}))e^{-g(x,\lambda)},$$
and for $x\to+\infty,\ \lambda_0\to-\infty,$ we have
$$e_l(x,t;\lambda)=\frac{1}{\sqrt{2}}(\lambda-\lambda_0)^{-1/4}(1+\mathcal{O}(|x|^{-7/6}))e^{g(x,\lambda)},$$
$$e_r(x,t;\lambda)=\frac{1}{\sqrt{2}}(\lambda-\lambda_0)^{-1/4}(1+\mathcal{O}(|x|^{-7/6}))e^{-g(x,\lambda)}.$$
This gives us the first and the second statements of the lemma.
The 3rd statement follows from formula \eqref{E_sym}. 
The 4th statement follows from the fact that $\det \mathbb{E}=-1$ and the definition \eqref{Eluldr_from_E}.
\end{proofof}

%

\begin{remark}\label{remark_analogy_step}
Since the properties of Jost solutions to the Sturm-Liouville equation \eqref{exx} with the potential $-2U(x,t)$ 
are different from the properties of Jost solutions associated with vanishing or bounded potentials, 
we suggest a way to develop some intuition for the properties of Jost solutions in this context. Namely,
in formula \eqref{exx}, instead of the function $U,$ take a function
$$u=\begin{cases}c_{r},\quad x>0,\\c_{l},\quad x<0.\end{cases}$$
The corresponding Jost solutions are  of the form
$$F_{r}=\begin{pmatrix}\dsfrac{\e^{-x\sqrt{\lambda-2c_r}}}{\sqrt{2}\sqrt[4]{\lambda-2c_r}},\\\\
\dsfrac{-\sqrt[4]{\lambda-2c_r}\cdot\e^{-x\sqrt{\lambda-2c_r}}}{\sqrt{2}}\end{pmatrix},
\qquad
F_{l}=\begin{pmatrix}\dsfrac{\e^{x\sqrt{\lambda-2c_r}}}{\sqrt{2}\sqrt[4]{\lambda-2c_r}},\\\\
\dsfrac{\sqrt[4]{\lambda-2c_r}\cdot\e^{x\sqrt{\lambda-2c_r}}}{\sqrt{2}}\end{pmatrix}.$$
We see that the right Jost solution $F_r,$ which is vanishing for $x\to+\infty$ and $\lambda\in\mathbb{C}\setminus(-\infty,2c_r],$ is analytic  in $\mathbb{C}\setminus(-\infty,2c_r].$ At the same time, the left Jost solution $F_l,$ which is vanishing for $x\to-\infty$ and $\lambda\in\mathbb{C}\setminus(-\infty,2c_l],$ is analytic  in $\mathbb{C}\setminus(-\infty,2c_l].$

We know that $U(x,t)\sim\sqrt[3]{-6x}, x\to\pm\infty,$ hence we can simulate this behavior of the potential $U(x,t)$ by taking $c_{l}\to+\infty$ and $c_{r}\to-\infty.$ We see that in the limit the cut for the domain of $F_r$ will shrink, and $F_r$ will become 
analytic in the whole complex plane, while the cut for $F_l$ will increase and in the limit $F_l$ will 
be analytic discontinuous across $\mathbb{R}.$

This analogy works only up to some extent. For example, the properties of the transmission and reflection coefficients are different. Indeed, the usual scattering relation
$$F_r=R(\lambda)F_l+T(\lambda)\ol{F_l(\ol{\lambda})},$$
where $R, T$ are the reflection and transsmission coefficients, respectively, becomes 
$$E_{r}=\i E_{lu}-\i E_{ld}=\i E_{lu}-\i\ \ol{E_{lu}(\ol{\lambda})},$$
and thus the transmission coefficient becomes $-\i,$ and the reflection coefficient becomes $\i.$
\end{remark}

\subsection{Analogy with Airy functions}\label{sect_Airy}
To gain some intuition, whenever it is possible we will use some similarity of the RHPs \ref{RHE}, 
\ref{RHEhat} with the following RHP, 
whose solution can be constructed explicitly in terms of Airy functions.
\begin{RH}\label{RHEAi}
 Find a $2\times2$ matrix-valued function $\mathbb{E}^{Ai}(x,t;\lambda),$ that
 \begin{enumerate}
  \item has the same analyticity and jump conditions as in RHP \ref{RHE},
  \item has the following asymptotics as $\lambda\to\infty,$ uniformly w.r.t. $\arg\lambda\in[-\pi,\pi]:$
  \begin{eqnarray*}\mathbb{E}^{Ai}(x,t;\lambda)&=&\frac{1}{\sqrt{2}}\lambda^{-\sigma_3/4}(\sigma_3+\sigma_1)
  \(I+\frac{\widetilde b^{Ai}}{\sqrt{\lambda}}+\mathcal{O}(\lambda^{-1})\)\e^{\theta^{Ai}(x,t;\lambda)\sigma_3},
  \end{eqnarray*}
where $$\theta^{Ai}=\theta^{Ai}(x,t;\lambda):=-\,\frac{t}{3}\lambda^{\frac32}+x\lambda^{\frac12},$$
and a scalar $\widetilde b^{Ai}=\widetilde b^{Ai}(x,t)$ is not fixed, but is introduced in order to fix the structure of the asymptotics.
 \end{enumerate}
\end{RH}

\noindent \textbf{The solution $U^{Ai}(x,t)$ of the KdV equation \eqref{KdV} associated with RHP \ref{RHEAi}} can be construced by formula \ref{U_RHE} (in whih we \clu{replaced} 
$\mathbb{E}$ with $\mathbb{E}^{Ai}$).	

\noindent The solution $\mathbb{E}^{Ai}(x,t;\lambda)$ of RHP \ref{RHEAi} can be construced as follows:
$$\mathbb{E}^{Ai}(x, t; \lambda):=\begin{cases}\begin{pmatrix}e_{lu}^{Ai} & -\i e_{ld}^{Ai}\\e_{lu,x}^{Ai} & -\i e_{ld,x}^{Ai}\end{pmatrix}, \arg\lambda\in\(\frac{6\pi}{7},\pi\),\qquad \begin{pmatrix}e_{lu}^{Ai} & e_{r}^{Ai}\\e_{lu,x}^{Ai} & e_{r,x}^{Ai}\end{pmatrix}, \arg\lambda\in\(0, \frac{6\pi}{7}\), \\\\
\begin{pmatrix}e_{ld}^{Ai} & \i e_{lu}^{Ai}\\e_{ld,x}^{Ai} & \i e_{lu,x}^{Ai}\end{pmatrix}, \arg\lambda\in\(-\pi, \frac{-6\pi}{7}\),\qquad \begin{pmatrix}e_{ld}^{Ai} & e_{r}^{Ai}\\e_{ld,x}^{Ai} & e_{r,x}^{Ai}\end{pmatrix}, \arg\lambda\in\(\frac{-6\pi}{7},0\),
\end{cases}$$
\noindent where

\begin{eqnarray}\nonumber e_{r}^{Ai}(x,t;\lambda) &:=&(-t)^{\frac16}2^{\frac13}\sqrt{\pi}Ai\(\ \(\frac{-t}{2}\)^{2/3}\(\lambda-\frac{2x}{t}\)\)=\\\nonumber
&=&\e^{-x\lambda^{1/2}+\frac{t}{3}\lambda^{3/2}}\frac{1}{\sqrt{2}\sqrt[4]{\lambda}}\(1+\frac{x^2}{2t\sqrt{\lambda}}+\frac{x^4+4 x t}{8t^2\lambda}+\frac{x^6+20x^3t+10t^2}{48t^3}\lambda^{-3/2}+\ldots\),
\\\nonumber
e_{r, x}^{Ai}(x,t;\lambda)&=&-\e^{-x\lambda^{1/2}+\frac{t}{3}\lambda^{3/2}}\frac{\sqrt[4]{\lambda}}{\sqrt{2}}\(1+\frac{x^2}{2t\sqrt{\lambda}}+\frac{x^4-4 x t}{8t^2\lambda}+\frac{x^6-4x^3t-14t^2}{48t^3}\lambda^{-3/2}+\ldots\),
\end{eqnarray}

\medskip
\begin{eqnarray*}e_{lu}^{Ai}(x,t;\lambda) &:=&\e^{-\pi\i/6}(-t)^{\frac16}2^{\frac13}\sqrt{\pi}Ai\(\ \e^{\frac{-2\pi\i}{3}} \(\frac{-t}{2}\)^{2/3}\(\lambda-\frac{2x}{t}\)\)=
\\
&=&\e^{x\lambda^{1/2}-\frac{t}{3}\lambda^{3/2}}\frac{1}{\sqrt{2}\sqrt[4]{\lambda}}\(1-\frac{x^2}{2t\sqrt{\lambda}}+\frac{x^4+4 x t}{8t^2\lambda}-\frac{x^6+20x^3t+10t^2}{48t^3}\lambda^{-3/2}+\ldots\),
\\
e_{lu, x}^{Ai}(x,t;\lambda) &=&\e^{x\lambda^{1/2}-\frac{t}{3}\lambda^{3/2}}\frac{\sqrt[4]{\lambda}}{\sqrt{2}}\(1-\frac{x^2}{2t\sqrt{\lambda}}+\frac{x^4-4 x t}{8t^2\lambda}-\frac{x^6-4x^3t-14t^2}{48t^3}\lambda^{-3/2}+\ldots\),
\end{eqnarray*}

\medskip
\begin{eqnarray*} e_{ld}^{Ai}(x,t;\lambda) &:=&\e^{\pi\i/6}(-t)^{\frac16}2^{\frac13}\sqrt{\pi}Ai\(\ \e^{\frac{2\pi\i}{3}} \(\frac{-t}{2}\)^{2/3}(\lambda-\frac{2x}{t})\)=
 \\
&=&\e^{x\lambda^{1/2}-\frac{t}{3}\lambda^{3/2}}\frac{1}{\sqrt{2}\sqrt[4]{\lambda}}\(1-\frac{x^2}{2t\sqrt{\lambda}}+\frac{x^4+4 x t}{8t^2\lambda}-\frac{x^6+20x^3t+10t^2}{48t^3}\lambda^{-3/2}+\ldots\),
\\
e_{ld, x}^{Ai}(x,t;\lambda) &=&\e^{x\lambda^{1/2}-\frac{t}{3}\lambda^{3/2}}\frac{\sqrt[4]{\lambda}}{\sqrt{2}}\(1-\frac{x^2}{2t\sqrt{\lambda}}+\frac{x^4-4 x t}{8t^2\lambda}-\frac{x^6-4x^3t-14t^2}{48t^3}\lambda^{-3/2}+\ldots\).
\end{eqnarray*}

\noindent 
The function $\mathbb{E}^{Ai}(x,t;\lambda)$ have the following uniform w.r.t. $\arg\lambda\in[-\pi,\pi]$ asymptotics as $\lambda\to\infty:$ 
\[\begin{split}\mathbb{E}^{Ai}&(x,t;\lambda)=\frac{1}{\sqrt{2}}\lambda^{-\sigma_3/4}(\sigma_3+\sigma_1)\(I-\frac{x^2}{2t\sqrt{\lambda}}+\mathcal{O}(\lambda^{-1})\)\e^{(-\,\frac{t}{3}\lambda^{3/2}+\lambda^{1/2}x)\sigma_3}=
\\
=&
\begin{pmatrix}1 & 0 \\\frac{-x^2}{2t} & 1\end{pmatrix}\begin{pmatrix}1+\frac{x^4+4xt}{8t^2\mu}+\frac{x^8+56tx^5+280t^2x^2}{384t^4\mu^2}+\ldots& -\,\frac{x^2}{2t\mu}-\frac{x^6+20tx^3+10t^2}{48t^3\mu^2}+\ldots
 \\\frac{x^6+8tx^3+7t^2}{24t^3\mu}+\ldots & 1-\frac{x^4+4tx}{8t^2\mu}-\frac{x^8+24tx^5+80t^2x^2}{128t^4\mu^2}+\ldots\end{pmatrix}
\cdot
\\&
\cdot \lambda^{-\sigma_3/4}\begin{pmatrix}1&1\\1&-1\end{pmatrix}
\e^{\theta^{Ai}\sigma_3}.
\end{split}\]

\noindent From here, by formulas \eqref{U_a1b1}, \eqref{asserE_dev} we get a solution to the KdV equation \eqref{KdV}: 
$$U^{Ai}(x,t)=\partial_x\frac{x^2}{2t}=\frac{x}{t},\qquad U^{Ai}_t+U^{Ai} U^{Ai}_x=U^{Ai}_t+U^{Ai} U^{Ai}_x+\frac{1}{12}\underbrace{U^{Ai}_{xxx}}_{=0}=0.$$
The function $U^{Ai}(x,t)$ as $t\to-\infty$ behaves \clu{similarly} in some sense with $U(x,t)$ as $t\to+\infty.$



\section{Jost solutions associated with a perturbation $u_{t_0}(x)$ of $U(x,t_0).$}{\label{section_Jost}}
\subsection{Left Jost solution}{\label{section_left_Jost}}
\begin{lem}\label{lem_Fl}
 Let $t_0\in\mathbb{R}$ and $u_{t_0}(x)\in L^1_{loc}\(\mathbb{R},\mathbb{R}\)$ be a locally integrable function such that 
$$\int\limits_{-\infty}^{-1}\frac{|u_{t_0}(x)-U(x,t_0)|\ \d x}{\sqrt[6]{|x|}}<\infty.$$
Then there exists a unique $2\times 1$ vector-valued function $F_{l}(x,t_0;\lambda)$ (which we call the {\rm{left Jost solution}}), 
which is differentiable in $x,$ and satisfies the \textrm{$x$-equation} (the subscript $_x$ denotes the derivative w.r.t. $x$)
\begin{equation}\label{x_eq_Fl}
 F_{l,x}=\begin{pmatrix}
      0 & 1 \\ \lambda-2u_{t_0}(x) & 0
     \end{pmatrix} F_l,\qquad F_l(x,t_0;\lambda)=:\begin{pmatrix}f_l(x,t_0;\lambda)\\f_{l,x}(x,t_0;\lambda)\end{pmatrix}
\end{equation}
such that 
\begin{enumerate}
 \item Analyticity: $F_l(x,t_0;\lambda)$ is analytic in $\lambda\in\mathbb{C}\setminus\mathbb{R}$ and continuous up to the boundary. Denote 
$$F_{l}\equiv\begin{cases}
              F_{lu}, \Im\lambda>0, \\ F_{ld}, \Im\lambda<0.
             \end{cases}
$$
\item Symmetry:
\begin{equation}\label{eq:sym}\ol{F_{lu}(x,t_0;\ol{\lambda})}=F_{ld}(x,t_0;\lambda),\qquad \mbox{i.e.}\quad \ol{F_{l}(x,t_0;\ol{\lambda})}=F_{l}(x,t_0;\lambda).\end{equation}
\item Large $x\to-\infty$ asymptotics: $$F_{l}(x,t_0;\lambda)=E_{l}(x,t_0;\lambda)(1+\mathcal{O}(\sigma_{l}(x))), \quad x\to-\infty,$$
where $$\sigma_l(x)=\int\limits_{-\infty}^{x}\frac{|u_{t_0}(y)-U(y,t_0)|\ \d y}{\sqrt[6]{|y|}}.$$
\item \label{det_prop} Determinant:
$$\det(F_{lu}, F_{ld})=W\left\{f_{lu}, f_{ld}\right\}=-\i.$$
\item Additional smoothness:
if $u_{t_0}(x)\in C^{n}(\mathbb{R}),$ then $F_l(x,t_0;\lambda)\in C^{n+2}(\mathbb{R}).$
\end{enumerate}
\end{lem}

\noindent If we strengthen the condition on the rate of convergence of $u$ to $U,$ namely if require exponential fast convergence, then we can extend $F_{lu},$ $F_{ld}$  analytically to some strips:
\begin{lem}\label{lem_Fl_exp}
 If in addition to conditions of Lemma \ref{lem_Fl}
\begin{equation}\label{eq_Fl_exp}\int\limits_{-\infty}^{-1}|u_{t_0}(y)-U(y,t_0)|\cdot|y|^{\frac{-1}{6}}\cdot\e^{C |y|^{\frac56}} \d y<\infty,\qquad C=\frac{48^{5/6}}{80}\cdot l>0,\end{equation}
then \begin{itemize}
      \item $F_{lu}$ can be extended analytically to the strip $\Im\lambda> -l,$
\item $F_{ld}$ can be extended analytically to the strip $\Im\lambda< l.$
\end{itemize}
In particular, if \eqref{eq_Fl_exp} is valid for all $C>0,$ then $F_{lu},$ $F_{ld}$ are entire functions in $\lambda$.
\end{lem}
\noindent In order to study large $\lambda$ behavior of $F_l$ we need to further strengthen the decaying conditions on $u_{t_0}(x)-U(x,t_0):$
\begin{lem}\label{lem_Fl_comp}
Suppose, in addition to conditions of Lemma \ref{lem_Fl}, that for some $A\in\mathbb{R}$ 
$$u_{t_0}(x)=U(x,t_0), \quad x<A.$$
Then, for any fixed $x\in\mathbb{R},$ uniformly w.r.t. $\arg\lambda\in[-\pi,0]\cup[0,\pi],$
$$F_{l}(x,t_0;\lambda)=E_l(x,t_0;\lambda)\(1+\mathcal{O}(\frac{1}{\sqrt{\lambda}})\),\quad \lambda\to\infty,$$
and the latter relation we understand in the sense that 
$f_l=e_l(1+\mathcal{O}(\frac{1}{\sqrt{\lambda}})), \quad f_{lx}=e_{lx}(1+\mathcal{O}(\frac{1}{\sqrt{\lambda}})).$
\end{lem}

\begin{proofof}{ \textit{of Lemmas \ref{lem_Fl}, \ref{lem_Fl_exp}, \ref{lem_Fl_comp}.}}
It is enough to prove the statements for $F_{lu},$ because of symmetry \eqref{eq:sym}.
We will look for solution of the $x-$equation \eqref{x_eq_Fl} as a solution to the integral equation (IE)
$$f_{lu}(x,t_0;\lambda)=e_{lu}(x,t_0;\lambda)+\int\limits_{-\infty}^x\i\cdot\(e_{lu}(x)e_{ld}(y)-e_{ld}(x)e_{lu}(y)\)\cdot 2(u_{t_0}(y)-U(y,t_0))\cdot f_{lu}(y,t_0;\lambda)\ \d y.$$

\noindent 
For fixed $\lambda$ and big enough negative $x,$ the function $e_{lu}(y)=e_{lu}(y,t_0;\lambda)$ does not vanish for $y<x$, so we can divide by it, getting
\begin{equation}\label{IE_Fl_El}\frac{f_{lu}(x,t_0;\lambda)}{e_{lu}(x,t_0;\lambda)}=1+\int\limits_{-\infty}^x\i\cdot\(e_{lu}(y)e_{ld}(y)-\frac{e_{ld}(x)}{e_{lu}(x)}e^2_{lu}(y)\)\cdot 2(u_{t_0}(y)-U(y,t_0))\cdot 
\frac{f_{lu}(y,t_0;\lambda)}{e_{lu}(y,t_0;\lambda)}\ \d y.\end{equation}
Lemma \ref{lemma_E_rl} yields, that for fixed $\lambda,$ $\Im\lambda>0,$ the kernel 
$$\hskip-0.5mm-
\hskip-0.5mm
\i
\hskip-0.5mm\cdot
\hskip-0.5mm
\(\hskip-1mm
e_{lu}(y)e_{ld}(y)\hskip-0.5mm-\hskip-0.5mm\frac{e_{ld}(x)}{e_{lu}(x)}e^2_{lu}(y)
\hskip-1mm
\)
\hskip-1mm
=
\hskip-1mm
\frac{1+\mathcal{O}(|y|^{-1/6})}{2\sqrt{\lambda-\lambda_0(y)}}-\frac{\(1+\mathcal{O}(|x|^{-1/6})\)\(1+\mathcal{O}(|y|^{-1/6})\)\e^{2(g(y)-g(x))}}
{2\sqrt{\lambda-\lambda_0(y)}},$$
and for $\Im\lambda<0$
$$\hskip-0.5mm\i\hskip-0.5mm\cdot\hskip-0.5mm\(\hskip-1mm
e_{lu}(y)e_{ld}(y)\hskip-0.5mm-\hskip-0.5mm\frac{e_{ld}(x)}
{e_{lu}(x)}e^2_{lu}(y)\hskip-1mm\)\hskip-0.5mm=\hskip-0.5mm\frac{1+\mathcal{O}(|y|^{-1/6})}{2\sqrt{\lambda-\lambda_0(y)}}-\frac{\(1+\mathcal{O}(|x|^{-1/6})\)\(1+\mathcal{O}(|y|^{-1/6})\)\e^{2(g(x)-g(y))}}
{2\sqrt{\lambda-\lambda_0(y)}}.$$
Hence, it is bounded by 
$$\left|\frac{C}{\sqrt{\lambda-\lambda_0(y)}}\right|\leq\frac{C}{|y|^{1/6}},\quad y\to-\infty \quad  \textrm \quad {for }\ \Im\lambda>0,$$
and by
$$\left|\frac{C\e^{2(g(x)-g(y))}}{\sqrt{\lambda-\lambda_0(y)}}\right|\leq\frac{C\exp\left\{{\frac{(1+\varepsilon)\lambda_0^{5/2}(y)|\Im\lambda|}{80}}\right\}}{|y|^{1/6}}\leq
\frac{C\exp\left\{{\frac{(1+\varepsilon)\cdot|48y|^{5/6}\cdot|\Im\lambda|}{80}}\right\}}{|y|^{1/6}}, \  y\to-\infty
\ \textrm \  {for }\ \Im\lambda<0.$$
\noindent 
Solvability of the IE \eqref{IE_Fl_El} for sufficiently large negative $x$ now follows by the successive approximation method.
Once the existence of $f_{lu}$ is established for sufficiently large negative $x,$ we can extend it to all real $x.$
The statement for the derivative $f'_{lu}$ (which is taken w.r.t. $x$) follows from the integral representation
$$
\frac{f'_{lu}(x,t_0;\lambda)}{e'_{lu}(x,t_0;\lambda)}=1+\int\limits_{-\infty}^x\i\cdot\(e_{lu}(y)e_{ld}(y)-\frac{e'_{ld}(x)}{e'_{lu}(x)}e^2_{lu}(y)\)\cdot 2(u_{t_0}(y)-U(y,t_0))\cdot 
\frac{f_{lu}(y,t_0;\lambda)}{e_{lu}(y,t_0;\lambda)}\ \d y.
$$
This proves statements 1,2,3 of  Lemma \ref{lem_Fl} and Lemma \ref{lem_Fl_exp}.
Statement 4 of Lemma \ref{lem_Fl} follows from the fact that the determinant does not depend on $x,$ and then we obtain it by taking the 
limit  $x\to-\infty$ and using property 4 of Lemma \ref{lemma_E_rl} of $E_{lu}, E_{ld}.$

To prove Lemma \ref{lem_Fl_comp}, we 
notice that since $(u-U)(y, t_0)=0$ for $y<A,$ the integral in \eqref{IE_Fl_El} is taken over a finite interval.
When $y$ varies over a finite interval and $\lambda\to\infty,$ the functions $e_{lu}(x,t_0;\lambda),$ $e_{ld}(x,t;\lambda)$ do not vanish, and have large $\lambda$ asymptotics followed from \eqref{eq_3a_E_asymp}, \eqref{Eluldr_from_E}.
Hence, the kernel of \eqref{IE_Fl_El} for $\Im\lambda>0$ admit the estimate
$$-\i\cdot\(e_{lu}(y)e_{ld}(y)-\frac{e_{ld}(x)}{e_{lu}(x)}e^2_{lu}(y)\)=\frac{1+\mathcal{O}(\lambda^{-1/2})}{2\sqrt{\lambda}}-\frac{\(1+\mathcal{O}(\lambda^{-1/2})\)\e^{2(\theta(y)-\theta(x))}}
{2\sqrt{\lambda}},$$
where $\theta(y)-\theta(x)=(y-x)\sqrt{\lambda}$ is bounded for $y<x$ and $\Im\lambda\geq 0,$ and by successive approximation method we obtain Lemma \ref{lem_Fl_comp}.
\end{proofof}

\begin{remark}
It is not trivial to extend the result of Lemma \ref{lem_Fl_comp} beyond the case of compactly supported perturbation. This is  due to the presence of the term $\frac{1}{\sqrt{\lambda-\lambda_0(y)}},$ in which both $\lambda$ and $\lambda_0(y)$ might be large, but their difference might be small.
\end{remark}

\subsection{Right Jost solution.}
\label{section_right_Jost}

\begin{lem}\label{lem_Fr}
 Let $t_0\in\mathbb{R}$ and $u_{t_0}(x)\in L^1_{loc}\(\mathbb{R},\mathbb{R}\)$ be a locally integrable function such that 
$$\int\limits^{+\infty}_{1}\frac{|u_{t_0}(x)-U(x,t_0)|\ \d x}{\sqrt[6]{x}}<\infty.$$
Then there exists a unique $2\times 1$ vector-valued function $F_{r}(x,t_0;\lambda)$ 
(which we call the {\rm{right Jost solution}}), which is differentiable in $x,$ and satisfies the \textrm{$x$-equation}
\begin{equation}\label{x_eq_Fr}
 F_{r,x}=\begin{pmatrix}
      0 & 1 \\ \lambda-2u_{t_0}(x) & 0
     \end{pmatrix} F_r,\qquad F_r(x,t_0;\lambda)=:\begin{pmatrix}f_r(x,t_0;\lambda)\\f_{r,x}(x,t_0;\lambda)\end{pmatrix}
\end{equation}
such that 
\begin{enumerate}
 \item Analyticity: $F_r(x,t_0;\lambda)$ is analytic in the whole complex plane $\lambda\in\mathbb{C}.$ 
\item Symmetry:
$$\ol{F_{r}(x,t_0;\ol{\lambda})}=F_{r}(x,t_0;\lambda).$$
\item Large $x\to+\infty$ asymptotics: $$F_{r}(x,t_0;\lambda)=E_{r}(x,t_0;\lambda)(1+\mathcal{O}(\sigma_{r}(x))), \quad x\to+\infty,$$
where $$\sigma_r(x)=\int\limits^{+\infty}_{x}\frac{|u_{t_0}(y)-U(y,t_0)|\ \d y}{\sqrt[6]{y}}.$$
\item Additional smoothness:
if $u_{t_0}(x)\in C^{n}(x\in\mathbb{R}),$ then $F_r(x,t_0;\lambda)\in C^{n+2}(x\in\mathbb{R}).$
\end{enumerate}
\end{lem}
\noindent In order to study large $\lambda$ behavior of $F_l$ we need to strengthen more the decaying conditions on $u_{t_0}(x)-U(x,t_0):$
\begin{lem}\label{lem_Fr_comp}
\clu{Assume} in addition to conditions of Lemma \ref{lem_Fr} that for some $B\in\mathbb{R}$ 
$$u_{t_0}(x)=U(x,t_0), \quad x>B.$$
Then for any fixed $x\in\mathbb{R}$ and small enough  $\varepsilon>0$ uniformly  w.r.t.  $\arg\lambda\in[-\pi+\varepsilon, \pi-\varepsilon]$
$$F_{r}(x,t_0;\lambda)=E_r(x,t_0;\lambda)\(1+\mathcal{O}(\frac{1}{\sqrt{\lambda}})\),\quad \lambda\to\infty,$$
and the latter relation we understand in the sense that 
$f_r=e_r(1+\mathcal{O}(\frac{1}{\sqrt{\lambda}})), \quad f_{rx}=e_{rx}(1+\mathcal{O}(\frac{1}{\sqrt{\lambda}})).$
\end{lem}

\begin{proofof}{ \textit{Lemmas \ref{lem_Fr}, \ref{lem_Fr_comp}.}}
The proof is very similar to the case of the \textit{left Jost solution}, but with slight differences, which we point out.
 We will look for solution of the $x-$equation \eqref{x_eq_Fr} as a solution to the integral equation (IE)
$$f_{r}(x,t_0;\lambda)=e_{r}(x,t_0;\lambda)+\int\limits^{+\infty}_x \(e_{r}(x)e_{l}(y)-e_{l}(x)e_{r}(y)\)\cdot 2(U(y,t_0)-u_{t_0}(y))\cdot f_{r}(y,t_0;\lambda)\ \d y.$$
Here the kernel 
$\(e_{r}(x)e_{l}(y)-e_{l}(x)e_{r}(y)\)$
does not have discontinuity across the real line $\lambda\in\mathbb{R},$
since in view of \eqref{Eluldr_relation}
$$\(e_{r}(x)e_{lu}(y)-e_{lu}(x)e_{r}(y)\)=\(e_{r}(x)e_{ld}(y)-e_{ld}(x)e_{r}(y)\).$$
For a fixed $\lambda$ and big enough positive $x,$ the function $e_{r}(y)=e_{r}(y,t_0;\lambda)$ does not vanish for $y>x$, so we can divide by it.
\begin{equation}\label{IE_Fr_Er}\frac{f_{r}(x,t_0;\lambda)}{e_{r}(x,t_0;\lambda)}=1+\int\limits^{+\infty}_x \(e_{l}(y)e_{r}(y)-\frac{e_{l}(x)}{e_{r}(x)}e^2_{r}(y)\)\cdot 2(U(y,t_0)-u_{t_0}(y))\cdot 
\frac{f_{r}(y,t_0;\lambda)}{e_{r}(y,t_0;\lambda)}\ \d y,\end{equation}
Lemma \ref{lemma_E_rl} yield that for a fixed $\lambda$ the kernel equals
$$\(e_{l}(y)e_{r}(y)-\frac{e_{l}(x)}{e_{r}(x)}e^2_{r}(y)\)=\frac{1+\mathcal{O}(|y|^{-1/6})}{2\sqrt{\lambda-\lambda_0(y)}}-\frac{\(1+\mathcal{O}(|x|^{-1/6})\)\(1+\mathcal{O}(|y|^{-1/6})\)\e^{2(g(y)-g(x))}}
{2\sqrt{\lambda-\lambda_0(y)}}.$$
Hence, it is bounded by 
$$\left|\frac{C}{\sqrt{\lambda-\lambda_0(y)}}\right|\leq\frac{C}{|y|^{1/6}},\quad y\to+\infty.$$
Solvability of the IE \eqref{IE_Fr_Er} for sufficiently large positive $x$ now follows by successive approximation method.
Once existence of $f_{r}$ is established for sufficiently large positive $x,$ we can extend it for all real $x.$
The statement for the derivative $f'_{r}$ (which is taken  w.r.t.  $x$) follows from the integral representation
$$
\frac{f'_{r}(x,t_0;\lambda)}{e'_{r}(x,t_0;\lambda)}=1+\int\limits^{+\infty}_x\(e_{l}(y)e_{r}(y)-\frac{e'_{l}(x)}{e'_{r}(x)}e^2_{r}(y)\)\cdot 2(U(y,t_0)-u_{t_0}(y))\cdot 
\frac{f_{r}(y,t_0;\lambda)}{e_{r}(y,t_0;\lambda)}\ \d y.
$$
This proves Lemma \ref{lem_Fr}.

To prove Lemma \ref{lem_Fr_comp}, we 
notice that since $(u-U)(y, t_0)=0$ for $y>B,$ the integral in \eqref{IE_Fr_Er} is taken over a finite interval.
When $y$ varies over a finite interval and $\lambda\to\infty,$ $\arg\lambda\in[-\pi+\varepsilon,\pi-\varepsilon],$ function $e_{r}(x,t_0;\lambda)$ does not vanish, and has large $\lambda$ asymptotics followed by \eqref{eq_3a_E_asymp}, \eqref{Eluldr_from_E}.
Hence, the kernel of \eqref{IE_Fr_Er} \clu{admits} an estimate
$$\(e_{l}(y)e_{r}(y)-\frac{e_{l}(x)}{e_{r}(x)}e^2_{r}(y)\)=\frac{1+\mathcal{O}(\lambda^{-1/2})}{2\sqrt{\lambda}}-\frac{\(1+\mathcal{O}(\lambda^{-1/2})\)\e^{2(\theta(x)-\theta(y))}}
{2\sqrt{\lambda}},$$
where $\theta(x)-\theta(y)=(x-y)\sqrt{\lambda}$ is bounded for $y>x,$ and by successive \clu{approximations} method we obtain the statement of Lemma \ref{lem_Fr_comp}.
\end{proofof}

\begin{remark}
It is not trivial to extend the result of Lemma \ref{lem_Fr_comp} beyond the case of compactly supported perturbation. This is  due to the presence of the term $\frac{1}{\sqrt{\lambda-\lambda_0(y)}},$ in which both $\lambda$ and $\lambda_0(y)$ might be large, but their difference might be small.
\end{remark}

\section{Spectral functions $a(\lambda)$ and $b(\lambda).$}\label{sect_abr}
\subsection{Function $a(\lambda).$}\label{sect_a}

\begin{lem}\label{lem_a}
\textit{\textbf{Scattering relation.}} Let $t_0\in\mathbb{R},$ and $u_{t_0}(x)\in L^1_{loc}(\mathbb{R})$ such that 
\begin{equation}\label{int_uU}\int\limits_{-\infty}^{+\infty}\dfrac{|u_{t_0}(x)-U(x,t_0)|\ \d x}{1+\sqrt[6]{|x|}}<\infty.\end{equation}
By Lemmas \ref{lem_Fl}, \ref{lem_Fr} there exist Jost solutions $F_l(x,t_0;\lambda),$ $F_r(x,t_0;\lambda).$ 
Define an analytic in $\lambda\in\mathbb{C}\setminus\mathbb{R}$
function $a(\lambda)=a(\lambda;t_0)$ by the formula 
\begin{equation}\label{a_def}a(\lambda):=\det(F_r,F_l),\quad a_u(\lambda):=\det(F_r,F_{lu}),\ a_d(\lambda):=\det(F_r,F_{ld}),
\end{equation}
\begin{equation}\label{a_auad} a(\lambda)\equiv\begin{cases}a_u(\lambda),\Im \lambda>0,\\a_d(\lambda),\Im\lambda<0.\end{cases}\end{equation}
Then, for $\lambda\in\mathbb{R},$
\begin{equation}\label{scat_rel}F_{r}(x,t_0;\lambda)=\i a_d(\lambda)F_{lu}(x,t_0;\lambda)-\i a_u(\lambda)F_{ld}(x,t_0;\lambda)\quad =\quad \i a(\lambda-\i 0)F_l(\lambda+\i 0)-\i
a(\lambda+\i 0) F_l(\lambda-\i 0).\end{equation}
\end{lem}
\begin{proof}
\noindent For any $\lambda\in\mathbb{R},$ the functions $F_{r}(\lambda),$ $F_l(\lambda+\i 0)=F_{lu}(\lambda),$ $F_l(\lambda-\i 0)=F_{ld}(\lambda)$ are the solutions of the first equation in \eqref{x-eq_t-eq_lambda-eq} ({\textit{$x-$equation}}). 
Hence, there exist 4 functions $a(\lambda+i 0),$ $a(\lambda-i 0),$ $d(\lambda+i 0),$ $d(\lambda-i 0)$ such that for
$\lambda\in\mathbb{R}$
\begin{equation}\nonumber
F_r(\lambda+\i 0)=-\i a(\lambda+\i 0)F_l(\lambda-\i 0)+\i
d(\lambda+\i 0) F_l(\lambda+\i 0),
\end{equation}
\begin{equation}\nonumber
F_r(\lambda-\i 0)=\i a(\lambda-\i 0)F_l(\lambda+\i 0)-\i
d(\lambda-\i 0) F_l(\lambda-\i 0).
\end{equation}
Due to symmetry properties of Lemmas \ref{lem_Fl}, \ref{lem_Fr}, and since
$F_r(\lambda+\i 0)=F_r(\lambda-\i 0)=F_r(\lambda),$ we conclude that 
$$a(\lambda+\i 0)=d(\lambda-\i 0),\  d(\lambda+\i
0)=a(\lambda-\i 0),  \ol{a(\lambda-\i 0)}=a(\lambda+\i 0),
\ \ol{a(\lambda-\i 0)}=a(\lambda+\i 0),\quad
\lambda\in\mathbb{R}.$$
Formula \eqref{a_def} for $a$ can now be obtained from \eqref{scat_rel} using property \ref{det_prop} of Lemma \ref{lem_Fl}. Formula \eqref{a_def} extends the domain of definition of $a(\lambda)$ from $\lambda\in\mathbb{R}\pm\i 0$ to $\lambda\in\(\mathbb{C}\setminus\mathbb{R}\)\cup\(\mathbb{R}+\i 0\)\cup(\mathbb{R}-\i 0).$
\end{proof}

\begin{lem}\label{lem_a_prop} \textbf{Properties of $a(\lambda).$}
Let $t_0$ and $u_{t_0}(x)$ be as in Lemma \ref{lem_a}, i.e. \eqref{int_uU} holds.
Then $a(\lambda)=a(\lambda;t_0)$ satisfies the following properties: 
\begin{enumerate}
\item \label{prop_sym_a} Symmetry:
$\ol{a(\ol\lambda)}=a(\lambda).$ 
\item Nonvanishing: $a(\lambda)\neq 0$ for
$\lambda\in\(\mathbb{C}\setminus\mathbb{R}\)\cup(\mathbb{R}+\i 0)\cup(\mathbb{R}-\i 0).$
\item\label{prop_ident_asymp_a} If in addition $$(u-U)(x, t_0)=0\quad \mbox{ for }\quad x<A  \quad \mbox{and }\quad x>B$$ 
for some real $A<B,$ then 
\begin{equation}\label{a_asymp}a(\lambda)=1+\frac{1}{\sqrt{\lambda}}\int\limits_{A}^{B}(U(x,t_0)-u_{t_0}(x))\d x+{\ol{_\mathcal{O}}}(\frac{1}{\sqrt{\lambda}})\end{equation} as
$\lambda\to\infty,$ uniformly in $\arg\lambda\in[-\pi+\varepsilon,0]\cup[0,\pi-\varepsilon],$
for any $\varepsilon>0.$

\item For compactly supported perturbation  $u_{t_0}(x)$ of $U(x,t_0),$ the functions $a_u(\lambda),$ $a_d(\lambda)$ can be extended analytically to the whole complex plane.
\end{enumerate}
\end{lem}

\begin{remark} Later on, in Lemma \ref{lem_b_a_prop}, we will see that the asymptotics \eqref{a_asymp} are valid not only outside of a cone around the negative real axis, but uniformly in the whole complex plane.
\end{remark}
\begin{proof}
Symmetry \ref{prop_sym_a} follows from the symmetry properties of Lemmas \ref{lem_Fl}, \ref{lem_Fr} and the definition \eqref{a_def} of $a(\lambda).$

To prove that $a(\lambda)\neq0$ everywhere, we suppose \clu{that, on the contrary},
there exists $\lambda^*$ such that $a(\lambda^*)=0.$ We have two possibilities: either
$\lambda^*\in\mathbb{R}$, or $\Im\lambda^*\neq0.$ If $\Im(\lambda^*)=0,$
then by symmetry \ref{prop_sym_a} we have
$a(\lambda^*+\i 0)=a(\lambda^*-\i 0)=0$ and hence by \eqref{scat_rel}
$F_r(x,t_0;\lambda^*)=0$ for any $x$, which contradicts the
asymptotics of $F_r(x,t_0;\lambda)$ for $x\to+\infty.$

Suppose that $\lambda^*\in\mathbb{C}\setminus\mathbb{R}$, then we may assume $\Im\lambda^*>0,$ without loss of
generality. In this case 
$F_r(x,t_0;\lambda^*)=cF_l(x,t_0;\lambda^*)$ for some constant $c.$
Thus, $F_r(x,t_0;\lambda^*)$ vanishes exponentially fast for both
$x\to\pm\infty,$ therefore by the usual scheme
$$\int\limits_{-\infty}^{+\infty}|f_r|^2+\int\limits_{-\infty}^{+\infty}2u|f_r|^2=\lambda\int\limits_{-\infty}^{+\infty}|f^2_r|$$
and hence $\lambda$ must be real.

Property 3 follows from the definition \eqref{a_def} of $a(\lambda)$ and large $\lambda$ asymptotics of $F_r, F_l$ from Lemmas \ref{lem_Fl_comp}, \ref{lem_Fr_comp}. Property 4 follows from the corresponding property of $F_l$ from Lemma \ref{lem_Fl_exp}.
\end{proof}

\subsection{Function $b(\lambda).$}\label{sect_b}
\begin{lem}\label{lem_com_sup}{\textit{\textbf{Form of Jost solutions for compactly supported perturbations.}}}
Let $t_0$ be real, $u_{t_0}(x)$ be a locally integrable function, and $$u_{t_0}(x)-U(x,t_0)=0\quad \mbox {for}\quad x<A\quad \mbox{ and }\quad x>B$$
for some real $A<B.$
Let the functions
$h_1(x, t_0; \lambda)$, $h_2(x, t_0; \lambda)$ be solutions of \eqref{exx}, i.e.
$$h_{xx}+2u_{t_0}(x) h=\lambda h,\quad A<x<B,$$
and let their Wronskian
$$W(\lambda)\equiv W(\lambda; t_0)=\left\{h_1, h_2\right\}\equiv h_1h_{2x}-h_2h_{1x}$$
not be identically 0. 

\noindent Then the Jost solutions have the form
\begin{equation}\label{fl_repr}f_{l}(x,t_0,\lambda)=\begin{cases}e_l(x,t_0,\lambda),\quad & x<A,
\\
\frac{1}{W(\lambda)}\left[\left\{e_l,h_2\right\}_A h_1(x,t_0,\lambda)-\left\{e_l,h_1\right\}_A h_2(x,t_0,\lambda)\right],\quad
& A<x<B,
\\b(\lambda; t_0)e_r(x,t_0,\lambda)+a(\lambda; t_0)e_l(x,t_0,\lambda),\quad
& x>B,
\end{cases}\end{equation}
and
\begin{equation}\label{fr_repr}f_{r}(x,t_0,\lambda)=\begin{cases}\i(a_d-a_u)e_{l}(x,t_0,\lambda)+a(\lambda; t_0)e_{r}(x,t_0,\lambda),\quad
&x<A,
\\
\frac{1}{W(\lambda)}\left[\left\{e_r, h_2\right\}_Bh_1(x,t_0,\lambda)-\left\{e_r,h_1\right\}_B h_2(x,t_0,\lambda)\right],\quad
&A<x<B,
\\e_r(x,t_0,\lambda),\quad & x>B.
\end{cases}\end{equation}


\noindent Here the function $a(\lambda;t_0)\equiv a(\lambda)$ is defined in \eqref{a_def},
 and an analytic in $\mathbb{C}\setminus\mathbb{R}$ and continuous up to the boundary $\mathbb{R}$ function $b(\lambda; t_0)\equiv b(\lambda),$
\begin{equation}\label{b_bubd}b(\lambda)\equiv\begin{cases}b_u(\lambda), \Im\lambda>0,\\ 
b_d(\lambda), \Im\lambda<0,\end{cases}\end{equation}
is determined by the representation \eqref{fl_repr}. \clu{Furthermore}, functions $b_u(\lambda),$ $b_d(\lambda),$ $a_u(\lambda),$ $a_d(\lambda)$ can be extended to entire functions, satisfying
\begin{equation}\label{b_a_rel}\frac{b_u}{a_u}-\frac{b_d}{a_d}=\frac{\i(a_ua_d-1)}{a_ua_d},\quad \mbox{or equivalently}\quad a_d b_u-a_u b_d=\i(a_ua_d-1),\quad \mbox{ for }\quad\lambda\in\mathbb{C}.\end{equation}
\end{lem}

\begin{proof}
Let $u_{t_0}(x)-U(x,t_0)=0$ for $x>B$ and $x<A.$ Let
$h_1(x,t_0,\lambda)$, $h_2(x,t_0,\lambda)$ be solutions of
$$h_{xx}+2u h=\lambda h,\quad A<x<B$$
being normalized as above.
Then $h_{1,2}$ are analytic in
$\lambda\in\mathbb{C}\setminus(-\infty,0]$
and
$$f_{l}(x,t_0,\lambda)=\begin{cases}e_l(x,t_0,\lambda),\quad x<A,
\\
\alpha_1(\lambda; t_0)h_1(x,t_0,\lambda)+\alpha_2(\lambda; t_0)h_2(x,t_0,\lambda),\quad
A<x<B,
\\ b(\lambda; t_0)e_r(x,t_0,\lambda)+\beta_2(\lambda; t_0)e_l(x,t_0,\lambda),\quad
x>B,
\end{cases}$$
and
$$f_{r}(x,t_0,\lambda)=\begin{cases}\delta_1(\lambda; t_0)e_{l}(x,t_0,\lambda)+\delta_2(\lambda; t_0)e_{r}(x,t_0,\lambda),\quad
x<A,
\\
\gamma_1(\lambda; t_0)h_1(x,t_0,\lambda)+\gamma_2(\lambda; t_0)h_2(x,t_0,\lambda),\quad
A<x<B,
\\e_r(x,t_0,\lambda),\quad x>B,
\end{cases}$$
with some coefficients $\alpha_{1,2}$, $b,$ $\beta_{2}$, $\gamma_{1,2}$ and
$\delta_{1,2}$ which are determined by the condition that the Jost
solutions $f_{l,r}$ are continuously differentiable at the points
$x=A, B.$

It follows from the scattering relation \eqref{scat_rel} 
$$\(F_{r}(\lambda)=\i a_d(\lambda)F_{lu}(\lambda)-\i a_u(\lambda)F_{ld}(\lambda),\quad \det(F_{ld}, F_{lu})=\i,\quad a(\lambda)=\det(F_r,F_l), \)$$ that 
$$\beta_2(\lambda)=\delta_2(\lambda)= a(\lambda),\quad \delta_1(\lambda; t_0)=\i(a_d-a_u),\quad$$
and that 
the function $b(\lambda)$ \eqref{b_bubd} satisfies the following conjugation relation on the real axis $\lambda\in\mathbb{R}$:
\begin{equation}\nonumber a_d b_u-a_u b_d=\i(a_ua_d-1),\quad \frac{b_u}{a_u}-\frac{b_d}{a_d}=\frac{\i(a_ua_d-1)}{a_ua_d}.\end{equation}
Since for \clu{a} compactly supported perturbation $u(x, t_0)$ the Jost solutions $f_l, f_r$ are entire, the functions $a_u,$  $a_d,$ $b_u,$ $b_d$ are also entire functions, and the latter relation is valid for all complex $\lambda.$

Finally, it is straightforward to express the coefficients $\alpha_{1,2}$ $\gamma_{1,2}$ as in \eqref{fl_repr}, \eqref{fr_repr}.
\end{proof}

\begin{remark} Later in Section \ref{sect_RH} we will see that quantity $b(\lambda)$ plays as fundamental role in the formulation of a Riemann-Hilbert problem as does $a(\lambda).$
However, $a(\lambda)$ can be defined by \eqref{a_def} for any perturbation, not necessarily compactly supported (we used compact support of the perturation only to study large $\lambda$ asymptotics of $a(\lambda)$), while $b(\lambda)$ is so far defined only for compactly supported perturbations.
\end{remark}
\begin{remark}It follows from \eqref{b_a_rel} that the function $$\mathcal{E}(\lambda)=\frac{b}{a}-\underbrace{\frac{1}{2\pi}\int\limits_{-\infty}^{+\infty}\frac{1-\frac{1}{a_u(s)a_d(s)}}{s-\lambda}d s}_{
\mathcal{O}(\frac{1}{\sqrt{\lambda}})
}
$$ is an entire function in the whole complex plane.
It satisfies the symmetry condition $$\ol{\mathcal{E}\(\,\ol{\lambda}\,\)}=\mathcal{E}(\lambda),$$
and, if $u_{t_0}(x)=c$ for $A<x<B,$ then it has the uniform w.r.t. $\arg\lambda$ asymptotics as $\lambda\to\infty$
\begin{equation}\label{E_function}\hskip0cm \mathcal{E}\hskip-0.5mm(\hskip-0.5mm\lambda\hskip-0.5mm)\hskip-0.5mm=\hskip-1mm
\e^{2\hskip-0.2mm\theta\hskip-0.2mm(\hskip-0.2mmB\hskip-0.2mm)\hskip-0.2mm}\hskip-1mm\left(
\hskip-2mm\frac{c\hskip-0.5mm-\hskip-1mm 2\hskip-0.1mm a_{\hskip-0.1mm 1}\hskip-0.5mm(\hskip-0.5mmB\hskip-0.5mm)\hskip-0.5mm
\hskip-0.5mm+\hskip-0.5mm b_{\hskip-0.2mm 1}^{\hskip-0.2mm 2}(\hskip-0.5mmB\hskip-0.5mm)}{2\lambda}\hskip-0.5mm
+
\hskip-0.5mm\mathcal{O}\hskip-0.2mm(\hskip-0.2mm\lambda^{\hskip-0.2mm-\hskip-0.2mm 3 \hskip-0.2mm / \hskip-0.2mm2}\hskip-0.2mm)
-\hskip-8mm\underbrace{\e^{2(\hskip-0.5 mm A\hskip-0.5mm-\hskip-0.5mmB\hskip-0.5mm)\hskip-0.5mm\sqrt{\hskip-0.5mm\lambda}}
\hskip-0.5mm\left[\hskip-0.5mm\frac{c\hskip-0.5mm-\hskip-0.5mm2\hskip-0.1mm a_{1}\hskip-0.5mm(\hskip-0.5mm A \hskip-0.5mm)\hskip-0.5mm+\hskip-0.5mm 
b_1^2(A)}{2\lambda}
\hskip-0.5mm+\hskip-0.5mm
\mathcal{O}\hskip-0.5mm(\hskip-0.5mm\lambda^{\hskip-0.5mm\frac{-3}{2}}\hskip-0.5mm)\hskip-0.5mm\right]}_
{\hskip-10mm\textrm{exp small term which becomes oscillatory for }\arg\lambda=\pm\pi\hskip10mm}\hskip-9mm\)\hskip-1mm+\hskip-1mm
\mathcal{O}\hskip-0.5mm\(\hskip-0.5mm\frac{\ln|\lambda|}{\lambda}\hskip-0.5mm\right).\end{equation}
The existence of an entire function with such uniform asymptotics at infinity is quite remarkable.
\end{remark}

\begin{lem}\label{lem_b_a_prop}
Let $u_{t_0}(x)$ be as in Lemma \ref{lem_com_sup}.
Denote $r(\lambda):=\frac{b(\lambda)}{a(\lambda)},$
\begin{equation}\label{r_rurd}r(\lambda)=\begin{cases}r_u(\lambda), \Im\lambda>0,\\r_d(\lambda), \Im\lambda<0.\end{cases}\end{equation}
Then
\begin{enumerate}
\item \label{prlem_symbr} Symmetry: $\ol{b(\ol{\lambda})}=b(\lambda),$ 
$\ol{r(\ol{\lambda})}=r(\lambda),$ i.e. $\ol{b_d(\ol{\lambda})}=b_u(\lambda),$ $\ol{r_d(\ol{\lambda})}=r_u(\lambda).$
\item \label{prlem_a1} 
As $\lambda\to\infty,$ uniformly  w.r.t. $\arg\lambda\in[-\pi,0]\cup[0,\pi],$
$$a(\lambda)=1+\frac{1}{\sqrt{\lambda}}\int\limits_{A}^{B}(U(x,t_0)-u_{t_0}(x))\d x+{\ol{_\mathcal{O}}}(\frac{1}{\sqrt{\lambda}}).$$
\item \label{prlem_r0} $r(\lambda)={\ol{_\mathcal{O}}}(\frac{1}{\sqrt{\lambda}})\e^{2\theta(B,t_0;\lambda)}.$
\item \label{prlem_adu0}$a_d(\lambda)-a_u(\lambda)={\ol{_\mathcal{O}}}(\frac{1}{\sqrt{\lambda}})\e^{-2\theta(A,t_0;\lambda)}.$
\item \label{prlem_ri} Property \ref{prop_Imru12} of the part II of Theorem \ref{teor_main} is satisfied, and hence 
$r_u(s)\neq \i$ and $r_d(s)\neq -\i$  for $s\in\mathbb{R};$
\\the roots of the equation $r_u(\lambda)=\i$ in the upper halfplane $\Im\lambda\geq 0$ can accumulate only along the rays $\arg\lambda=\frac{5\pi}{7}, \frac{3\pi}{7}, \frac{\pi}{7};$
\\the roots of the equation $r_d(\lambda)=-\i$ in the lower halfplane $\Im\lambda\leq 0$ are symmetric to the roots of $r_u=\i,$ and can accumulate only along the rays $\arg\lambda=\frac{-5\pi}{7}, \frac{-3\pi}{7}, \frac{-\pi}{7}.$

\item \label{prlem_Imr0} For $s\in\mathbb{R},$ $\Im r_u(s)={\ol{_\mathcal{O}}}(\frac{1}{\sqrt{s}})$ as $s\to\pm\infty.$
\item \label{prlem_rudi} Property  \ref{prop_rurdneqi} of the part II of Theorem \ref{teor_main} is satisfied.


\item \label{prlem_ainr} The functions $a_u,$ $a_d$ can be expressed in terms of $r_u$ by formulas \eqref{au_ru}, \eqref{ad_ru}.
\end{enumerate}
\end{lem}


\begin{remark}
 By Picard's theorem, the set of values of a non-constant entire function is either the whole complex 
 plane, or the complex plane minus a single point. Property \ref{prlem_rudi} hence says that the set of values of 
 $r_u(\lambda)-\ol{r_u(\ol{\lambda})}$ is either a constant, or $\mathbb{C}\setminus\left\{\i\right\}.$ \\If it is a constant, it is 0, since for $\lambda\in\mathbb{R}$ we have $r_u(\lambda) -\ol{r_u(\ol{\lambda})}=2\Im r_u(\lambda),$ which tends to 0 as $\lambda\to\pm\infty$ by property \ref{prlem_Imr0}. \clu{Furthermore}, if $r_u(\lambda)=r_d(\lambda)$ for all $\lambda\in\mathbb{C},$ then by formula \eqref{b_a_rel} we have $a_u(\lambda)a_d(\lambda)=a_u(\lambda)\ol{a_u(\ol{\lambda})}\equiv 1.$
 Hence, the function $$f(\lambda)=\begin{cases}a_u(\lambda),\ \Im\lambda>0,\\a_d^{-1}(\lambda),\ \Im\lambda<0
                                  \end{cases}$$
 satisfies the jump condition $\frac{f_+(\lambda)}{f_-(\lambda)}=a_ua_d=1$ for 
 $\lambda\in\mathbb{R},$ and hence $a_u(\lambda)\equiv a_d(\lambda)\equiv 1.$
 Hence, $b_u(\lambda)=b_d(\lambda).$
\end{remark}

\begin{remark}
 Lemma \ref{lem_a_prop} ensures that $a_u(\lambda)\neq 0$ for $\Im\lambda\geq 0,$ and $a_d(\lambda)\neq 0$ for $\Im\lambda\leq 0.$
Formulas \eqref{au_ru}, \eqref{ad_ru} show that $a_u(\lambda)\neq 0,$ $a_d(\lambda)\neq 0$ for all $\lambda\in\mathbb{C}.$
\end{remark}
\begin{proofof}{ \textit{of Lemma \ref{lem_b_a_prop}.}}
The symmetry property follows from the corresponding symmetry for $f_l,$ $f_r.$
\clu{Furthermore}, it follows from the representations \eqref{fl_repr}, \eqref{fr_repr} that 
$$\alpha_1=\frac{1}{W(\lambda)}W\left\{e_l, h_2\right\}|_A,\quad \alpha_2=-\frac{1}{W(\lambda)}W\left\{e_l,h_1\right\}|_A,$$
\begin{equation}\label{a_h12}a(\lambda)=\frac{1}{W(\lambda)}\left[W\left\{e_r,h_1\right\}_B\cdot W\left\{e_l,h_2\right\}_A-W\left\{e_r,h_2\right\}_B\cdot W\left\{e_l,h_1\right\}_A\right],\end{equation}
\begin{equation}\label{b_h12}b(\lambda)=\frac{1}{W(\lambda)}\left[-W\left\{e_l,h_2\right\}_A\cdot W\left\{e_l,h_1\right\}_B+W\left\{e_l,h_1\right\}_A\cdot W\left\{e_l,h_2\right\}_B\right],\end{equation}
\begin{equation}\label{adu_h12}\i(a_d-a_u)=\frac{1}{W(\lambda)}\left[W\left\{e_r,h_2\right\}_B\cdot W\left\{e_r,h_1\right\}_A-W\left\{e_r,h_1\right\}_B\cdot W\left\{e_r,h_2\right\}_A\right].\end{equation}


\begin{figure}[ht!]
 \includegraphics[width=0.8\linewidth]{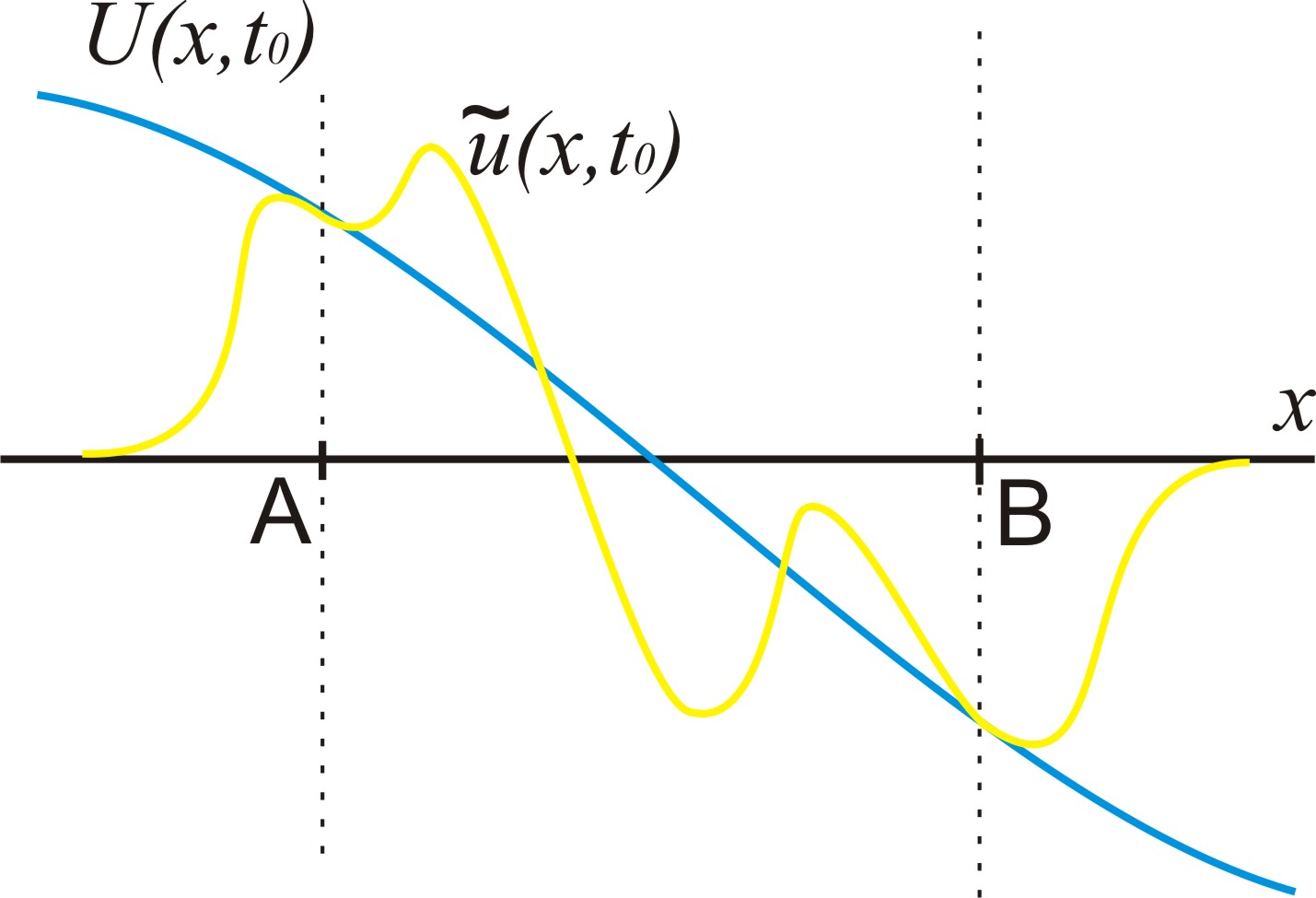}
 \caption{{{Initial function $u_{t_0}(x)$ equals $U(x,t_0)$ outside of the interval $x\in[A,B],$ and equals some function $\widetilde u_{t_0}(x)$ inside $x\in[A,B].$}}}
\label{Fig_Uutilde}
\end{figure}

We can choose $h_1, h_2$ to be normalized in such a way that they admit the
integral representations written below, 
and given enough smoothnes of $u_{t_0}(x),$ namely assuming that $u_{t_0}(x)$ is $N$ times differentiable (observe that the functions 
$R(x,y),$ $L(x,y)$ are one time more regular than function $u_{t_0}(x)$), we can develop asymptotic 
series of $h_1,$ $h_2$ for 
large $\lambda$ as follows:
\begin{enumerate}
\item 
\begin{equation}\label{h1}\hskip-2cm\begin{split}&h_1(x,t_0;\lambda)=\frac{1}{\sqrt{2}\sqrt[4]{\lambda}}\(1-\int\limits_{x}^{+\infty}R(x,y,t_0)\e^{\sqrt{\lambda}(x-y)}\d y\)\e^{-\theta(x,t_0;\lambda)}
\\
\qquad &=\frac{\e^{-\theta(x,t_0;\lambda)}}{\sqrt{2}\sqrt[4]{\lambda}}\(1-\frac{R(x,x,t_0)}{\sqrt{\lambda}}-\frac{1}{\sqrt{\lambda}}\int\limits_{x}^{+\infty}R_y(x,y,t_0)\e^{\sqrt{\lambda}(x-y)}\d y\)
\\
&=\frac{\e^{\frac{t}{3}\lambda^{3/2}-\frac{1}{105}\lambda^{7/2}}}{\sqrt{2}\sqrt[4]{\lambda}}\(\e^{-x\sqrt{\lambda}}+\sum\limits_{k=0}^N\frac{\partial_y^kR(x,y)\cdot\e^{-y\sqrt{\lambda}}}{\lambda^{\frac{k+1}{2}}}\Big|_{x}^{+\infty}-\frac{1}{\lambda^{\frac{N+1}{2}}}\int\limits_{x}^{+\infty}\partial_y^{N+1}R(x,y)\cdot\e^{-y\sqrt{\lambda}}\d y\),
\end{split}
\end{equation}
\item 
\begin{equation}\label{h1x}\hskip-2cm\begin{split}&h_{1x}(x,t_0;\lambda)=\frac{-\sqrt[4]{\lambda}}{\sqrt{2}}\(1-\frac{R(x,x)}{\sqrt{\lambda}}+\frac{1}{\sqrt{\lambda}}\int\limits_{x}^{+\infty}R_x(x,y,t_0)\e^{\sqrt{\lambda}(x-y)}\d y\)\e^{-\theta(x,t_0;\lambda)}
\\
&=\frac{-\sqrt[4]{\lambda}\e^{\frac{t}{3}\lambda^{3/2}-\frac{1}{105}\lambda^{7/2}}}{\sqrt{2}}\(\e^{-x\sqrt{\lambda}}-\frac{R(x,x)\e^{-x\sqrt{\lambda}}}{\sqrt{\lambda}}-
\sum\limits_{k=1}^N\frac{\partial_x\partial_y^{k-1}R(x,y)\cdot\e^{-y\sqrt{\lambda}}}{\lambda^{\frac{k+1}{2}}}\Big|_{x}^{+\infty}\right.
\\
&\qquad\qquad\qquad\qquad
\left.+\frac{1}{\lambda^{\frac{N+1}{2}}}\int\limits_{x}^{+\infty}\partial_x\partial_y^{N}R(x,y)\cdot\e^{-y\sqrt{\lambda}}\d y\),\end{split}\end{equation}
\item 
\begin{equation}\label{h2}\hskip-2cm\begin{split}&h_2(x,t_0;\lambda)=\frac{1}{\sqrt{2}\sqrt[4]{\lambda}}\(1+\int\limits^{x}_{-\infty}L(x,y,t_0)\e^{\sqrt{\lambda}(y-x)}\d y\)\e^{\theta(x,t_0;\lambda)}\\
\qquad &=\frac{\e^{\theta(x,t_0;\lambda)}}{\sqrt{2}\sqrt[4]{\lambda}}\(1+\frac{L(x,x,t_0)}{\sqrt{\lambda}}-\frac{1}{\sqrt{\lambda}}\int\limits^{x}_{-\infty}L_y(x,y,t_0)\e^{\sqrt{\lambda}(y-x)}\d y\)
=
\frac{\e^{-\frac{t}{3}\lambda^{3/2}+\frac{1}{105}\lambda^{7/2}}}{\sqrt{2}\sqrt[4]{\lambda}}\cdot
\\
&\cdot
\(\e^{x\sqrt{\lambda}}+\sum\limits_{k=0}^N\frac{(-1)^k\partial_y^kL(x,y)\cdot\e^{y\sqrt{\lambda}}}
{\lambda^{\frac{k+1}{2}}}\Big|^{x}_{-\infty}+\frac{(-1)^{N+1}}{\lambda^{\frac{N+1}{2}}}\int\limits^{x}_{-\infty}\partial_y^{N+1}L(x,y)\cdot\e^{y\sqrt{\lambda}}\d y\),
\end{split}\end{equation}
\item  \begin{equation}\label{h2x}\hskip-2cm\begin{split}&h_{2x}(x,t_0;\lambda)=\frac{\sqrt[4]{\lambda}}{\sqrt{2}}\(1+\frac{L(x,x)}{\sqrt{\lambda}}+\frac{1}{\sqrt{\lambda}}\int\limits^{x}_{-\infty}L_x(x,y,t_0)\e^{\sqrt{\lambda}(y-x)}\d y\)\e^{\theta(x,t_0;\lambda)},
\\
&=\frac{\sqrt[4]{\lambda}\e^{-\frac{t}{3}\lambda^{3/2}+\frac{1}{105}\lambda^{7/2}}}{\sqrt{2}}\(\e^{x\sqrt{\lambda}}+\frac{L(x,x)\e^{x\sqrt{\lambda}}}{\sqrt{\lambda}}+
\sum\limits_{k=1}^N\frac{(-1)^{k+1}\partial_x\partial_y^{k-1}L(x,y)\cdot\e^{y\sqrt{\lambda}}}{\lambda^{\frac{k+1}{2}}}\Big|^{x}_{-\infty}
\right.
\\
\hskip-5cm&\left.+\frac{(-1)^{N}}{\lambda^{\frac{N+1}{2}}}\int\limits^{x}_{-\infty}\partial_x\partial_y^{N}L(x,y)\cdot\e^{y\sqrt{\lambda}}\d y\),
\end{split}\end{equation}
\end{enumerate}
with kernels $R(x,y),$ $L(x,y)$ satisfying the following integral equations:
\begin{equation}\label{RL_IE}
\begin{split}
R(x,y)=\int\limits_{\frac{x+y}{2}}^{+\infty}\widetilde u(s,t_0)\ \d s-
2\int\limits_{\frac{x+y}{2}}^{+\infty}\d\alpha\int\limits_0^{\frac{y-x}{2}}\widetilde u(\alpha-\beta, t_0)R(\alpha-\beta, 
\alpha+\beta)\ \d\beta,
\\
L(x,y)=-\int\limits^{\frac{x+y}{2}}_{-\infty}\widetilde u(s,t_0)\ \d s-
2\int\limits^{\frac{x+y}{2}}_{-\infty}\d\alpha\int\limits_0^{\frac{x-y}{2}}\widetilde u(\alpha+\beta, t_0)L(\alpha+\beta, 
\alpha-\beta)\ \d\beta,
\end{split}
\end{equation}
where $\widetilde u_{t_0}(x)$ is a compactly supported function, which has the same regularity as $u_{t_0}(x),$ and coincides with $u_{t_0}(x)$ on the interval $x\in[A,B]:$
$$\widetilde u_{t_0}(x)=\begin{cases}
                             u_{t_0}(x), A\leq x\leq B,\\0, x<A-1 \ \ \mbox{ and }\ \  x>B+1.
                            \end{cases}
$$
Notice that $$R(x,y)=0 \quad \mbox{ for }\quad x+y>2B+2,\qquad\mbox{ and }\quad L(x,y)=0\quad \mbox{ for }x+y<2A-2.$$

Let us recall the formula of integration by parts for a function of bounded variation $f\in BV[a,b]\subset L^1[a,b],$ and a
differentiable function $\varphi.$ First of all, every function of bounded variation can be represented as a sum of an absolutely continuous function and a step function:
$$f(x)=f_{ac}(x)+f_{s}(x),$$
where $f_{ac}$ is the absolutely continuous function, $$f_{ac}(x)=\int\limits_{a}^xf_{ac}'(x) \d x,\qquad f'\in L^{1}[a,b],$$
and $f_s$ is the step function,
$$f_s(x)=_{ _{\hskip-3mm a.e.}}\sum\limits_{j}\alpha_j\chi_{[x_j,b]}(x),\qquad \sum\limits_j |\alpha_j|<\infty,$$
where $\left\{x_j\right\}_j$ is an at most countable set of points in $[a,b]$ and $\chi_{[x_j,b]}$ is the characteristic function of the segment $[x_j, b]$  (see \cite[chapter 2.3.5, exercise 7, p.94; chapter 7.2.3 p.227]{Kadets} for details).
Then the formula of integration by parts takes the form 
\begin{equation}\label{integrbypart}\int\limits_{a}^b f(x)\varphi'(x)\d x
=\sum\limits_j\alpha_j \varphi(x)\Big|_{x_j}^b+f_{ac}(x)\varphi(x)\Big|_{a}^b-\int\limits_{a}^{b}f_{ac}'(x)\varphi(x)\d x=
-\int\limits_{-\infty}^{+\infty}\widetilde\varphi(x)\d\(\chi_{[a,b]}(x)\widetilde f(x)\),\end{equation}
where $\widetilde f$ is an extension of $f$ from $[a,b]$ to $\mathbb{R},$ which itself is a function of bounded variation, and $\widetilde \varphi$  is an extension of $\varphi$ from $[a,b]$ to $\mathbb{R},$ which itself  is a differentiable function. \clu{Furthermore},
by $\d\(\chi_{[a,b]}(x)\widetilde f(x)\)$ we denoted the signed measure corresponding to the function of bounded variation $\chi_{[a,b]}(x)f(x).$

With this in mind, formulas \eqref{h1}-\eqref{h2x} make sense also for a function $u_{t_0}(x)$ which is $N-1$ time differentiable, with $u^{(N-1)}(x,t_0)$ locally of bounded variation. In this case we need to interchange the last two terms 
of order $\lambda^{\frac{N+1}{-2}}$ in \eqref{h1}-\eqref{h2x}  with the terms 
\begin{equation}\label{h1_konc}\frac{1}{\lambda^{\frac{N+1}{2}}}\int\limits_{-\infty}^{+\infty}\e^{-y\sqrt{\lambda}}\d_y\(\chi_{[x,+\infty)}(x)\partial_y^NR(x,y)\),
\frac{1}{\lambda^{\frac{N+1}{2}}}\int\limits_{-\infty}^{+\infty}\e^{-y\sqrt{\lambda}}\d_y\(\chi_{[x,+\infty)}(x)\partial_x\partial_y^{N-1}R(x,y)\),
\end{equation}
\begin{equation}\label{h2_konc}
\frac{1}{\lambda^{\frac{N+1}{2}}}\int\limits^{+\infty}_{-\infty}\e^{-y\sqrt{\lambda}}\d_y\(\chi_{(-\infty,x]}(x)\partial_y^NL(x,y)\),
\frac{1}{\lambda^{\frac{N+1}{2}}}\int\limits^{+\infty}_{-\infty}\e^{-y\sqrt{\lambda}}\d_y\(\chi_{(-\infty,x]}(x)\partial_x\partial_y^{N-1}L(x,y)\),
\end{equation}
respectively, which are also of order \clu{$\mathcal{O}(\lambda^{-\frac{N+1}{2}}).$} 

\noindent 
Hence, uniformly  w.r.t.  $\arg\lambda\in[-\pi,\pi],$ as $\lambda\to\infty,$
$$W(\lambda)\equiv W(t_0,\lambda):=W\left\{h_1,h_2\right\}\equiv \left\{h_1,h_2\right\}:=h_1h_{2x}-h_{1x}h_2=1+\mathcal{O}(\frac{1}{\sqrt{\lambda}}),$$
{\color{black}provided that $R, L$ are differentiable. Let us notice, that $R_x,  R_y, L_x, L_y $ have the same regulrity w.r.t. $x,y$ as $u_{t_0}(x).$}


\noindent It was shown in Lemma \ref{lem_a_prop} that asymptotics \eqref{a_asymp} are valid outside of a cone around the negative real axis. Hence, it is enough to
study the behavior of $a(\lambda)$ in the \clu{sector} around $\mathbb{R}_-.$ To this end we rewrite the expression \eqref{a_h12} for $\arg\lambda\in [\frac{5\pi}{7}+\varepsilon, \pi]$ using 
\eqref{Eluldr_relation}:
$$a_u(\lambda)=\frac{1}{W(\lambda)}\Big[\i\underbrace{\left\{e_{lu},h_2\right\}_A}_{\e^{2\theta(A)} \ol{_\mathcal{O}}(\frac{1}{\sqrt{\lambda}})}\cdot\underbrace{\left\{e_{lu},h_1\right\}_B}_{-1+\mathcal{O}(\frac{1}{\sqrt{\lambda}})}-\underbrace{\i\left\{e_{lu},h_1\right\}_A\cdot \left\{e_{lu},h_2\right\}_B}_{\e^{2\theta(B)}\ol{_\mathcal{O}}(\frac{1}{\sqrt{\lambda}})}
$$
$$\qquad\qquad\qquad-\i\underbrace{\left\{e_{lu},h_2\right\}_A}_{\e^{2\theta(A)}\ol{_\mathcal{O}}(\frac{1}{\sqrt{\lambda}})}\cdot \underbrace{\left\{e_{ld},h_1\right\}_B}_{\e^{-2\theta(B)}\ol{_\mathcal{O}}(\frac{1}{\sqrt{\lambda}})}+\underbrace{\i \left\{e_{lu},h_1\right\}_A\cdot \left\{e_{ld},h_2\right\}_B}_{1+\mathcal{O}(\frac{1}{\sqrt{\lambda}})}\Big],$$
and hence, since for $C\in\left\{A, B\right\},$
$$\left\{e_{lu},h_2\right\}_C=\e^{2\theta(C,t_0;\lambda)}\ol{_\mathcal{O}}(\frac{1}{\sqrt{\lambda}})\  \mbox{ and }
\  \left\{e_{lu},h_1\right\}_C=-1+\frac{R(C,C)-b_1(C)}{\sqrt{\lambda}}+\ol{_\mathcal{O}}(\frac{1}{\sqrt{\lambda}})\ \mbox{ for }\ \Im\lambda\geq 0,$$
and for $\arg\lambda\in[\frac{5\pi}{7}+\varepsilon, \pi]$
$$\left\{e_{ld},h_1\right\}_C=\e^{-2\theta(B,t_0;\lambda)}\ol{_\mathcal{O}}(\frac{1}{\sqrt{\lambda}})\ \mbox{ and }\
\left\{-\i e_{ld},h_2\right\}_C=1+\frac{L(C,C)-b_1(C)}{\sqrt{\lambda}}+\ol{_\mathcal{O}}(\frac{1}{\sqrt{\lambda}}) ,$$
and the function $$\e^{2\theta(A,t_0;\lambda)-2\theta(A,t_0;\lambda)}=\e^{2(A-B)\sqrt{\lambda}}$$ is bounded, 
we obtain that 
$$a_u(\lambda)W(\lambda)=1+\frac{b_1(A)-b_1(B)+L(B,B)-R(A,A)}{\sqrt{\lambda}}+\mathcal{O}(\frac{1}{\sqrt{\lambda}})$$ 
for $\arg\lambda\in[\frac{5\pi+\varepsilon}{7},\pi],$ 
and 
$$W(\lambda)=1+\frac{L(x,x)-R(x,x)}{\sqrt{\lambda}}+\ol{_\mathcal{O}}(\frac{1}{\sqrt{\lambda}})\ \mbox{ for }\arg\lambda\in[-\pi,\pi],$$
and hence, recalling property \eqref{U_a1b1} ($\partial_x b_{1}=-U$) and equations \eqref{RL_IE}, we obtain
$$a(\lambda)=1+\frac{1}{\sqrt{\lambda}}\int\limits_{A}^B(U(x,t_0)-u_{t_0}(x)\d x)+\ol{_\mathcal{O}}(\frac{1}{\sqrt{\lambda}}),$$
which proves the second statement of the lemma.
\clu{Furthermore}, from \eqref{b_h12}, the above asymptotics, and the second statement of the lemma we obtain the third statement of the lemma.

To obtain the fourth statement, for $\Im\lambda\in[-\pi+\varepsilon,0]\cup[0, \pi-\varepsilon]$ we use 
formula \eqref{adu_h12} together with the asymptotics
$$\left\{e_r, h_2\right\}_C=1+\mathcal{O}(\frac{1}{\sqrt{\lambda}}),\quad \left\{e_r, h_1\right\}_C=\mathcal{O}(\frac{1}{\sqrt{\lambda}})\e^{-2\theta(C, t_0; \lambda)},\qquad \Im\lambda\in[-\pi+\varepsilon,0]\cup[0, \pi-\varepsilon],$$
which give
$$(a_d-a_u)=\mathcal{O}(\frac{1}{\sqrt{\lambda}})\e^{-2\theta(A, t_0;\lambda)}+\mathcal{O}(\frac{1}{\sqrt{\lambda}})\e^{-2\theta(B, t_0;\lambda)}=
\mathcal{O}(\frac{1}{\sqrt{\lambda}})\e^{-2\theta(A,t_0;\lambda)}.$$
For $\arg\lambda\in[\frac{5\pi}{7}+\varepsilon, \pi]$ we substitute $e_r=\i e_{lu}-\i e_{ld}$ in formula \eqref{adu_h12}, which
gives
$$\i(a_d-a_u)=\mathcal{O}(\frac{1}{\sqrt{\lambda}})+\e^{2\theta(A)}\mathcal{O}(\frac{1}{\sqrt{\lambda}})+
\e^{-2\theta(A)}\mathcal{O}(\frac{1}{\sqrt{\lambda}})+\e^{2\theta(B)}\mathcal{O}(\frac{1}{\sqrt{\lambda}})+
\e^{-2\theta(B)}\mathcal{O}(\frac{1}{\sqrt{\lambda}})+$$
$$+\e^{2(A-B)\sqrt{\lambda}}\mathcal{O}(\frac{1}{\lambda})+\e^{2(B-A)\sqrt{\lambda}}\mathcal{O}(\frac{1}{\lambda}),
$$
and
\[\begin{split}&\i(a_d-a_u)\e^{2\theta(A, t_0; \lambda)}=\e^{2\theta(A)} \mathcal{O}(\frac{1}{\sqrt{\lambda}})+\e^{4\theta(A)}\mathcal{O}(\frac{1}{\sqrt{\lambda}})+
\mathcal{O}(\frac{1}{\sqrt{\lambda}})+\e^{2\theta(A)+2\theta(B)}\mathcal{O}(\frac{1}{\sqrt{\lambda}})+
\\
&+
\e^{2(A-B)\sqrt{\lambda}}\mathcal{O}(\frac{1}{\sqrt{\lambda}})+
\e^{2\theta(A)}\e^{2(A-B)\sqrt{\lambda}}\mathcal{O}(\frac{1}{\lambda})+\e^{2\theta(B)}\mathcal{O}(\frac{1}{\lambda})=
\mathcal{O}(\frac{1}{\sqrt{\lambda}})\quad \mbox{for }\quad \arg\lambda\in[\frac{5\pi}{7}+\varepsilon, \pi].
\end{split}\] 
The fifth statement that $r_u(\lambda)\neq \i$ for $\lambda\in\mathbb{R}$ follows from relation $\eqref{b_a_rel}$
$$r_u(\lambda)-r_d(\lambda)=\i(1-\frac{1}{a_ua_d}).$$
Indeed, by Lemma \ref{lem_a_prop}, $$a_u(\lambda)\neq 0 \mbox{ for }\Im\lambda\geq 0\quad \mbox{ and }\quad a_d(\lambda)\neq 0 \mbox{ for }\Im\lambda\leq 0.$$ Hence, assuming that $r_u(\lambda^*)=\i$ for some real $\lambda^*,$ 
and hence, by symmetry $\ol{r_d(\ol{\lambda})}=r_u(\lambda)$ also $r_d(\lambda^*)=-\i,$ one obtains 
$$2=1-\frac{1}{a_u(\lambda^*)a_d(\lambda^*)},$$ which cannot be true since  $a_u(\lambda^*)a_d(\lambda^*)=|a_u(\lambda^*)|^2 \geq 0.$
The \clu{remaining part} of property 5 follows from the asymptotics of $r(\lambda)$ described in property 3.

Properties 6, 7, 8 follow from \eqref{b_a_rel}, applied to a real $\lambda=s\in\mathbb{R}.$
Indeed, we have 
$$r_u(s)-r_d(s)=\i \hskip-0.5mm \(\hskip-1mm 1 \hskip-0.5mm - \hskip-0.5mm \frac{1}{a_u(s)a_d(s)}\)\Longrightarrow r_u(s)-\ol{r_u(s)}=\hskip-1mm \i \hskip-0.5mm \(\hskip-1mm 1 \hskip-0.5mm - \hskip-0.5mm \frac{1}{a_u(s)\ol{a_u(s)}}\)
\Longrightarrow2\Im r_u(s)=1-\frac{1}{|a_u(s)|^2}.$$
Hence, together with the asymptotics of $a(\lambda)$ from property 2, this gives us properties 6, 7, 8. To obtain property 9, we consider 
  the scalar conjugation problem for the function
$$f(\lambda)=\begin{cases}
   a_u(\lambda), \Im\lambda>0,\\
   \frac{1}{a_d(\lambda)}, \Im\lambda<0,
  \end{cases},\qquad \frac{f_+}{f_-}=|a_u(s)|^2=\frac{1}{1-2\Im r_u(s)},
$$
and then use the Sokhotsky-Plemelj formula.
This finishes the proof of  Lemma \ref{lem_b_a_prop}.
\end{proofof}

We think of $r(\lambda)\equiv \dfrac{b(\lambda)}{a(\lambda)}$ as of a reflection coefficient, and hence we would 
expect that the rate of vanishing of $r(\lambda)$ as $\lambda\to\infty$ is related to the smoothness of the 
initial function $u_{t_0}(x).$ The following lemma shows that this is indeed the case.
\begin{lem}\label{lem_r_refined}{\textbf{Refined decay of $r(\lambda).$}}
Let $t_0\in\mathbb{R}$ and $u_{t_0}(x)$ be equal to $U(x,t_0)$ outside of an interval $x\in[A,B].$ Then if 
\begin{enumerate}
\item $u_{t_0}(x)-U(x,t_0)\in BV_{loc}$ is a function of bounded variation, then
$$a(\lambda)=
1+\frac{1}{\sqrt{\lambda}}\int\limits_{A}^{B}(U(x,t_0)-u_{t_0}(x))\d x+\mathcal{O}(\frac{1}{\lambda}),
\quad a_d-a_u=\mathcal{O}(\frac{1}{\lambda})\e^{-2\theta(A)}, \quad \lambda\to\infty,
$$ 
$$r(\lambda)=\mathcal{O}\(\frac{1}{\lambda}\)\e^{2\theta(B,t_0;\lambda)},\ \lambda\to\infty,\qquad \Im r_u(s)=\mathcal{O}(\frac{1}{s}),\ s\to\pm\infty.$$
\item $u_{t_0}(x)$ is $N$ times differentiable ($N=0,1,2,\ldots$), and $u^{(N)}(x,t_0)-U^{(N)}(x,t_0)$ is a function of bounded variation, then  $$r(\lambda)=\mathcal{O}\(\lambda^{-\frac{N}{2}-1}\)\e^{2\theta(B,t_0;\lambda)},\quad \lambda\to\infty.$$
\end{enumerate}
\end{lem}

\begin{proofof}{ of Lemma \ref{lem_r_refined}.}
For the parts of function $u_{t_0}(x)$ outside of the interval $x\in[A,B],$ where it equals $U(x,t_0),$ we have Jost solutions 
$E_r, E_l,$ properties of which are described in Lemma \ref{lemma_E_rl}, Section \ref{sect_Elr}. In order to use the machinery of Jost solutions also in the interval $x\in[A,B],$ we take a compactly supported function $\widetilde u_{t_0}(x),$ which on the interval $x\in[A,B]$ coincides with $u_{t_0}(x)$ (see Figure \ref{Fig_Uutilde}).

Function $b(\lambda)$ has representation \eqref{b_h12}. Now we need to develop an asymptotic series for all the ingredients in the above formula.
For the functions $h_1,$ $h_{1x},$ $h_2,$ $h_{2x}$ we use formulas \eqref{h1}-\eqref{h2x}, and for $e_l,$ $e_{lx}$ formula \eqref{asserE}.

For $u_{t_0}(x)\in BV_{loc}$ we can take $N=1$ in formulas \eqref{h1}-\eqref{h2x} with remainder terms \eqref{h1_konc}, \eqref{h2_konc}, to obtain
$$\left\{e_l, h_1\right\}_C=-1-\dfrac{b_1(C)-R(C,C)}{\sqrt{\lambda}}+\mathcal{O}(\dfrac{1}{\lambda}),$$
with $C=A$ or $C=B.$ Let us develop asymptotics for the term $\left\{e_l, h_2\right\}_C.$
For the sake of clarity we precede the general case $N$ by the cases $N=0,1,2.$ We have 
\begin{equation}\label{elh2}\begin{split}&\left\{e_l, h_2\right\}_C\hskip-0.5mm\cdot \hskip-0.5mm 2\e^{-2C\sqrt{\lambda}+\frac{2t_0}{3}\lambda^{3/2}-\frac{2}{105}\lambda^{7/2}}\hskip-1mm=\hskip-1mm
\(
\hskip-1mm
1
\hskip-.5mm
+
\hskip-.5mm\dfrac{L}{\sqrt{\lambda}}
\hskip-.5mm
+
\hskip-.5mm
\dfrac{L_x}{\lambda}
\hskip-.5mm
-
\hskip-.5mm
\dfrac{L_{xy}}{\lambda^{3/2}}
\hskip-.5mm+
\hskip-.5mm
\mathcal{O}(\dfrac{1}{\lambda^2})\hskip-.5mm
\)\hskip-1mm
\(\hskip-.5mm1\hskip-.5mm+\hskip-.5mm
\dfrac{b_1}{\lambda}
\hskip-.5mm
+\hskip-.5mm
\dfrac{a_1}{\lambda}
\hskip-.5mm
+\hskip-.5mm
\dfrac{b_2}{\lambda^{3/2}}\hskip-.5mm+\hskip-.5mm\mathcal{O}(\dfrac{1}{\lambda^2})\hskip-1mm\)\hskip-.5mm
\\&
-\(1+\dfrac{L}{\sqrt{\lambda}}-\dfrac{L_y}{\lambda}+\dfrac{L_{yy}}{\lambda^{3/2}}+\mathcal{O}(\dfrac{1}{\lambda^2})\)\(1+\dfrac{b_1}{\lambda}+\dfrac{b_1^2+d_1}{\lambda}+\dfrac{b_1a_1+c_1}{\lambda^{3/2}}+\mathcal{O}(\dfrac{1}{\lambda^2})\)\Big|_{y=x=C}
\\&
=\hskip-1mm\dfrac{a_1\hskip-0.5mm-\hskip-0.5mm b_1^2\hskip-0.5mm-\hskip-0.5mmd_1\hskip-0.5mm+\hskip-0.5mmL_x\hskip-0.5mm+\hskip-0.5mmL_y}{\lambda}\hskip-0.5mm+\hskip-0.5mm\dfrac{-\hskip-0.5mmL_{xy}\hskip-0.5mm-\hskip-0.5mmL_{yy}\hskip-0.5mm+\hskip-0.5mm b_1\hskip-0.5mm(\hskip-0.5mm L_x\hskip-0.5mm +\hskip-0.5mm L_y\hskip-0.5mm) \hskip-0.5mm+ \hskip-0.5mm (\hskip-0.5mm a_1 \hskip-0.5mm-\hskip-0.5mm d_1 \hskip-0.5mm-\hskip-0.5mm b_1^2) L\hskip-0.5mm+\hskip-0.5mmb_2\hskip-0.5mm-\hskip-0.5mmb_1\hskip-0.5mm a_1 \hskip-0.5mm- \hskip-0.5mm c_1}{\lambda^{3/2}}\Big|_{y=x=C}\hskip-0.5mm + \hskip-0.5mm \mathcal{O}(\hskip-0.5mm \lambda^{-2}),
\end{split}\end{equation}
where we assumed that $u_{t_0}(x)$ is twice differentiable and $u^{(2)}(x,t_0)\in BV_{loc}.$ If the function $u_{t_0}(x)$ is only 1 time differentiable with $u'(x,t_0)\in BV_{loc},$ then, in formula \eqref{elh2} we would have to
reduce the \clu{expansion} by the last element,
and if $u_{t_0}(x)$ is just locally a function of bounded variation, we reduce it by the last two elements.

Hence, for $u_{t_0}(x)\in BV_{loc}$ we already have that 
$$\left\{e_l, h_2\right\}_C\cdot 2\e^{-2C\sqrt{\lambda}+\frac{2t_0}{3}\lambda^{3/2}-\frac{2}{105}\lambda^{7/2}}=\mathcal{O}(\frac{1}{\lambda}),$$
and hence $$b(\lambda)=\mathcal{O}(\frac{1}{\lambda})\e^{2B\sqrt{\lambda}-\frac{2t_0}{3}\lambda^{3/2}+\frac{2}{105}\lambda^{7/2}}.$$
For $u_{t_0}(x)$ differentiable with $u'(x,t_0)\in BV_{loc}$ we need to check that \begin{equation}\label{term1}\(a_1-b_1^2-d_1+L_x+L_y\)\Big|_{y=x=C}=0,\end{equation}
and for 
$u_{t_0}(x)$ 2 times differentiable with $u''(x,t_0)\in BV_{loc}$ we need to check that \begin{equation}\label{term32}\(-L_{xy}-L_{yy}+b_1(L_x+L_y)+(a_1-d_1-b_1^2)L+b_2-b_1a_1-c_1\)\Big|_{y=x=C}=0.\end{equation}

\noindent From \eqref{RL_IE}, setting $L(x,y)=:H(\frac{x+y}{2},\frac{x-y}{2}),$ \clu{ and differentiating, and then coming back from $H$ to $L,$}
 we obtain 
$$L(x,x)=-\int\limits_{-\infty}^x\widetilde u(s,t_0)\d s,\qquad \mbox{ for }\quad u\in L^1_{loc},$$
$$
L_x+L_y\Big|_{y=x}=-\widetilde u_{t_0}(x),\qquad \mbox{ for }\quad u\in C,
$$
$$L_{xy}+L_{yy}\Big|_{y=x}=-\frac12\widetilde u'(x,t_0)-\widetilde u_{t_0}(x)\int\limits_{-\infty}^{x}\widetilde u(s,t_0)\d s,
\qquad \mbox{ for }\quad u\in C^1.$$
On the other hand, from \eqref{U_a1b1} and $d_1=-a_1$ we obtain that 
$$2a_1-b_1^2=U,\quad b_2-b_1a_1-c_1=\frac{-U_x}{2}+U b_1,$$
and since $\widetilde u(C,t_0)=U(C,t_0)$ for $C=A,B$ and continuous $u,$ and $\widetilde u'(C,t_0)=U'(C,t_0)$ for $C=A,B$ and continuously differentiable $u,$ we see that \eqref{term1} and \eqref{term32} holds, and hence we obtain the statements of the lemma.

The general case can be proven as follows:
denote 
$$f(x)=h_1(-x,t_0;\lambda)\cdot \sqrt{2}\cdot \sqrt[4]{\lambda} \cdot \e^{\frac{-t_0}{3}\lambda^{3/2}+\frac{1}{105}\lambda^{7/2}},
\quad 
g(x)=e_l(-x,t_0;\lambda)\cdot \sqrt{2}\cdot \sqrt[4]{\lambda}\cdot \e^{\frac{-t_0}{3}\lambda^{3/2}+\frac{1}{105}\lambda^{7/2}}.$$
If $u_{t_0}(x)$ is $N$ times differentiable with $u^{(N)}(x,t_0)\in BV_{loc},$ then
$$f(x)=\e^{x\sqrt{\lambda}}\(\sum\limits_{j=0}^{N}\frac{\alpha_j}{\lambda^{j/2}}+\mathcal{O}(\lambda^{-N/2-1})\),\quad f'(x)=\sqrt{\lambda}\e^{x\sqrt{\lambda}}\(\sum\limits_{j=0}^{N}\frac{\alpha_j+\alpha_{j-1,x}}{\lambda^{j/2}}+
+\mathcal{O}(\lambda^{-N/2-1})\),
$$
$$
f''(x)=\lambda\e^{x\sqrt{\lambda}}\(\sum\limits_{j=0}^{N}\frac{\alpha_j+2\alpha_{j-1,x}+\alpha_{j-2,xx}}{\lambda^{j/2}}+\mathcal{O}(\lambda^{-N/2-1})\),$$
and similar formulas, with $\alpha_j(x)$ substituted by $\beta_j(x),$ hold for $g(x).$ 
%
We put here $\alpha_0=\beta_0=1, \alpha_{-j}=\beta_{-j}=0, j\geq 1.$
Substituting the above expansions into
$$f_{xx}+2\widetilde u(x)f=\lambda f,\quad g_{xx}+2 U(x)g=\lambda g,$$
one obtains
 \begin{equation}\label{alpha_der}\alpha_{j,x}=-u\alpha_{j-1}-\frac12\alpha_{j-1,xx},\quad \beta_{j,x}=-u\beta_{j-1}-\frac12\beta_{j-1,xx},\qquad j=1,\ldots,N.\end{equation}
Now, the term of order $\e^{2x\sqrt{\lambda}}\lambda^{-k/2}$ in the Wronskian
$$\sqrt{\lambda}\(f(x)g'(x)-f'(x)g(x)\)$$ is equal to
$$\sum\limits_{j=0}^k\alpha_{k-j}\(\beta_j+\beta_{j-1,x}\)-\sum\limits_{j=0}^k\beta_{k-j}\(\alpha_j+\alpha_{j-1,x}\),$$
and substituting subsequently expressions \eqref{alpha_der} instead of $\alpha_{j,x},$ $\beta_{j,x},$ we find that the above term is equal to 0, since the function $\widetilde u (x)-U(x)$ and its first $N-1$ derivatives vanish at the point $x=A.$
\end{proofof}

\subsection{Example: compactly supported perturbation with a constant in the middle.}\label{sect_constant}
For the initial function 
$$u(x, t_0)=\begin{cases}c,\ \  A<x<B,\\
U(x,t_0),\ \  x<A, x>B,\end{cases}$$
where $c\in\mathbb{R}$ is a constant, the functions $h_1, h_2$ have the explicit form
$$h_1(x)=\alpha\widehat{h_1}(x)+\beta \widehat{h_2}(x),\qquad h_2(x)=\gamma\widehat{h_1}(x)+\delta \widehat{h_2}(x),$$
where
$$\widehat{h_1}=\frac{1}{\sqrt{2}\sqrt[4]{\lambda-2c}}\e^{-x\sqrt{\lambda-2c}+\frac{t}{3}\lambda^{3/2}-\frac{1}{105}\lambda^{7/2}},\quad \widehat{h_2}=\frac{1}{\sqrt{2}\sqrt[4]{\lambda-2c}}\e^{x\sqrt{\lambda-2c}-\frac{t}{3}\lambda^{3/2}+\frac{1}{105}\lambda^{7/2}},$$
and 
$$\begin{cases}\alpha=\frac12\(\sqrt[4]{\frac{\lambda-2c}{\lambda}}+\sqrt[4]{\frac{\lambda}{\lambda-2c}}\)\e^{B(\sqrt{\lambda-2c}-\sqrt{\lambda})},\\\beta=\frac12\(\sqrt[4]{\frac{\lambda-2c}{\lambda}}-\sqrt[4]{\frac{\lambda}{\lambda-2c}}\)\e^{-B(\sqrt{\lambda-2c}+\sqrt{\lambda})},\end{cases}
\qquad
\begin{cases}\gamma=\frac12\(\sqrt[4]{\frac{\lambda-2c}{\lambda}}-\sqrt[4]{\frac{\lambda}{\lambda-2c}}\)\e^{A(\sqrt{\lambda-2c}+\sqrt{\lambda})},\\\delta=\frac12\(\sqrt[4]{\frac{\lambda-2c}{\lambda}}+\sqrt[4]{\frac{\lambda}{\lambda-2c}}\)\e^{-A(\sqrt{\lambda-2c}-\sqrt{\lambda})}.\end{cases}
$$
\clu{Furthermore}, in the formulas for $a(\lambda),$ $b(\lambda)$ in Section \ref{sect_b} we can \clu{replace} $h_1, h_2$ everywhere with $\widehat{h}_1, \widehat{h}_2,$ and we can write the large $\lambda$ asymptotics of $a,$ $b$ in a more explicit way. Namely, \clu{($a_1,$ $b_1$ below are defined in asymptotic expansion \eqref{asserEhat}),} 
$$\hskip0cm a(\lambda)=1+\frac{\(c(A-B)+b_1(A)-b_1(B)\)}{\sqrt{\lambda}}+\frac{\(c(A-B)+b_1(A)-b_1(B)\)^2}{2\lambda}+\mathcal{O}(\lambda^{-3/2}),$$ $\hfill\arg\lambda\in(-\pi+\varepsilon,-0]\cup[+0,\pi-\varepsilon),$

\medskip

$$a(\lambda)=1+\frac{\(c(A-B)+b_1(A)-b_1(B)\)}{\sqrt{\lambda}}+\frac{\(c(A-B)+b_1(A)-b_1(B)\)^2}{2\lambda}+\mathcal{O}(\lambda^{-3/2})+$$
$$\hskip0cm+\underbrace{\i\left\{\e^{2\theta(A)}\(\frac{\left[c-2a_1(A)+b_1^2(A)\right]}{2\lambda}+\mathcal{O}(\lambda^{-3/2})\)-\e^{2\theta(B)}\(\frac{\left[c-2a_1(B)+b_1^2(B)\right]}{2\lambda}+\mathcal{O}(\lambda^{-3/2})\)\right\}}_{\textrm{exponentially small term, which becomes oscillatory of order }\mathcal{O}\(\lambda^{-1}\) \textrm{ at }\arg\lambda=\pi-0},$$ $\hfill\arg\lambda\in[\frac{5\pi}{7}+\varepsilon,\pi-0],$

\medskip

$$a(\lambda)=1+\frac{\(c(A-B)+b_1(A)-b_1(B)\)}{\sqrt{\lambda}}+\frac{\(c(A-B)+b_1(A)-b_1(B)\)^2}{2\lambda}+\mathcal{O}(\lambda^{-3/2})-$$
$$\hskip-0cm\underbrace{-\i\left\{\e^{2\theta(A)}\(\frac{\left[c-2a_1(A)+b_1^2(A)\right]}{2\lambda}+\mathcal{O}(\lambda^{-3/2})\)-\e^{2\theta(B)}\(\frac{\left[c-2a_1(B)+b_1^2(B)\right]}{2\lambda}+\mathcal{O}(\lambda^{-3/2})\)\right\}}_{\textrm{exponentially small term, which becomes oscillatory of order }\mathcal{O}\(\lambda^{-1}\) \textrm{ at }\arg\lambda=-\pi+0},$$ $\hfill\arg\lambda\in[-\pi+0, -\pi+\varepsilon)],$

\medskip

$$a_u(\lambda)=1-\frac{(A-B)c+b_1(A)-b_1(B)}{\sqrt{\lambda}}+\frac{\((A-B)c+b_1(A)-b_1(B)\)^2}{2\lambda}+\mathcal{O}(\lambda^{-3/2})+$$
$$\hskip-0cm\underbrace{\frac{-\i}{2\lambda}\(\e^{-2\theta(B)}(c-2a_1(B)+b_1^2(B)+\mathcal{O}(\lambda^{-3/2}))-\e^{-2\theta(A)}(c-2a_1(A)+b_1^2(A)+\mathcal{O}(\lambda^{-3/2}))\)}_{\textrm{leading term}},$$ $\hfill\arg\lambda\in[-\pi+0,\frac{-5\pi}{7}-\varepsilon],$

\medskip

$$\bullet\ \Omega_{I}\quad  a_u(\lambda)=1+\dfrac{c(A-B)+b_1(A)-b_1(B)}{\sqrt{\lambda}}+\frac{\(c(A-B)+b_1(A)-b_1(B)\)^2}{2\lambda}+\mathcal{O}(\lambda^{-3/2}),$$ $\hfill\arg\lambda\in[0,\pi-\varepsilon],$

$\bullet\ \Omega_{IV}\quad a_u(\lambda)=1+\dfrac{c(A-B)+b_1(A)-b_1(B)}{\sqrt{\lambda}}+\dfrac{\(c(A-B)+b_1(A)-b_1(B)\)^2}{2\lambda}+\mathcal{O}(\lambda^{-3/2})+$
$$\hskip-0cm +\frac{-\i}{2\lambda}\left[\e^{-2\theta(B)}\(c-2a_1(B)+b_1^2(B)+\mathcal{O}(\lambda^{-1/2})\)-\e^{-2\theta(A)}\(c-2a_1(A)+b_1^2(A)+\mathcal{O}(\lambda^{-1/2})\)\right],$$ $\hfill\arg\lambda\in[-\pi+\varepsilon,0],$

$$\bullet\ \Omega_{II}\quad a_u(\lambda)=1+\frac{c(A-B)+b_1(A)-b_1(B)}{\sqrt{\lambda}}+\frac{\(c(A-B)+b_1(A)-b_1(B)\)^2}{2\lambda}+\mathcal{O}(\lambda^{-3/2})+$$
$$\hskip-0cm +\frac{-\i}{2\lambda}\left[\e^{2\theta(B)}\(c-2a_1(B)+b_1^2(B)+\mathcal{O}(\lambda^{-1/2})\)-\e^{2\theta(A)}\(c-2a_1(A)+b_1^2(A)+\mathcal{O}(\lambda^{-1/2})\)\right],$$ $\hfill\arg\lambda\in[\frac{5\pi}{7}+\varepsilon,\pi],$

$$\bullet \ \Omega_{III}\quad a_u(\lambda)=1-\frac{c(A-B)+b_1(A)-b_1(B)}{\sqrt{\lambda}}+\frac{\(c(A-B)+b_1(A)-b_1(B)\)^2}{2\lambda}+\mathcal{O}(\lambda^{-3/2})+$$
$$+\frac{-\i}{2\lambda}\left[\e^{-2\theta(B)}\(c-2a_1(B)+b_1^2(B)\)-\e^{-2\theta(A)}\(c-2a_1(A)+b_1^2(A)\)\right],
\arg\lambda\in[-\pi,\frac{-5\pi}{7}-\varepsilon].$$
\medskip
Asymptotics for $a_d(\lambda)$ in all the sectors of the complex plane $\lambda$ follows from the asymptotics of $a_u(\lambda)$ by the formula $a_d(\lambda)=\ol{a_u(\ol{\lambda})}.$
\clu{Furthermore},
$$b(\lambda)=\(\frac{-c+2a_1(A)-b_1^2(A)}{2\lambda}+\mathcal{O}(\lambda^{-3/2})\)\e^{2\theta(A)}+
\(\frac{c-2a_1(B)+b_1^2(B)}{2\lambda}+\mathcal{O}(\lambda^{-3/2})\)\e^{2\theta(B)}$$
$$=\e^{2\theta(B)}\(\frac{c-2a_1(B)+b_1^2(B)}{2\lambda}+\mathcal{O}(\lambda^{-3/2})-\underbrace{\e^{2(A-B)\sqrt{\lambda}}\left[\frac{c-2a_1(A)+b_1^2(A)}{2\lambda}+\mathcal{O}(\lambda^{\frac{-3}{2}})\right]}_{\hskip-1.5cm\textrm{exponentially small term which becomes oscillatory for }\arg\lambda=\pm\pi}\)$$
$$\hskip10cm \arg\lambda\in[-\pi+0,-0]\cup[+0,\pi-0],$$

$$\hskip-0cm b_u(\lambda)=\frac{\e^{-2\theta(B)}\left[c-2a_1(B)+b_1^2(B)+\mathcal{O}(\lambda^{-1/2})\right]-\overbrace{\e^{-2\theta(A)}\left[c-2a_1(A)+b_1^2(A)+\mathcal{O}(\lambda^{-1/2})\right]}^{\textrm{main part}}}{2\lambda},$$ $\hfill\arg\lambda\in[-\pi+0,\frac{-5\pi}{7}-\varepsilon],$

$$\hskip-0cm b_d(\lambda)=\frac{\e^{-2\theta(B)}\left[c-2a_1(B)+b_1^2(B)+\mathcal{O}(\lambda^{-1/2})\right]-\overbrace{\e^{-2\theta(A)}\left[c-2a_1(A)+b_1^2(A)+\mathcal{O}(\lambda^{-1/2})\right]}^{\textrm{main part}}}{2\lambda},$$ $\hfill\arg\lambda\in[\frac{5\pi}{7}+\varepsilon, \pi-0].$

\begin{remark}The functions $a_u(\lambda),$ $a_d(\lambda) $ are entire, i.e. they do not have any jump on the real axis. This is not obvious since their representation \eqref{a_h12} involves $h_{1,2},$ which do have jumps 
across some part of the real axis.
\end{remark}

{
\begin{remark} Consider the KdV equation \eqref{KdV} (with reverse time, i.e. instead of $t\geq t_0$ we take $\clu{t\leq t_0<0}$)
$$u_t(x,t)+u(x,t)u_x(x,t)+u_{xxx}(x,t)=0, \quad \clu{t\leq t_0<0}$$ with the initial datum of the form
$$u_{t_0}(x)=\begin{cases}\frac{x}{t\clu{_0}},\ x<A, \ x>B, \\ c, A<x<B,\end{cases}$$
where $A<B,c, \clu{t_0<0}$ are real constants.
Based on the unperturbed Jost solutions $e^{Ai}_{l,r}$ associated with the function $\frac{x}{t}$ and 
constructed in subsection \ref{sect_Airy}, one can construct the Jost solutions $f^{Ai}_{l,r}$ associated with the 
function $u_{t_0}(x),$ and then the corresponding spectral functions $a^{Ai}(\lambda),$ $b^{Ai}(\lambda),$ 
$r^{Ai}(\lambda).$
Then the function $a^{Ai}(\lambda)$ does not vanish nowhere, but if we take 
the constant $c$ big enough, then $a^{Ai}(\lambda)$ takes some values which are very close to $0$ (``quasi-spectrum").
\end{remark}}

\section{Construction of Riemann-Hilbert problems}\label{sect_RH}
In order to construct a solution $u(x,t)$ to the KdV equation, which at the time $t=t_0$ is equal to the given initial function $u(x, t_0),$ 
our strategy is to construct a solution to a Riemann-Hilbert problem out of the Jost solutions $F_l(x,t_0;\lambda),$ $F_r(x,t_0;\lambda),$ in such a way that this RH problem 
makes sense also for $t\neq t_0.$ 
There are several ways to do this.
In this section the initial function $u_{t_0}(x)$ is a compactly supported perturbation of $U(x,t_0),$ i.e. $u(x, t_0)$ 
satisfies the conditions of Lemma \ref{lem_com_sup}. The functions  $a(\lambda),b(\lambda),r(\lambda)$ are the 
spectral functions associated with the initial function $u_{t_0}(x).$

\subsection{RH problem appropriate for $t>t_0, x\in\mathbb{R}$ and $t=t_0,$ $x>B.$}\label{sect_RHPtp}
{\color{black}
Let us notice that  relation \eqref{b_a_rel}
$$\frac{b_u}{a_u}-\frac{b_d}{a_d}=\frac{\i(a_ua_d-1)}{a_ua_d}$$
together with scattering relation \eqref{scat_rel} 
$$\frac{F_{lu}}{a_u}-\frac{F_{ld}}{a_d}=\frac{-i}{a_ua_d}F_r$$
imply
$$\frac{b_u}{a_u}F_r-\frac{b_d}{a_d}F_r=\frac{\i(a_ua_d-1)}{a_ua_d}F_r$$
Hence, substracting the two latter formulas and multiplying them by $\i,$ we get 
$$\frac{\i\(F_{lu}-b_u F_r\)}{a_u}-\frac{\i\(F_{ld}-b_d F_r\)}{a_d}=F_r.$$

\noindent Now we are ready to define the piece-wise analytic in $\lambda$ matrix-valued function
$$P(x,t_0;\lambda):=
\begin{cases}
\(\dsfrac{F_{lu}-b_uF_r}{a_u}\e^{-\theta}, \quad \(F_r-\frac{\i(F_{lu}-b_uF_r)}{a_u}\)\e^{\theta}\), \Omega_{II}, & 
\quad \(\dsfrac{1}{a_u}F_{lu}\e^{-\theta}, F_r\e^{\theta}\), \Omega_I,\\\\
\(\dsfrac{F_{ld}-b_dF_r}{a_d}\e^{-\theta},\quad \(F_r+\frac{\i(F_{ld}-b_dF_r)}{a_d}\)\e^{\theta}\), \Omega_{III},&
 \quad \(\dsfrac{1}{a_d}F_{ld}\e^{-\theta}, F_r\e^{\theta}\), \Omega_{IV},
\end{cases}$$
where we denoted 
$$\Omega_{II}=\left\{\lambda:\ \ \arg\lambda\in(\frac{6\pi}{7},\pi)\right\},\ \ 
\Omega_{I}=\left\{\lambda:\ \ \arg\lambda\in(0,\frac{6\pi}{7})\right\},$$
$$\Omega_{III}=\left\{\lambda:\ \ \arg\lambda\in(-\pi,\frac{-6\pi}{7})\right\},\ \ 
\Omega_{IV}=\left\{\lambda:\ \ \arg\lambda\in(\frac{-6\pi}{7},0)\right\},$$

\noindent This matrix-valued function $P(x,t_0, \lambda)$ satisfies the following RH problem at the time $t=t_0$ for $x>B:$
\begin{RH}\label{RHtp}{\textbf{(appropriate for $t>t_0, x\in\mathbb{R}$ or $t=t_0$ and $x>B.$)}}
To find a $2\times2$ matrix-valued function $P(x,t;\lambda),$ which
\begin{itemize} 
\item is analytic in $\lambda\in\mathbb{C}\setminus\Sigma,$ where $\Sigma$ is as in \eqref{Sigma},
\item has the following jump $P_+=P_-J_P$ across $\Sigma:$
$$J_P=\begin{pmatrix}1 & 0 \\ \dsfrac{-i \e^{-2\theta}}{a_ua_d}& 1\end{pmatrix},\gamma_0,
\qquad J_P=\begin{pmatrix}0&-\i\\-\i&0\end{pmatrix},\rho,$$
$$J_P=\begin{pmatrix}1+\overbrace{\dsfrac{\i b_u}{a_u}}^{\textrm{small due to } b_u} & \overbrace{i\e^{2\theta}}^{\textrm{exp decay due to exp}}\\\\\underbrace{\dsfrac{b_u \e^{-2\theta}}{a_u}}_{\textrm{decay due to } b_u} & 1\end{pmatrix},\gamma_3,
\quad J_P=\begin{pmatrix}1 & \i\e^{2\theta}\\\\ \dsfrac{-b_d \e^{-2\theta}}{a_d} & 1-\dsfrac{\i b_d}{a_d} \end{pmatrix}, \gamma_{-3},$$
\noindent where we denoted $$\gamma_0=(0,+\infty),\ \ \gamma_3=(\e^{6\pi\i/7}\infty,0),\ \ \gamma_{-3}=(\e^{-6\pi\i/7}\infty,0),\ \ \rho=(-\infty,0),$$
\item has the following asymptotics as $\lambda\to\infty,$ which is uniform w.r.t. $\arg\lambda\in[-\pi,\pi]:$
$$P=\frac{1}{\sqrt{2}}\lambda^{-\sigma_3/4}\begin{pmatrix}1 & 1\\1&-1\end{pmatrix}\(I+b\sigma_3\frac{1}{\sqrt{\lambda}}+\mathcal{O}(\lambda^{-1})\),$$
where $b=b(x,t)$ is some scalar (which is not fixed, but is introduced in order to fix the form of the asymptotics).

\end{itemize}
\end{RH}

\begin{figure}[ht]
\center
\begin{tikzpicture}
\node at (-9,3) {$P_+=P_-J_{P}:$};

\draw 
[decoration={markings, mark=at position  0.5 with {\arrow{<}}}, postaction={decorate}]
(0,0) -- (-5.40581, -2.6033);
\draw [decoration={markings, mark=at position  0.5 with {\arrow{>}}}, postaction={decorate}] 
(-5.40581, 2.6033) -- (0,0);

\draw[decoration={markings, mark=at position  0.75 with {\arrow{>}}}, postaction={decorate}]
[decoration={markings, mark=at position  0.25 with {\arrow{>}}}, postaction={decorate}](-9,0) -- (4,0);

\node at (.6,1) {\color{blue}$\(\dsfrac{F_{lu}\e^{-\theta}}{a_u},F_r \e^{\theta}\)$};
\node at (.6,-1) {\color{blue}$\(\dsfrac{F_{ld}\e^{-\theta}}{a_d},F_r \e^{\theta}\)$};
\node at (-5.5,0.6) {\color{blue}$\(\overbrace{\dsfrac{F_{lu}- b_u F_r}{a_u}}, \overbrace{F_r\(1+\frac{\i b_u}{a_u}\)-\frac{\i F_{lu}}{a_u}}\)$};
\node at (-6.6,1.4) {\color{blue}$\cdot \e^{-\theta}$};
\node at (-4.6,1.4) {\color{blue}$\cdot \e^{\theta}$};
\node at (-5,-0.6) {\color{blue}$\(\underbrace{\dsfrac{F_{ld}-b_d F_r}{a_d}},\underbrace{F_r\(1-\frac{\i b_d}{a_d}\)+\frac{\i F_{ld}}{a_d}}\)$};
\node at (-6.6,-1.4) {\color{blue}$\cdot \e^{-\theta}$};
\node at (-4.2,-1.4) {\color{blue}$\cdot \e^{\theta}$};
\node at (-5.8,2.8) {$\frac{6\pi}{7}$};
\draw (-5.8, 2.8) circle [radius=9pt];
\node at (-5.8,-2.8) {$\frac{-6\pi}{7}$};
\draw (-5.8,-2.8) circle [radius=9pt];

\node at (3,.8){$\begin{pmatrix}1 & 0 \\ \dsfrac{-\i \e^{-2\theta}}{a_ua_d} & 1\end{pmatrix}$};
\node at (-1.8,-2.6){$\begin{pmatrix}1 & \i \e^{2\theta}\\\\ \dsfrac{-b_d \e^{-2\theta}}{a_d} & 1-\dsfrac{\i b_d}{a_d}\end{pmatrix}$};
\node at (-2.,2.5){$\begin{pmatrix}1+\dsfrac{\i b_u}{a_u} & \i \e^{2\theta}\\\\ \dsfrac{b_u \e^{-2\theta}}{a_u} & 1\end{pmatrix}$};
\node at (-9.2,0.5){$\begin{pmatrix}0 & -\i \\ -\i & 0\end{pmatrix}$};

\end{tikzpicture}
\end{figure}
\vskip1cm

\noindent For this RH problem to be meaningful the jumps must vanish as $\lambda\to\infty.$ Notice that, by Lemma \ref{lem_b_a_prop},$$r_u\e^{-2\theta(x,t;\lambda)}=\mathcal{O}(\frac{1}{\sqrt{\lambda}})\cdot\e^{\frac{2(t-t_0)}{3}\lambda^{3/2}+2(B-x)\lambda^{1/2}},$$
\noindent and hence this is indeed the case for $t>t_0, x\in\mathbb{R}$ and for $t=t_0, x>B.$

Together with RHP \ref{RHtp} we can consider another one, for the function $\widehat{P}(x,t;\lambda),$ with the same analyticity 
and jump condition, but with different asymptotics as $\lambda\to\infty:$
\begin{RH}\label{RHtphat}
Find a function $\widehat{P}(x,t;\lambda),$ with analyticity and jump as in RH \ref{RHtp}, and with the asymptotic condition altered:
 \begin{itemize}
  \item asymptotics as $\lambda\to\infty$
  $$\widehat{P}(x,t;\lambda)=\(I+\mathcal{O}(\frac{1}{\lambda})\)\dfrac{\lambda^{-\sigma_3/4}}{\sqrt{2}}
  \begin{pmatrix}
   1&1\\1&-1
  \end{pmatrix}.
$$
 \end{itemize}
\end{RH}

\subsection{RH problem appropriate for $t<t_0, x\in\mathbb{R}$ and $t=t_0, x<A.$}
Another way to construct a RH problem is as follows.
Define a piece-wise analytic matrix-valued function
$$N=\begin{cases}
\(\dsfrac{F_{lu}}{a_u}\e^{-\theta},\quad \(F_r-\i a_d F_{lu}\)\e^{\theta}\), \Omega_{II}, \quad &\(\dsfrac{1}{a_u}F_{lu}\e^{-\theta}, F_r\e^{\theta}\), \Omega_{I},\\\\
\(\dsfrac{F_{ld}}{a_d}\e^{-\theta},\quad \(F_r+\i a_u F_{ld}\)\e^{\theta}\), \Omega_{III}, \quad &\(\dsfrac{1}{a_d}F_{ld}\e^{-\theta}, F_r\e^{\theta}\), \Omega_{IV}.
\end{cases}$$

This matrix-valued function $N(\lambda)$ satisfies the following RH problem at the time $t=t_0$ for $x<A:$
\begin{RH}\label{RHtn}
{\textbf{(appropriate for $t<t_0, x\in\mathbb{R}$ or $t=t_0$ and $x<A.)$}}To find a $2\times2$ matrix-valued function $N(x,t;\lambda),$ which
\begin{itemize} 
\item is analytic in $\lambda\in\mathbb{C}\setminus\Sigma,$ where $\Sigma$ is as in \eqref{Sigma},
\item has the following jump $N_+=N_-J_N$ across $\Sigma:$
$$J_N=\begin{pmatrix}1 & 0 \\ \dsfrac{-i \e^{-2\theta}}{a_ua_d}& 1\end{pmatrix},\gamma_0,
\qquad J_N=\begin{pmatrix}0&-\i a_ua_d\\\frac{-\i}{a_ua_d}&0\end{pmatrix},\rho$$
$$J_N=\begin{pmatrix}1 & \overbrace{i a_ua_d\e^{2\theta}}^{\textrm{exp decay due to exp}}\\\\ 0 & 1\end{pmatrix},\gamma_3,
\quad J_N=\begin{pmatrix}1 & \i a_ua_d\e^{2\theta} \\ 0 & 1\end{pmatrix}, \gamma_{-3}.$$
\item has the following asymptics as $\lambda\to\infty,$ which is uniform w.r.t. $\arg\lambda\in[-\pi,\pi]:$
$$N=\frac{1}{\sqrt{2}}\lambda^{-\sigma_3/4}\begin{pmatrix}1 & 1\\1&-1\end{pmatrix}\(I+b\sigma_3\frac{1}{\sqrt{\lambda}}+\mathcal{O}(\lambda^{-1})\),$$
where $b=b(x,t)$ is some scalar (which is not fixed, but introduced in order to fix the form of the asymptotics).
\end{itemize}
\end{RH}
\noindent 
\noindent A condition for this RH problem \ref{RHtp} to be meaningful is that the jumps vanish as $\lambda\to\infty.$  Notice that by Lemma \ref{lem_b_a_prop} $$(a_d-a_u)\e^{2\theta(x,t;\lambda)}=\mathcal{O}(\frac{1}{\sqrt{\lambda}})\cdot\e^{\frac{2(t_0-t)}{3}\lambda^{3/2}+2(x-A)\lambda^{1/2}},$$
\noindent and hence this RH problem \ref{RHtn} is good for $t<t_0$ and for $t=t_0, x<A.$

\begin{figure}[ht]
\center
\begin{tikzpicture}
\node at (-8.8,2.8) {$N_+=N_-J_{N}:$};

\draw 
[decoration={markings, mark=at position  0.5 with {\arrow{<}}}, postaction={decorate}]
(0,0) -- (-5.40581, -2.6033);
\draw [decoration={markings, mark=at position  0.5 with {\arrow{>}}}, postaction={decorate}] 
(-5.40581, 2.6033) -- (0,0);

\draw[decoration={markings, mark=at position  0.75 with {\arrow{>}}}, postaction={decorate}]
[decoration={markings, mark=at position  0.25 with {\arrow{>}}}, postaction={decorate}](-9,0) -- (4,0);

\node at (.6,1) {\color{blue}$\(\dsfrac{F_{lu}\e^{-\theta}}{a_u},F_r \e^{\theta}\)$};
\node at (.6,-1) {\color{blue}$\(\dsfrac{F_{ld}\e^{-\theta}}{a_d},F_r \e^{\theta}\)$};
\node at (-5.5,0.6) {\color{blue}$\(\overbrace{\dsfrac{F_{lu}}{a_u}}, \overbrace{F_r-\i a_d F_{lu}}\)$};
\node at (-6.1,1.4) {\color{blue}$\cdot \e^{-\theta}$};
\node at (-5.1,1.4) {\color{blue}$\cdot \e^{\theta}$};
\node at (-5,-0.6) {\color{blue}$\(\underbrace{\dsfrac{F_{ld}}{a_d}},\underbrace{F_r+\i a_u F_{ld}}\)$};
\node at (-5.8,-1.4) {\color{blue}$\cdot \e^{-\theta}$};
\node at (-4.2,-1.4) {\color{blue}$\cdot \e^{\theta}$};
\node at (-5.8,2.8) {$\frac{6\pi}{7}$};
\draw (-5.8, 2.8) circle [radius=9pt];
\node at (-5.8,-2.8) {$\frac{-6\pi}{7}$};
\draw (-5.8,-2.8) circle [radius=9pt];

\node at (3,.8){$\begin{pmatrix}1 & 0 \\ \dsfrac{-\i \e^{-2\theta}}{a_ua_d} & 1\end{pmatrix}$};
\node at (-1.8,-2.6){$\begin{pmatrix}1 & \i a_ua_d \e^{2\theta}\\\\ 0 & 1\end{pmatrix}$};
\node at (-2.,2.5){$\begin{pmatrix}1 & \i a_ua_d\e^{2\theta}\\\\ 0 & 1\end{pmatrix}$};
\node at (-8.6,0.5){$\begin{pmatrix}0 & -\i a_ua_d \\ \frac{-\i}{a_ua_d} & 0\end{pmatrix}$};

\end{tikzpicture}
\end{figure}

Together with the RHP \ref{RHtn} we can consider another one, for the function 
$\widehat{N}(x,t;\lambda),$ with the same analyticity 
and jump condition, but with different asymptotics as $\lambda\to\infty:$
\begin{RH}\label{RHtnhat}
To find a function $\widehat{N}(x,t;\lambda),$ with analyticity and jump conditions as in RH \ref{RHtn}, and with the asymptotic condition altered:
 \begin{itemize}
  \item asymptotics as $\lambda\to\infty$
  $$\widehat{N}(x,t;\lambda)=\(I+\mathcal{O}(\frac{1}{\lambda})\)\dfrac{\lambda^{-\sigma_3/4}}{\sqrt{2}}
  \begin{pmatrix}
   1&1\\1&-1
  \end{pmatrix}.
$$
 \end{itemize}
\end{RH}
}

\subsection{RH problem appropriate for all $t\in\mathbb{R},$ $x\in\mathbb{R}.$}
A third way to construct a RH problem is to define a piece-wise {\textit{meromorphic}} in $\lambda$ matrix-valued function
$$
\mathbb{F}(x,t_0;\lambda)=
\begin{cases}
 \begin{pmatrix}
  \frac{F_{lu}}{a_u}, & F_{r}-\frac{\i}{a_u+\i b_u} F_{lu} 
 \end{pmatrix}, \Omega_{II},
\qquad
\begin{pmatrix}
 \frac{F_{lu}}{a_u}, & F_{r}
\end{pmatrix}, \Omega_I,
\\
 \begin{pmatrix}
  \frac{F_{ld}}{a_d}, & F_{r}+\frac{\i}{a_d-\i b_d} F_{ld} 
 \end{pmatrix}, \Omega_{III},
\qquad
\begin{pmatrix}
 \frac{F_{ld}}{a_d}, & F_{r}
\end{pmatrix}, \Omega_{IV},
\end{cases}
$$
By Lemma \ref{lem_b_a_prop}, there is at most a finite number of zeros of $a_u+\i b_u$ in the region $\Omega_{II},$ moreover, those zeros do not lie on the real axis. In case if some zeros fall on the border $\gamma_3$ between $\Omega_{II}$ and $\Omega_{I},$ we will locally deform a bit the line $\gamma_3,$ so that $\gamma_3$ would be  free of zeros of $a_u+\i b_u.$
Symmetrically, we will move $\gamma_{-3}.$ We keep the same notations for the deformed rays $\gamma_3, \gamma_{-3}.$

Function $\mathbb{F}(x, t_0; \lambda)$ solves the following RHP \ref{RHF} at the time $t=t_0:$

\begin{RH}\label{RHF} {\textbf{(Appropriate for all real $t$ and $x.$)}}
To find a $2\times2$ matrix-valued function $\mathbb{F}(x,t;\lambda),$ which
\begin{enumerate}
 \item is analytic in $\lambda\in\mathbb{C}\setminus\Sigma,$
where $\Sigma$ is as in \eqref{Sigma},
\item has the following jump $\mathbb{F}_+=\mathbb{F}_-J_{\mathbb{F}}$ across $\Sigma:$
$$J_{\mathbb{F}}=\begin{cases}
                  \begin{pmatrix}
                   \frac{-\i r_u}{1-\i r_d} & \frac{-\i}{(1+\i r_u)(1-\i r_d)} \\ \frac{-\i}{a_ua_d} & \frac{\i r_d}{1+\i r_u}
                  \end{pmatrix},\ \lambda\in\rho,
\qquad
\begin{pmatrix}
 1 & 0 \\ \frac{-\i}{a_ua_d} & 1
\end{pmatrix}, \ \lambda\in\gamma_0,
\\\\
\qquad
\begin{pmatrix}
 1 & \frac{\i}{1+\i r_u} \\ 0 & 1
\end{pmatrix}, \ \lambda\in\gamma_3,
\qquad 
\begin{pmatrix}
 1 & \frac{\i}{1-\i r_d} \\ 0 & 1
\end{pmatrix}, \ \lambda\in\gamma_{-3},
\end{cases},
$$
\item \label{pole_cond} has the following pole conditions at the roots of $r_u=\i,$ $r_d=-\i$: 
\\for $\lambda^*\in II,$ $\Im\lambda^*>0$ such that $r_u(\lambda^*)=\i,$
$$\mathbb{F}_{[2]}(\lambda)+\frac{\i}{1+\i r_u(\lambda)}\mathbb{F}_{[1]}=\mathcal{O}(1)\quad \mbox{ and }\quad  \mathbb{F}_{[1]}=\mathcal{O}(1) \quad \mbox{ for }\quad \lambda\to\lambda^*,$$
$$\mathbb{F}_{[2]}(\lambda)-\frac{\i}{1-\i r_d(\lambda)}\mathbb{F}_{[1]}=\mathcal{O}(1) \quad \mbox{ and }\quad \mathbb{F}_{[1]}=\mathcal{O}(1) \quad \mbox{ for }\quad\lambda\to\ol{\lambda^*},$$
\item has the following asymptotics as $\lambda\to\infty,$ which is uniform w.r.t. $\arg\lambda\in[-\pi,0]\cup[0,\pi]:$
$$\mathbb{F}(x,t;\lambda)=\frac{1}{\sqrt{2}}\lambda^{-\sigma_3/4}\begin{pmatrix}1&1\\1&-1\end{pmatrix}\(I+\widetilde b\sigma_3\frac{1}{\sqrt{\lambda}}
+{\color{black} {\ol{_\mathcal{O}}}(\frac{1}{\sqrt\lambda})}\) \e^{\theta(x, t; \lambda)\sigma_3},$$
where $\widetilde b=\widetilde b(x,t)$ is some scalar, which is not fixed, but introduced in order to fix the form of the asymptotics.
\end{enumerate}
\end{RH}

Together with RH problem \ref{RHF} for the function $\mathbb{F}(x,t;\lambda)$ we will also consider another one for the function 
$\widehat{\mathbb{F}}(x,t;\lambda),$ which has the same analyticity, pole, jump conditions as $\mathbb{F}(x,t;\lambda),$ but asymptotic condition as $\lambda\to\infty$ is replaced with 
 \begin{enumerate}
  \item [{\textit{4a.}}] Asymptotics as $\lambda\to\infty:$ 
  $$\widehat{\mathbb{F}}(\lambda)=\(I+{\ol{_\mathcal{O}}}(1)\)\dfrac{\lambda^{-\sigma_3/4}}{\sqrt{2}}
  \begin{pmatrix}
   1&1\\1&-1
  \end{pmatrix}\e^{\theta\sigma_3}.$$
 \end{enumerate}

\begin{figure}[ht]
\center
\begin{tikzpicture}
\node at (-9,3) {{\color{blue}$M=\widehat{\mathbb{F}}\e^{-\theta\sigma_3},\ M_+=M_-J_M:$}};

\draw 
[decoration={markings, mark=at position  0.5 with {\arrow{<}}}, postaction={decorate}]
(0,0) -- (-5.40581, -2.6033);
\draw [decoration={markings, mark=at position  0.4 with {\arrow{>}}}, postaction={decorate}] 
(-5.40581, 2.6033) -- (0,0);

\draw[decoration={markings, mark=at position  0.85 with {\arrow{>}}}, postaction={decorate}]
[decoration={markings, mark=at position  0.3 with {\arrow{>}}}, postaction={decorate}](-11,0) -- (3.5,0);

\node at (.1,1) {\color{blue}$\(\dsfrac{F_{lu}\e^{-\theta}}{a_u},F_r \e^{\theta}\)$};
\node at (.6,-1) {\color{blue}$\(\dsfrac{F_{ld}\e^{-\theta}}{a_d},F_r \e^{\theta}\)$};
\node at (-4.5,0.6) {\color{blue}$\(\overbrace{\dsfrac{F_{lu}}{a_u}}, \overbrace{F_r-\frac{\i}	{a_u+\i b_u} F_{lu}}\)$};
\node at (-5.6,1.4) {\color{blue}$\cdot \e^{-\theta}$};
\node at (-4.1,1.4) {\color{blue}$\cdot \e^{\theta}$};
\node at (-5,-0.8) {\color{blue}$\(\underbrace{\dsfrac{F_{ld}}{a_d}},\underbrace{F_r+\frac{\i }{a_d-\i b_d}F_{ld}}\)$};
\node at (-6.1,-1.6) {\color{blue}$\cdot \e^{-\theta}$};
\node at (-4.2,-1.6) {\color{blue}$\cdot \e^{\theta}$};
\node at (-5.8,2.8) {$\frac{6\pi}{7}$};
\draw (-5.8, 2.8) circle [radius=9pt];
\node at (-5.8,-2.8) {$\frac{-6\pi}{7}$};
\draw (-5.8,-2.8) circle [radius=9pt];

\node at (2.5,.8){$\begin{pmatrix}1 & 0 \\ \dsfrac{-\i \e^{-2\theta}}{a_ua_d} & 1\end{pmatrix}$};
\node at (-1.8,-2.6){$\begin{pmatrix}1 & \dsfrac{\i \e^{2\theta}}{1-\i r_d}\\\\ 0 & 1\end{pmatrix}$};
\node at (-2.,2.5){$\begin{pmatrix}1 & \dsfrac{\i \e^{2\theta}}{1+\i r_u}\\\\ 0 & 1\end{pmatrix}$};
\node at (-8.9,0.8){$\begin{pmatrix}\frac{-\i r_u \e^{\theta_--\theta_+}}{1-\i r_d} & \frac{-\i}{(1+\i r_u)(1-\i r_d)} \\\\ \frac{-\i}{a_ua_d} & \frac{\i r_d\ \e^{\theta_+-\theta_-}}{1+\i r_u}\end{pmatrix}$};

\end{tikzpicture}
\end{figure}

\noindent Let us mention that there are 2 ways to rewrite meromorphic RHP \ref{RHF} into a regular one.

\vskip5mm	\noindent \textbf{Regular RH problem.} The first one, is to redefine $\mathbb{F}$ in small neighborhoods of the points $\lambda^*\in\Omega_{II}$ 
with $r_u(\lambda^*)=\i,$ and points $\ol{\lambda^*}\in\Omega_{III}$ with $r_d(\ol{\lambda^*})=-\i:$
$$\mathbb{F}_{reg}(x,t;\lambda)=\begin{pmatrix}
                    \mathbb{F}_{[1]}(\lambda), & \mathbb{F}_{[2]}(\lambda)+\frac{\i}{1+\i r_u(\lambda)}\mathbb{F}_{[1]}(\lambda)
                   \end{pmatrix}, |\lambda-\lambda^*|<\varepsilon, \lambda^*\in\Omega_{II}, r_u(\lambda^*)=\i, 
$$
$$\mathbb{F}_{reg}(x,t;\lambda)=\begin{pmatrix}
                    \mathbb{F}_{[1]}(\lambda), & \mathbb{F}_{[2]}(\lambda)-\frac{\i}{1-\i r_d(\lambda)}\mathbb{F}_{[1]}(\lambda)
                   \end{pmatrix}, |\lambda-\ol{\lambda^*}|<\varepsilon, \lambda^*\in\Omega_{II}, r_u(\lambda^*)=\i, 
$$
$$\mathbb{F}_{reg}(x,t;\lambda)=\mathbb{F}(x,t;\lambda)\quad\mbox{elsewhere}.$$
The function $\mathbb{F}_{reg}(x,t;\lambda)$ is regular at $\lambda^*,$ $\ol{\lambda^*},$ and solves RH problem \ref{RHF} with pole conditions replaced by additional jumps across circles $\mathcal{C}_j,$
$\ol{\mathcal{C}_j}$ around the points $\lambda^*,$ $\ol{\lambda^*},$ oriented counter-clock-wise:
\begin{enumerate}\item[{\textit{ 3a.}}] $ \mathbb{F}_{reg,+}=\mathbb{F}_{reg,-}
\begin{pmatrix}
 1 & \frac{\i}{1+\i r_u(\lambda)} \\0&1
\end{pmatrix}, \mathcal{C}_j,\qquad
\mathbb{F}_{reg,+}=\mathbb{F}_{reg,-}
\begin{pmatrix}
 1 & \frac{-\i}{1-\i r_d(\lambda)} \\0&1
\end{pmatrix}, \ol{\mathcal{C}_j}.
$\end{enumerate}
RHP for $\widehat{\mathbb{F}}_{reg}$ is the same as for $\mathbb{F}_{reg},$ but with the asymptotic condition replaced with
 \begin{enumerate}
  \item [{\textit{4a.}}] Asymptotics as $\lambda\to\infty:$ 
  $$\widehat{\mathbb{F}}_{reg}(\lambda)=\(I+{\ol{_\mathcal{O}}}(1)\)\dfrac{\lambda^{-\sigma_3/4}}{\sqrt{2}}
  \begin{pmatrix}
   1&1\\1&-1
  \end{pmatrix}\e^{\theta\sigma_3}.$$
 \end{enumerate}

\noindent {\textbf{Shifted RH problem.}}
Another way is to move the intersection point of the contour $\Sigma$ from $\lambda=0$ to some point $\lambda=\lambda_0<<0.$ Indeed, as follows from Lemma \ref{lem_b_a_prop}, roots of $r_u(\lambda)=\i$ can accumulate only along the rays $\arg\lambda=\frac{5\pi}{7}, \frac{3\pi}{7}, \frac{\pi}{7},$
and hence, if we move the intersection point of the contour $\Sigma$ from $0$ to some $\lambda_0<<0,$ and denote such a contour
by $\lambda_0+\Sigma,$ corresponding domains by $\lambda_0+\Sigma_{I, II, III, IV},$ and rays by $\lambda_0+\gamma_{0,3,-3}, \lambda+\rho,$
then for large enough negative $\lambda_0$ the region $\lambda_0+\Omega_{II}$ will not contain any roots of $r_u=\i.$

We call $\mathbb{F}_{\lambda_0}$ the function, obtained from $\mathbb{F}$ by such a shift.
It solves RHP \ref{RHF} with contour $\Sigma$ changed to $\lambda_0+\Sigma,$ and asymptotics
\begin{enumerate}
 \item [{\textit{4.}}] $$\mathbb{F}_{\lambda_0}(x,t;\lambda)=\frac{1}{\sqrt{2}}(\lambda-\lambda_0)^{-\sigma_3/4}\begin{pmatrix}1&1\\1&-1\end{pmatrix}
 \(I+\frac{\widehat b\sigma_3}{\sqrt{\lambda-\lambda_0}}+{\color{red}\color{black} {\ol{_\mathcal{O}}}(\frac{1}{\sqrt{\lambda-\lambda_0}})} \)
 \e^{\theta_{\lambda_0}(x,t;\lambda-\lambda_0)},$$
as $\lambda\to\infty,$ where \begin{equation}\label{thetalambda0}\theta_{\lambda_0}(x,t;\lambda)=\frac{1}{105}(\lambda-\lambda_0)^{\frac72}+
 \frac{\lambda_0}{30}(\lambda-\lambda_0)^{\frac52}
 +\(\frac{\lambda_0^2}{24}-\frac{t}{3}\)(\lambda-\lambda_0)^{\frac32}+\(\frac{\lambda_0^3}{48}-\frac{t\lambda_0}{2}+x\)(\lambda-\lambda_0)^{\frac12}\end{equation}
and the scalar $\widehat b=\widehat b(x,t)$ is not fixed, but determines the form of the asymptotics.
\end{enumerate}
\begin{remark}
 The function $\theta_{\lambda_0}(x,t;\lambda)$ is connected with $\theta(x,t;\lambda)$ in the following way:
 $$\theta_{\lambda_0}(x,t;\lambda) = \theta(x,t;\lambda)+\frac{h_1}{\sqrt{\lambda}}+\mathcal{O}(\lambda^{-3/2}),$$
 where \begin{equation}\label{g1} g_1=\frac{-\lambda_0^4}{384}+\frac{t\lambda_0^2}{8}-\frac{x\lambda_0}{2}.\end{equation}
 \clu{Furthermore}, quantities $\widetilde b$ and $\widehat b$ in the asymptotics for $\mathbb{F}$ and $\mathbb{F}_{\lambda_0}$
 are related as 
 $$\widehat b=\widetilde b-g_1.$$
 \end{remark}

The RHP for $\widehat{\mathbb{F}}_{\lambda_0}(x,t;\lambda)$ is the same as for $\mathbb{F}_{\lambda_0},$ but with the asymptotic condition 4. 
replaced with
 \begin{enumerate}
  \item [{\textit{4a.}}] Asymptotics as $\lambda\to\infty:$ 
  $$\widehat{\mathbb{F}}_{\lambda_0}(\lambda)=\(I+{\ol{_\mathcal{O}}}(1)\)\dfrac{
  (\lambda-\lambda_0)^{-\sigma_3/4}}{\sqrt{2}}
  \begin{pmatrix}
   1&1\\1&-1
  \end{pmatrix}\e^{\theta_{\lambda_0}\sigma_3}.$$
 \end{enumerate}

The functions $\widehat{\mathbb{F}}_{\lambda_0}$ and $\widehat{\mathbb{F}}$ are related in the following way:
for those $\lambda$ not lying between $\gamma_{\pm 3}$ and $\gamma_{\pm3}+\lambda_0,$
\begin{equation}\label{FFlambda0}\widehat{\mathbb{F}}_{\lambda_0}(x,t;\lambda)=\begin{pmatrix}1 & 0 \\ g_1 & 1\end{pmatrix}
\widehat{\mathbb{F}}(x,t;\lambda).\end{equation}



\section{Existence of solution to the RH problems}\label{sect_RHP_exist}
RHPs constructed in Section \ref{sect_RH} make sense not only for spectral 
functions $a(\lambda), b(\lambda), r(\lambda)$ associated with a compactly supported 
perturbation $u_{t_0}(x)$ of $U(x,t_0),$ but also for a wider range of functions $a(\lambda), b(\lambda), r(\lambda).$ 
We list below the properties of functions $r(\lambda),$ $b(\lambda),$ $a(\lambda),$ which are 
sufficient to consider the RH problems from Section \ref{sect_RH}, and in what follows we 
do not associate $r(\lambda), b(\lambda), a(\lambda)$ with  initial function 
$u_{t_0}(x),$ but only assume that they are three arbitrary functions, that satisfy properties 
\ref{prop_abr} (the decay property for $r(\lambda),$ listed in Properties \ref{prop_abr}, corresponds to the initial function $u_{t_0}(x),$ which is locally a function of bounded variation, and not just locally integrable).

\begin{Properties}\label{prop_abr} 
{\textbf{of spectral functions $a(\lambda),$ $r(\lambda),$ $b(\lambda)=a(\lambda)r(\lambda).$}}
\\ Define the set $\mathcal{M}$ of functions 
$\,r_u: \mathbb{R}\to\mathbb{C}\,$
 that satisfy the following properties:

\begin{enumerate}
\item $r_u$ can be extended to an entire function,
 \item $\Im r_u(s)={\color{black}\mathcal{O}(s^{-1})}\quad\mbox{ for }\ \  s\in\mathbb{R},\ \ s\to\pm\infty,$
 \item $\Im r_u(s)<\frac12$ for $s\in\mathbb{R},$
 \item $r_u(\lambda)-\ol{r_u(\ol{\lambda})}\neq\i$ for all $\lambda\in\mathbb{C}.$
\end{enumerate}
Define $r_d(\lambda):=\ol{r_u(\ol{\lambda})},$
and $a_u(\lambda),$ $a_d(\lambda)$ by formulas \eqref{au_ru}, \eqref{ad_ru}, set 
$$b_u(\lambda)=a_u(\lambda) r_u(\lambda),\quad b_d(\lambda)=a_d(\lambda) r_d(\lambda),$$ 
and define functions $a(\lambda), b(\lambda), r(\lambda)$ by formulas \eqref{a_auad}, \eqref{b_bubd}, \eqref{r_rurd}.
Then obviously, by the Sokhotsky-Plemelji formula, $a_u(\lambda), a_d(\lambda), b_u(\lambda), b_d(\lambda)$ are entire functions 
satisfying the symmetry condition
$$a_u(\lambda)=\ol{a_d(\ol{\lambda})},\quad b_u(\lambda)=\ol{b_d(\ol{\lambda})},
\quad
r_u(\lambda)=\ol{r_d(\ol{\lambda})},$$ and multiplying \eqref{au_ru} and \eqref{ad_ru} we see that  relation \eqref{b_a_rel} holds.
\\\clu{Furthermore}, the defining properties of the set $\mathcal{M}$ are: 
there exist real $t_0, A<B$ such that
 \begin{enumerate}\setcounter{enumi}{4}
 \item $r_u(\lambda)={\color{black}\mathcal{O}(\frac{1}{\lambda})}\e^{2\theta(B,t_0,\lambda)}\ \ $ uniformly w.r.t. 
 $\arg\lambda\in[0,\pi],$
 \item $a_d-a_u={\color{black}\mathcal{O}(\frac{1}{\lambda})}\e^{-2\theta(A,t_0;\lambda)}\ \ $ uniformly w.r.t. 
 $\arg\lambda\in[-\pi,0]\cup[0,\pi],$
\end{enumerate}
\end{Properties}
\begin{remark}
Lemmas \ref{lem_a_prop}, \ref{lem_b_a_prop}, \ref{lem_r_refined} show, that the spectral function  $a,b,r$ associated with compactly supported perturbation $u_{t_0}(x)\in BV_{loc}$ of $U(x,t_0),$ satisfy properties \ref{prop_abr}, and hence $r_u$ belongs to $\mathcal{M}.$
\end{remark}

\begin{theorem}\label{teor_ExistF}
 Let functions $a(\lambda), b(\lambda), r(\lambda)$ satisfy properties \ref{prop_abr}. Then the RHP \ref{RHF} for the function $\widehat{\mathbb{F}}(x,t;\lambda)$ with asymptotic condition [4] replaced by [4a] has a unique solution.
\end{theorem}

\begin{proof}
The uniqueness part of the theorem is obvious.
For quite general RHPs, the scheme how to prove their solvability was introduced by \cite{Zhou89}. 
For some particular cases, the scheme was realised with all the necessary details by 
\cite[Theorems 5.3, 5.6, p.1387--1406, steps 1,2,3]{DKMVZ},
for RHP \ref{RHFhat} in the particular case 
$r\equiv 0, a\equiv 1$ by \cite{Claeys Vanlessen}, and in other situation by \cite[section 2.3, p.18]{IKO}.
The main distinction of our case is that the jump matrix $J_L$ (defined in Preparatory Step 3 below) does not 
converge exponentially fast to the identity as $\lambda\to\infty.$ Indeed, $J_L-I$ is exponentially small as 
$\mu\to\infty$ along $\gamma_{0},$ $\gamma_{3},$ $\gamma_{-3},$ but not across $\rho.$

The general scheme from \cite{Zhou89} consists of 3 steps: after reformulating the RHP as a singular integral equation 
(SIE),
Step 1 is to prove that the corresponding singular integral operator is Fredholm. Step 2 is to show 
that the index of that operator is $0.$ Step 3 is to show that the kernel of the corresponding operator is 0. This proves the invertibility of 
the operator.
Before applying this scheme, we need to make 3 more preparatory steps: Preparatory	 Step 1 is to get rid of poles, Preparatory Step 2 is 
to make identity asymptotics at infinity, and Preparatory Step 3 is to make jump equals $I$ at all the junctions of the contour.


\textbf{Preparatory Step 1: Get rid of poles.} To this end we consider the shifted RHP for the 
$\widehat{\mathbb{F}}_{\lambda_0}$
with large enough negative $\lambda_0.$ 
Denote $$\mu=\lambda-\lambda_0.$$
RHP \ref{RHFhat} for $\widehat{\mathbb{F}}$ and the one for $\widehat{\mathbb{F}}_{\lambda_0}$ 
are equivalent to each other.
Indeed, for $\lambda$ not lying between $\gamma_{\pm 3}$ and $\gamma_{\pm 3}+\lambda_0,$ they are related by formula
\eqref{FFlambda0}. 

\textbf{Preparatory Step 2: To make identity asymptotics at infinity.}
To this end we use a slight modification of the function $Z(.)$ defined in \eqref{Z}, in which we change 
the rays from $\arg z=\pm\frac{2\pi}{3}$ to $\arg\mu=\pm\frac{6\pi}{7},$ but keep the same notation $Z(.).$

Now define
\begin{equation}\label{Y}Y(\mu)=Z(\mu)\e^{-\frac23\mu^{2/3}\sigma_3}, \ \ \mbox{ hence }\ \ 
Y(\mu)=\begin{pmatrix}
            1+\mathcal{O}(\mu^{-3}) & \mathcal{O}(\mu^{-2}) \\ \mathcal{O}(\mu^{-1}) & 1+\mathcal{O}(\mu^{-3})
           \end{pmatrix}\mu^{-\sigma_3/4} \begin{pmatrix}1&1\\1&-1\end{pmatrix}, \mu\to\infty,
\end{equation}
and define
$$M(x,t;\mu)=\widehat{\mathbb{F}}_{\lambda_0}(x,t; \lambda)\e^{-\theta_{\lambda_0}(x,t;\mu)}Y^{-1}(\mu).$$
The function $M(\mu)=M(x,t; \mu)$ solves the following RHP:

\begin{RH}\label{RHM}
To find a $2\times2$ matrix-valued function $M(x,t;\mu),$ which
\begin{enumerate}
\item is analytic in $\mu\in\mathbb{C}\setminus\Sigma,$ 
where $\Sigma=\mathbb{R}\cup\gamma_3\cup\gamma_{-3};$ 
\item has the following jump $M_+(\mu)=M_-(\mu)J_{M}(\mu)$ across $\Sigma:$ $J_M=$\\
$\begin{cases}
                  Y_-\begin{pmatrix}
                   \frac{1}{(1+\i r_u)(1-\i r_d)} & \frac{r_u\ \e^{-2\theta_{\lambda_0}^+}}{1-\i r_d} \\\\ 
                   \frac{-r_d\ \e^{-2\theta_{\lambda_0}^-}}{1+\i r_u} & \frac{1}{a_ua_d}
                  \end{pmatrix}Y_-^{-1}, \mu\in\rho,
\
Y_-\begin{pmatrix}
 1 & 0 \\ \frac{-\i\e^{-2\theta_{\lambda_0}}}{a_ua_d}+\i\e^{-\frac43\mu^{\frac32}} & 1
\end{pmatrix}Y_-^{-1}, \ \mu\in\gamma_0,
\\\\
\
Y_-\begin{pmatrix}
 1 & \frac{\i\e^{2\theta_{\lambda_0}}}{1+\i r_u} -\i\e^{\frac43\mu^{3/2}} \\ 0 & 1
\end{pmatrix}Y_-^{-1}, \ \mu\in\gamma_3,
\qquad 
Y_-\begin{pmatrix}
 1 & \frac{\i\e^{2\theta_{\lambda_0}}}{1-\i r_d} -\i\e^{\frac43\mu^{3/2}} \\ 0 & 1
\end{pmatrix}Y_-^{-1}, \ \mu\in\gamma_{-3},
\end{cases}
$
\\where $r_u=r_u(\lambda)=r_u(\mu+\lambda_0),$ and the same is for $r_d, a_u, a_d,$
and $\theta_{\lambda_0}=\theta_{\lambda_0}(\mu);$
\item has the following asymptotics as $\mu\to\infty,$ which is uniform w.r.t. $\arg\mu\in[-\pi,0]\cup[0,\pi]:$
$$M(x,t;\mu)=I+{\ol{_\mathcal{O}}}(1).$$
\end{enumerate}
\end{RH}
\textbf{Preparatory Step 3: make the jump at the junction $\mu=0$ equals $I.$}
This can be done due to the identity product 
$$J_M|_{\gamma_{-3}}(0)\cdot J_M|_{\rho}(0)\cdot J_M|_{\gamma_{3}}(0)\cdot J_M^{-1}|_{\gamma_{0}}(0)=I.$$
Indeed, pick up 4 points $\mu_j\notin\Omega_j,$ $j=I,II,III,IV,$ and define
$$L(x,t;\mu)=M(x,t;\mu)G,\qquad G=G_j, \mu\in\Omega_j, \qquad G_j=I+\frac{-\mu_j}{\mu-\mu_j}(B_j(x,t)-I),$$
where the matrices $B_j(x,t)$ are to be determined.
The jumps for the function $L$ are: $L_+=L_-J_L,$ where $J_L=G_-^{-1}J_MG_+,$ and
$$J_L|_{\gamma_3}=G_{II}^{-1}J_M|_{\gamma_3}G_I,\ \ 
J_L|_{\gamma_{-3}}=G_{IV}^{-1}J_M|_{\gamma_3}G_{III},\ \ 
J_L|_{\gamma_0}=G_{IV}^{-1}J_M|_{\gamma_0}G_I,\ \ 
J_L|_{\rho}=G_{III}^{-1}J_M|_{\rho}G_{II}.$$
We can take 
$$B_I=I,\ \ B_{II}=J_M|_{\gamma_3}(0),\ \ B_{III}=J_M|_{\rho}(0)J_M|_{\gamma_3}(0),\ \ 
B_{IV}=J_M|_{\gamma_{-3}}(0)\cdot J_M|_{\rho}(0)\cdot J_M|_{\gamma_3}(0),$$
and the jump matrix $J_L$ equals $I$ at the origin $\mu=0$ on every ray $\gamma_{\pm3},$ $\gamma_0,$ $\rho.$

\noindent \textbf{Reformulation of RHP for $L$ as a SIE.}

The RH problem for $L$ is equivalent to the 
following singular integral equation (SIE)
\begin{equation}\label{SIE} L_-=I+\mathcal{C}_-(L_-(J_L-I)),\quad \mbox{or}\quad \left[Id-\mathcal{C}_-(.(J_L-I))\right]\circ (L_--I)=\mathcal{C}_-((J_L-I))\end{equation}
where $\mathcal{C}_{\pm}$ are the Cauchy operators acting in $L_p(\Sigma),$ $p>1,$ which for a H\"{o}lder continuous function $f$ act as 
$$\mathcal{C}_{\pm}f(\lambda)=\frac{1}{2\pi\i}\int\limits_{\Sigma}\frac{f(s)\ \d s}{(s-\lambda)_{\pm}}=\pm\frac{f(\lambda)}{2}+\frac{1}{2\pi\i}
p.v.\int
\limits_{\Sigma}\frac{f(s)\ \d s}{s-\lambda}.$$
Since $J_L-I\in L_{\infty}(\Sigma),$ then $\mathcal{C}_{J_L}:=\mathcal{C}_-(.(J_L-I))$ is an operator acting in 
$L_p(\Sigma).$
Define also the (Hilbert) operator $$H f(\lambda):=\frac{1}{2\pi\i}p.v.\int\limits_{\Sigma}\dfrac{f(s) \d s}{s-\lambda}.$$
By Sokhotsky--Plemelj formula, $$\mathcal{C}_+=\frac12 Id+H,\quad \mathcal{C}_-=-\frac12 Id+H.$$

\noindent {\color{black} For the fact that $\mathcal{C}_{\pm}$ are operators acting in $L_2(\Sigma)$ or $L_p(\Sigma)$ for contours with self-intersections we refer to \cite[section 4.4, p.137, Theorem 4.15]{BK97}, see also \cite[section 2.2]{L18}.}

\noindent If $L_-\in I+L_p(\Sigma)$ (i.e. $L_--I\in L_p(\Sigma)$) is the solution of \eqref{SIE}, then the function 
$$L=I+\mathcal{C}(L_-(J_L-I)),\quad \mbox{ where } (\mathcal{C}f)(\mu):=\frac{1}{2\pi\i}
\int\limits_{\Sigma}\frac{f(s)\ \d s}{s-\mu},$$ 
is the solution of RHP \ref{RHM}.
Indeed, using the Sokhotsky-Plemelji formula $\mathcal{C}_+-\mathcal{C}_-=I,$ one obtains
$$L_+=I+\mathcal{C}_+(L_-(J_L-I))=I+L_-(J_L-I)+\mathcal{C}_-(L_-(J_L-I))=L_-(J_L-I)+L_-=L_-J_L.$$
This solution $L$ is to be understood in the $L_p$ sense (see \cite[p. 1388]{DKMVZ} for more details).
However, using local analyticity of the jump matrix $J_L,$ one can obtain that $L$ is 
a solution of the RHP also in the sense of continuous boundary values. 

In order to prove that equation \eqref{SIE} has a solution $L_-\in I+L_p(\Sigma),$ i.e. $L_--I\in L_p(\Sigma),$ 
following the approach from 
\cite{Zhou89}, \cite[pp.1387--1395]{DKMVZ}, 
it suffices to show that the operator
$$\mathcal{C}_{J_L}=\mathcal{C}_-(.(J_L-I))$$ is invertible in $L_p(\Sigma).$
This consists of 3 steps:
step 1 is to show that $I-\mathcal{C}_{J_L}$ is a 
Fredholm operator, step 2 is to show that the index of $I-\mathcal{C}_{J_L}$ is 0, 
and step 3 is to show that the kernel of $I-\mathcal{C}_{J_L}$ is $\left\{0\right\}.$

\textbf{Step 1.} To show that $I-\mathcal{C}_{J_L}$ 
is a Fredholm operator, one can show that it has a pseudo-inverse, 
i.e. there exists an operator $O$ such that $O(I-\mathcal{C}_{J_L})-I$ and $(I-\mathcal{C}_{J_L}))O-I$ 
are compact operators (see, for instance, \cite[Prop 3.3.11, p.109--110]{Pedersen}).
Let us take $$O=I-\mathcal{C}_{J_L^{-1}}:=I-\mathcal{C}_-(.(J_L^{-1}-I)).$$
Then similarly as in \cite[Step 1, p.1389]{DKMVZ} one can show that
$$(I-\mathcal{C}_{J_L})(I-\mathcal{C}_{J_L^{-1}})f=f+\mathcal{C}_-\left[\mathcal{C}_+(f\widetilde w)w\right],
\quad \mbox{where}\quad w=J_L-I,\ \widetilde w=J_L^{-1}-I$$
and 
$$(I-\mathcal{C}_{J_L^{-1}})(I-\mathcal{C}_{J_L})f=f+\mathcal{C}_-\left[\mathcal{C}_+(f w)\widetilde w\right].$$
Since both the operators $(I-\mathcal{C}_{J_L^{-1}})(I-\mathcal{C}_{J_L}),$ 
$(I-\mathcal{C}_{J_L})(I-\mathcal{C}_{J_L^{-1}})$
are of the same form, it suffices to prove compactness for one of it. To prove that the operator 
$$
K:f\mapsto \mathcal{C}_-\left[\mathcal{C}_+(f\widetilde w)w\right]
$$
is compact,
we follow the approach of \cite[p. 1400-1401]{DKMVZ}, i.e. approximate the continuous function  $w$ 
(it is continuous at the origin because of our Preparatory Step 3) by rational 
functions $w_{\varepsilon}$, $$ || w_L-w_{\varepsilon} ||_{L^{\infty}}<\varepsilon$$ for any positive $\varepsilon.$
Then $$K_{\varepsilon}\to K \mbox{ as } \varepsilon\to 0,$$ where $K_{\varepsilon}$ is defined by almost the same formula as 
$K,$ but with $w$ replaced with $w_{\varepsilon}.$ Hence, it is enough to show that $K_{\varepsilon}$ is compact for every $\varepsilon.$
If $$w_{\varepsilon}(\mu)=\sum\limits_{\nu}\dfrac{\alpha_{\nu}}{\mu-\mu_{\nu}},$$ then in the same way as in \cite[p. 1401]{DKMVZ} one shows that for a weakly convergent to $0$ 
sequence $f_n\in L_2(\Sigma_L),$ the sequence
$$(K_{\varepsilon}f_n)(\mu)=\sum\alpha_{\nu}\dfrac{\mathcal{C}[f_n \widetilde w_{L}](\mu_{\nu})}{\mu_{\nu}-\mu}$$
strongly converges to 0.



\textbf{Step 2.} The proof that the index of $I-\mathcal{C}_{J_M}$ is 0 is the same as in \cite[p.1390]{DKMVZ}, and is based on the fact that 
$I-s \mathcal{C}_{J_M}$ is Fredholm for all scalar $s$, which can be proved in the 
same way as in Step 1, and then using 
continuity of index and the fact that the identity operator has index 0.

\textbf{Step 3.}
To prove that $\ker(I-\mathcal{C}_{J_M})=0,$ we take any element from the kernel, i.e.
$L_{0,-}\in L_2(\Sigma)$
such that 
$$
L_{0,-}-\mathcal{C}_-[L_{0,-}(J_L-I)]=0,
$$
and we will show that $L_{0,-}=0.$ Indeed, define the function $$L_{0}(\mu)=\mathcal{C}[L_{0,-}(J_L-I)](\mu)=\frac{1}{2\pi\i}\int\limits_{\Sigma_L}\frac{L_{0,-}(s)(J_L(s)-I)\ \d s}{s-\mu}.$$
It satisfies the following RH problem:
\begin{enumerate}
 \item $L_{0}(\mu)$ is analytic in $\mu\in\mathbb{C}\setminus\Sigma_L;$
 \item $L_{0,+}(\mu)=L_{0,-}(\mu)J_L(\mu),$ $\mu\in\Sigma,$
 \item $L_{0}(\mu)=\mathcal{O}(\mu^{-1}),$ $\mu\to\infty.$
\end{enumerate}

The function $L_0(\mu)$ satisfies the above RH problem both in $L^2$ sense and, using local 
analyticity of the jump matrices, as in \cite[p. 1402, Proposition 5.7]{DKMVZ}, 
in continuous sense.

So, $L_0$ solves the RH problem with the same jumps as $L,$ but with zero asymptotics at infinity. We need to show that $L_0\equiv 0.$

Making all the transformations which led us from $\widehat{\mathbb F}_{\lambda_0}(x,t;\mu)$ to $L$ 
in the reverse order, starting from $L_0,$ we come to the following RH problem for a function 
$\widehat{\mathbb{F}}_{\lambda_0,0}(x,t;\mu):$

\begin{RH}\begin{enumerate}
 \item [1,2.] Analyticity, jumps are as in RHP for $\widehat{\mathbb F}_{\lambda_0}(x,t;\mu)$.
 \item [3b.] Asymptotics at infinity: $$\widehat{\mathbb{F}}_{\lambda_0, 0}(\mu)=
 \mathcal{O}(\frac{1}{\mu})\mu^{-\sigma_3/4}\begin{pmatrix}1 & 1 \\1 & -1\end{pmatrix}\e^{\theta_{\lambda_0}\sigma_3}.$$
\end{enumerate}\end{RH}
We thus need to show that $\widehat{\mathbb{F}}_{\lambda_0, 0}(\mu)\equiv 0.$

To prove the latter, 
define the matrix
$$
A(\mu)=\widehat{\mathbb{F}}_{\lambda_0,0}\e^{-\theta_{\lambda_0}\sigma_3}
\begin{cases}
I, & \quad  \arg(\mu)\in(\frac{-6\pi}{7},0),\\
\begin{pmatrix}
 1 & \frac{-\i \e^{2\theta}}{1-\i r_d}\\0&1
\end{pmatrix}, & \quad \arg(\mu)\in(-\pi, \frac{-6\pi}{7}),\\
\begin{pmatrix}
 1 & \frac{\i \e^{2\theta}}{1+\i r_u}\\0&1
\end{pmatrix}\begin{pmatrix}
 0& \i\\\i&0
\end{pmatrix}, & \quad \arg(\mu)\in(\frac{6\pi}{7},\pi),\\
\begin{pmatrix}
 0& \i\\\i&0
\end{pmatrix}, &\quad  \arg(\mu)\in(0, \frac{6\pi}{7},0).\\
\end{cases}
$$


\clu{Furthermore}, the function $A(\mu)$ has asymptotics $A(\mu)=\mathcal{O}(\mu^{-3/4})$ as $\mu\to\infty,$
and jumps only on the real axis, reading
$$A_+(\mu)=A_-(\mu)J_A(\mu),\quad J_A(\mu)=\begin{cases}
                                                            \begin{pmatrix}
                                                             0 & \i \e^{\theta_{\lambda_0,-}-\theta_{\lambda_0,+}}\\ 
                                                             \i \e^{\theta_{\lambda_0,+}-\theta_{\lambda_0,-}} 
                                                             & \frac{1}{a_ua_d}
                                                            \end{pmatrix},\mu<0,
\\
                                                            \begin{pmatrix}
                                                             0 & \i \\ \i  & \frac{1}{a_ua_d}
                                                            \end{pmatrix},\mu>0.
                                                           \end{cases}.
$$
Now, integrating $A(\mu)\ol{A^T(\ol{\mu})}$ over $\mathbb{R}+\i 0,$ and adding the result to its Hermite conjugate, 
we obtain
$$0=\int\limits_{\mathbb{R}}A_+(\mu)\ol{A^T_+(\ol{\mu})}\d\mu=\int\limits_{\mathbb{R}}A(\mu-\i 0)J_A(\mu)\ol{A^T(\mu-\i 0)}\d \mu,$$
$$0=\int\limits_{\mathbb{R}}A(\mu-\i 0)\ol{J^T_A(\mu)}\ol{A^T(\mu-\i 0)}\d \mu,$$
since $J_A(\mu)+\ol{J^T_A(\mu)}=\begin{pmatrix}
                                         0&0\\0&\frac{2}{a_ua_d}
                                        \end{pmatrix},
$
and $a_u(\mu)a_d(\mu)>0$ for $\mu\in\mathbb{R},$ we conclude that 
the second column of $A_-$ is identically 0. Then, from the jump relations for $A$ one concludes that also the first column of $A_+$ is identically 0.
Furthermore,  for the rest of the elements in $A$ we obtain, using the jump conditions,
that 
$$\begin{cases}A_{12}^+=\i A_{11}^-,\ \mu\in(0, +\infty),\\
A_{12}^+=\i A_{11}^-\e^{\theta_{\lambda_0,-}-\theta_{\lambda_0,+}},\ \mu\in(-\infty, 0).
\end{cases}
$$
and $A_{12}^+=\mathcal{O}(\mu^{-3/4}),$ $A_{11}^-=\mathcal{O}(\mu^{-3/4}).$
Hence, the function $$f(\mu):=\begin{cases}
                            A_{12}^+ \e^{\theta_{\lambda_0}},\Im\mu>0,\\\i A_{11}^-\e^{\theta_{\lambda_0}}, \Im \mu<0
                           \end{cases}
$$
is entire and has uniform asymptotics (compare with the function $\mathcal{E}$ defined by formula \eqref{E_function})
\begin{equation}\label{f_vanishing}f(\mu)=\mathcal{O}(\mu^{-3/4})\e^{\theta_{\lambda_0}(\mu)},\quad \mu\to\infty.\end{equation}
\begin{prop}\label{prop_vanish}
An entire function $f(\mu),$ which has the asymptotics \eqref{f_vanishing} uniformly w.r.t. $\arg\mu\in[-\pi,\pi],$ equals 0 identically.
\end{prop}
\begin{proofof}{\textit{of Proposition \ref{prop_vanish}}}
The proof  is very similar to the one in \cite[p.18]{IKO}, \cite[p.1395]{DKMVZ}, but for the convenience of the reader we give it here.
First we recall the following Carlson theorem, which is a variant of the maximum modulus principle:
\begin{theorem}\label{theor_Carlson} \cite[p.236]{RS}
Let $b(z)$ be a function, holomorphic in $\Re z>0$ and continuous up to $\Re z\geq 0.$ Let $|b(z)|\leq M \e^{A|z|}$ for $\Re z\geq 0$ and $|b(\i y)|\leq M \e^{-B|y|}$ for $y\in\mathbb{R}.$ Then $b(z)\equiv 0.$
\end{theorem}
Now we define the function $h(\mu)=f(\mu)\e^{-\theta_{\lambda_0}(\mu)},$ which is discontinuous  across the \clu{half-line} $\mu\in(-\infty,0],$ where it has the jump 
\begin{equation}\label{jump_h}h(\mu+\i 0)=h(\mu-\i 0)\e^{2\theta_{\lambda_0}(\mu-\i 0)},
\end{equation} 
and define the new variable $$\zeta=\sqrt{\mu},$$ and the new function 
$$\widetilde{h}(\zeta)=\begin{cases}h(\zeta^2), \Re\zeta>0, \\ h(\zeta^2)\e^{-\theta_{\lambda_0}(\zeta^2)}, \Re \zeta<0.\end{cases}$$
Despite the uniform definition for $\Re\zeta<0,$ the function $\widetilde h(\zeta)$ is discontinuous across the \clu{half-line} $\zeta\in(-\infty,0],$ but because of the jump \eqref{jump_h}, it is continuous across the imaginary line $\zeta\in\mathbb{R}.$
Now we introduce the variable $$z=\zeta^{7/8},\quad \zeta=z^{8/7},$$ with the standard cut across $\zeta\in(-\infty,0],$ $z\in(-\infty,0],$ and consider the function
$$\widehat h(z)=\widetilde h(z^{8/7}).$$
The function $\widehat h(z)$ is continuous and bounded in $\Re z\geq 0,$ analytic in $\Re z>0,$ and for $z\in\i\mathbb{R}$ it has the super exponential decay (with some positive $c, C>0$)
$$|h(z)|\leq C \e^{-c z^8}.$$
Hence, by the Carlson's Theorem \ref{theor_Carlson}, $h(z)\equiv0$ for $\Re z\geq 0,$ and hence $f(\mu)\equiv 0$ for any $\mu\in\mathbb{C}.$
\end{proofof}

This finishes the proof that $\widehat{F}_{\lambda_0, 0}(\mu)$ is identically 0, and hence that $L_0(\mu)$ is identically 0, and thus that RHP 
\ref{RHF} for the function $\widehat{\mathbb{F}}(x,t;\lambda)$ with asymptotic condition [4] replaced with [4a] is solvable in the $L_2$ sense. Then using analyticity of the jump matrices, we can argue that this solution is indeed a solution in the continuous sense.
\end{proof}

\begin{corollary}\label{cor_Cinfty}
If a function $r_u(s)$ vanishes as $\mathcal{O}(|s|^{-1-N/2}), s\to-\infty,$ then the function $F(x,t;\lambda)$ is $N$ times differentiable in $x,$ and $\lfloor N/3 \rfloor$ times differentiable in $t.$
\end{corollary}
\begin{proof}
Every time that we differentiate the jump matrix $J_M$  w.r.t.  $x,$ we gain a factor $\sqrt{\lambda}$ under the \clu{symbol} of the integral. 
Every time that we differentiate $J_M$  w.r.t.  $t,$ we gain $\lambda^{3/2}.$ Thus, in order to have convergent integrals in \eqref{SIE}, we need to have stronger vanishing of $r_u(s), s\to-\infty.$

Let us consider differentiation w.r.t. $x$ and $N\geq 1.$ Then the derivative  w.r.t.  $x$ of the jump matrix $J_x$ is still in $L_2(\Sigma_M),$ 
and differentiating equation \eqref{SIE}  w.r.t.  $x,$ we obtain
\begin{equation}\label{SIE_x}
M_{-,x}=\mathcal{C}_-(M_{-,x}(J_M-I))+\mathcal{C}_-(M_-J_{M,x}), \ \ \mbox{ or }\ \  [Id-\mathcal{C}_-(.(J_M-I))]\circ(M_{-,x})=\mathcal{C}_-(M_-J_{M,x})\end{equation}
and this equation is of the same kind as \eqref{SIE}, just with a different right-hand-side. Since the operator $Id-\mathcal{C}_-(.(J_M-I))$ in \eqref{SIE} is invertible, the same is true for \eqref{SIE_x}.
\end{proof}

\begin{corollary}
\begin{enumerate}
\item RHP \ref{RHtp} has a unique solution $\widehat{P}(x,t;\lambda)$ for $t>t_0, x\in\mathbb{R}$ and for $t=t_0, x\geq B,$
\item RHP \ref{RHtn} has a unique solution $\widehat{N}(x,t;\lambda)$ for $t<t_0, x\in\mathbb{R}$ and for $t=t_0, x\leq A,$ 
\item Further, for $t>t_0, x\in\mathbb{R}$ and for $t=t_0, x>B$ the function $\widehat{P}(x,t;\lambda)$ admits
the following full asymptotic expansion:
$$\widehat P(x,t;\lambda)=\frac{1}{\sqrt{2}}\(I+\sum\limits_{j=1}^{\infty}P_j(x,t)\lambda^{-j}+\mathcal{O}(\lambda^{-\infty})\)\lambda^{-\sigma_3/4}\begin{pmatrix}1&1\\1&-1\end{pmatrix},$$
and for $t<t_0, x\in\mathbb{R}$ and for $t=t_0, x<A$ the function $\widehat{N}(x,t;\lambda)$ admits
the following full asymptotic expansion:
$$\widehat N(x,t;\lambda)F^{\sigma_3}(\lambda)=\frac{1}{\sqrt{2}}\(I+\sum\limits_{j=1}^{\infty}N_j(x,t)\lambda^{-j}+\mathcal{O}(\lambda^{-\infty})\)\lambda^{-\sigma_3/4}\begin{pmatrix}1&1\\1&-1\end{pmatrix},$$
where the scalar function $F$ is given by
\begin{equation}\label{F}F(\lambda)=\exp\left\{-\dfrac{\sqrt{\lambda}}{2\pi}\int\limits_{-\infty}^{0}\dfrac{\ln|a_u(s)|^2\ \d s}{(s-\lambda)\ \sqrt{|s|}}\right\}.\end{equation}
\item $\widehat P(x,t;\lambda),$ $\widehat N(x,t;\lambda)$ are smooth (infinitely many times differentiable)  in $x,t.$
\item the functions
\begin{equation}\label{PN_PNhat}
 \begin{split}
P(x,t;\lambda)=\begin{pmatrix}1&0\\(P_1(x,t;\lambda))_{12} & 1\end{pmatrix}\widehat{P}(x,t;\lambda),  
\\
N(x,t;\lambda)=\begin{pmatrix}1&0\\(N_1(x,t;\lambda))_{12} & 1\end{pmatrix}\widehat{N}(x,t;\lambda)
 \end{split}
\end{equation}
solve RHPs \ref{RHtp}, \ref{RHtn}, respectively.
\end{enumerate}
\end{corollary}

\begin{proof}$[1.,2.]$
Let us define the functions $\widehat P,$ $\widehat N$ by the following formulas:
\begin{equation}\label{P_F}\widehat P(x,t;\lambda) = \widehat{\mathbb{F}}\e^{-\theta\sigma_3}
\begin{cases}
\begin{pmatrix}
\frac{1}{1+\i r_u} & 0 \\ -r_u\e^{-2\theta} & 1+\i r_u
\end{pmatrix},  \Omega_{II},
\\
\begin{pmatrix}
\frac{1}{1-\i r_d} & 0 \\ -r_d\e^{-2\theta} & 1-\i r_d
\end{pmatrix},  \Omega_{III},
\\ I ,\quad \mbox{ elsewhere, }
\end{cases}
\widehat N(x,t;\lambda) = \widehat{\mathbb{F}}\e^{-\theta\sigma_3}
\begin{cases}
\begin{pmatrix}
a_u & \frac{b_d\e^{2\theta}}{1+i r_u} \\ 0 & \frac{1}{a_u}
\end{pmatrix}, \Omega_{II},
\\
\begin{pmatrix}
a_d & \frac{b_u\e^{2\theta}}{1-i r_d} \\ 0 & \frac{1}{a_d}
\end{pmatrix}, \Omega_{III},
\\ I ,\quad \mbox{ elsewhere.}
\end{cases}\end{equation}
It is straightforward to check that the function $\widehat{P}$ \eqref{P_F} satisfies RHP \ref{RHtp}, and the function 
$\widehat{N}$ satisfies RHP \ref{RHtn}, and that all the possible poles of $\widehat{\mathbb{F}}$ disappear for $\widehat P,$ $\widehat N.$ One needs to be careful however that the asymptotics of $\widehat P$ as $\lambda\to\infty$ are not spoiled by the term
$$r_u(\lambda)\e^{-2\theta(x,t;\lambda)}=\mathcal{O}(\frac{1}{\lambda})\e^{2\theta(B,t_0;\lambda)-2\theta(x,t;\lambda)}
=\mathcal{O}(\frac{1}{\lambda})\e^{\frac{2(t-t_0)}{3}\lambda^{3/2}+2(B-x)\lambda^{1/2}},$$
which is bounded in $\Omega_{II}$ for $t>t_0, x\in\mathbb{R}$ and for $t=t_0, x\geq B,$ and 
 that the asymptotics of $\widehat N$ as  $\lambda\to\infty$ are not spoiled by the term
$$\frac{b_d\e^{2\theta}}{1+\i r_u}=\(-\i a_d+\dfrac{\i}{a_u+\i b_u}\)\e^{2\theta}=\mathcal{O}(\frac{1}{\lambda})\e^{2\theta(x,t;\lambda)-2\theta(A,t_0;\lambda)}=
\mathcal{O}(\frac{1}{\lambda})\e^{\frac{2(t_0-t)}{3}\lambda^{3/2}+2(x-A)\lambda^{1/2}},$$
which is bounded for $\lambda\in\Omega_{II}$ for $t<t_0, x\in\mathbb{R}$ and for $t=t_0, x\leq A.$


$[3.]$ The statement concerning the full asymptotic expansion is similar to \cite[lemma 2.3, (ii), p. 1168]{Claeys Vanlessen}. Now we proceed to the details.
\noindent Let us observe that the function $$M_{P}(x,t;\lambda)=\widehat{P}(x,t;\lambda)Y^{-1}(\lambda),$$ where $Y(\lambda)$ is defined in \eqref{Y}, solves the following RH problem:
\begin{RH}\label{RHMP}
To find a $2\times2$ matrix-valued function $M_P(x,t;\lambda),$ which
\begin{enumerate}
\item is analytic in $\lambda\in\mathbb{C}\setminus\Sigma,$ $\Sigma=\gamma_0\cup\gamma_r\cup\gamma_{-3}\cup\rho,$
\item has the following jump $M_{P,+}=M_{P,-}J_{M_P}$ accros $\Sigma:$
$$J_{M_P}=\begin{cases}
                  I, \ \lambda\in\rho,
\qquad
Y_-\begin{pmatrix}1 & 0 \\ \frac{-\i \e^{-2\theta}}{a_ua_d}+\i\e^{-\frac43\lambda^{3/2}} & 1\end{pmatrix}
Y_-^{-1}, \ \lambda\in\gamma_0,
\\\\
Y_-\begin{pmatrix}1+\i r_u & \i\e^{2\theta} \\ r_u\e^{-2\theta} & 1\end{pmatrix}
\begin{pmatrix} 1 & -\i \e^{\frac43\lambda^{3/2}} \\0&1\end{pmatrix}
Y_-^{-1}, \ \lambda\in\gamma_3,
\\\\
Y_-\begin{pmatrix}1 & \i\e^{2\theta} \\ -r_d\e^{-2\theta} & 1-\i r_d\end{pmatrix}
\begin{pmatrix} 1 & -\i \e^{\frac43\lambda^{3/2}} \\0&1\end{pmatrix}
Y_-^{-1}, \ \lambda\in\gamma_{-3},
\end{cases}
$$
\item has the following asymptotics as $\lambda\to\infty,$ which is uniform w.r.t. $\arg\lambda\in[-\pi, \pi]:$
$$M_P(x,t;\lambda)=I+\ol{_\mathcal{O}}(1).$$
\end{enumerate}
\end{RH}
We see that for $t>t_0, x\in\mathbb{R}$ and for $t=t_0,$ $x>B$ the jumps for $M_P$ are exponentially close to $I$ on the infinite parts of the contour $\Sigma,$ and hence from the SIE of the type \eqref{SIE}, which now reads as
$$M_{P,-}=I+\mathcal{C}_{\Sigma,-}\(M_{P,-}(J_{M_P}-I)\),$$
and from the representation of $M_P$ in terms of $M_{P,-}$
$$M_{P}=I+\mathcal{C}_{\Sigma}\(M_{P,-}(J_{M_P}-I)\)$$
we obtain that $M_P$ possesses the full asymptotic expansion
$$M_{P}=I+\sum\limits_{j=1}^{\infty} M_{P,j}\lambda^{-j}+\mathcal{O}(\lambda^{-\infty}).$$


\noindent Coming to $\widehat N(x,t;\lambda),$ we observe 
that the jump matrix for $\widehat N$ on $\lambda\in\rho$ is not exactly the same as the one of $Y$ \eqref{Y}, and we hence first need to transform it to such a jump. In order to do this, we make a transformation
$$\widehat{N}^{(1)}(x,t;\lambda)=N(x,t;\lambda)F^{\sigma_3}(\lambda),
\qquad \widehat{N}^{(1)}_+=\widehat{N}^{(1)}_- J^{(1)}_{N},\quad 
J^{(1)}_{N}=F_-^{-\sigma_3}(\lambda)J_{N}F_+^{\sigma_3},$$
where the scalar function $F$ is analytic in $\lambda\in\mathbb{C}\setminus (-\infty,0]$ solves the scalar conjugation problem
$$F_+F_-=a_ua_d,\lambda\in(-\infty, 0),$$
has asymptotics $F\to 1$ as $\lambda\to\infty,$
and hence can be found explicitly by formula \eqref{F}.
Let us observe that 
$\quad F(\lambda)=$ 
$$\exp\hskip-1mm\left\{\hskip-1mm\frac{1}{2\pi \sqrt{\lambda}}\hskip-1mm\(\int\limits_{-\infty}^{0}\hskip-0.5mm
\frac{\ln|a_u(s)|^2 \d s}{\sqrt{|s|}}
\hskip-0.5mm
-
\hskip-1mm
\int\limits_{-\infty}^{0}
\hskip-1.5mm
\frac{s\ln|a_u(s)|^2 \d s}{(s-\lambda)\sqrt{|s|}}\)\right\}
\hskip-0.5mm
=\hskip-0.5mm
\exp\hskip-0.5mm
\left\{
\hskip-0.5mm
\frac{1}{2\pi \sqrt{\lambda}}
\hskip-0.5mm
\(\int\limits_{-\infty}^{0}\hskip-0.5mm
\frac{\ln|a_u(s)|^2 \d s}{\sqrt{|s|}}+\ol{_\mathcal{O}}(1)\)\right\}.$$
The jump matrix for $\widehat{N}^{(1)}$ on $\rho$ now equals $$\begin{pmatrix}0&-\i\\-\i&0\end{pmatrix},$$
and now we observe that the function $$M_{N}(x,t;\lambda)=\widehat{N}(x,t;\lambda)F^{\sigma_3}(\lambda)Y^{-1}(\lambda)$$ (with $Y(\lambda)$ defined by \eqref{Y}) solves the following RH problem:
\begin{RH}\label{RHMN}
To find a $2\times2$ matrix-valued function $M_N(x,t;\lambda),$ which
\begin{enumerate}
\item is analytic in $\lambda\in\mathbb{C}\setminus\Sigma,$ $\Sigma=\gamma_0\cup\gamma_r\cup\gamma_{-3}\cup\rho,$
\item has the following jump $M_{N,+}=M_{N,-}J_{M_N}$ accros $\Sigma:$ $\qquad J_{M_N}=$
$$=\begin{cases}
                  I, \ \lambda\in\rho,
\qquad
Y_-\begin{pmatrix}
 1 & 0 \\ \frac{-\i F^2(\lambda)\e^{-2\theta}}{a_ua_d}+\i\e^{-\frac43\lambda^{3/2}} & 1
\end{pmatrix}Y_-^{-1}, \ \lambda\in\gamma_0,
\\\\
Y_-\begin{pmatrix}
 1 & \frac{\i a_u a_d \e^{2\theta}}{F^2(\lambda)} -\i\e^{\frac43\lambda^{3/2}} \\ 0 & 1
\end{pmatrix}Y_-^{-1}, \ \lambda\in\gamma_3,
\quad 
Y_-\begin{pmatrix}
 1 & \frac{\i a_u a_d \e^{2\theta}}{F^2(\lambda)} -\i\e^{\frac43\lambda^{3/2}} \\ 0 & 1
\end{pmatrix}Y_-^{-1}, \ \lambda\in\gamma_{-3},
\end{cases}
$$
\item has the following asymptotics as $\lambda\to\infty,$ which is uniform w.r.t. $\arg\lambda\in[-\pi,0]\cup[0,\pi]:$
$$M_N(x,t;\lambda)=I+\ol{_\mathcal{O}}(1).$$
\end{enumerate}
\end{RH}
We see that the jumps for $M_N$ are exponentially close to $I$ on the infinite parts of the contour $\Sigma,$ and hence from the SIE of the type \eqref{SIE}, which now reads as
$$M_{N,-}=I+\mathcal{C}_{\Sigma,-}\(M_{N,-}(J_{M_N}-I)\),$$
and from the representation of $M_N$ in terms of $M_{N,-}$
$$M_{N}=I+\mathcal{C}_{\Sigma}\(M_{N,-}(J_{M_N}-I)\)$$
we obtain that $M_N$ possesses the full asymptotic expansion
$$M_{N}=I+\sum\limits_{j=1}^{\infty} M_{N,j}\lambda^{-j}+\mathcal{O}(\lambda^{-\infty}).$$

$[4.]$ To show that $\widehat P,$ $\widehat N$ are infinitely many times differentiable with respect to $x,t,$ one proceeds in the same way as in the proof of Theorem \ref{teor_ExistF}. Since the jumps of $M_P,$ $M_N$
are exponentially close to $I$ in the infinite parts of the contour $\Sigma,$ we can differentiate the corresponding SIE as many times as we wish.
Hence, the arguments of Corollary \ref{cor_Cinfty} can be repeated infinitely many times.

$[5.]$ We need only to check the fact that the functions $P(x,t;\lambda),$ $N(x,t;\lambda)$ indeed have asymptotics prescribed by RHPs \ref{RHtp}, \ref{RHtn}, and this can be done by direct computations.
\end{proof}

\subsection{Reconstruction of $u(x,t)$ in terms of $\widehat{\mathbb{F}}.$}

\begin{theorem}\label{teor_u_F}
Let $a_1^{\mathbb{F}},$ $b_1^{\mathbb{F}}$ be the corresponding scalar coefficients in the expansion 
$$\widehat{\mathbb{F}}(x, t; \lambda) = \(I+\begin{pmatrix}a_1^{\mathbb{F}}(x, t)\frac{1}{\lambda}+\ol{_\mathcal{O}}\(\frac{1}{\lambda}\) & b_1^{\mathbb{F}}(x, t)\frac{1}{\lambda}+\ol{_\mathcal{O}}\(\frac{1}{\lambda}\) 
\\ \ol{_\mathcal{O}}(1) & \ol{_\mathcal{O}}(1)\end{pmatrix}\)\frac{1}{\sqrt{2}}\lambda^{-\sigma_3/4}\e^{\theta\sigma_3}.$$
Denote $$u(x,t):=2a_1^{\mathbb{F}}-b_1^{\mathbb{F}}.$$
Then $u(x,t)$ solves the KdV equation \eqref{KdV}.
\end{theorem}
\begin{proof}
The proof goes along the well-known scheme by Zakharov--Shabat. Let $\widehat{\mathbb{F}}$ admit an asymptotic expansion of the form \eqref{asserEhat}, with coefficients $a_j, b_j, c_j, d_j,$ changed with $a_j^{\mathbb{F}}, b_j^{\mathbb{F}}, c_j^{\mathbb{F}}, d_j^{\mathbb{F}},$ and J=3. Suppose that we can differentiate this expansion 1 time w.r.t. $x, t.$
Define $$\mathbb{F}(x,t;\lambda):=\begin{pmatrix}1 & 0 \\ b_1^{\mathbb{F}}(x,t) & 1\end{pmatrix}\widehat{\mathbb{F}}(x,t;\lambda).$$
Notice that the jumps for the function $\widehat{\mathbb{F}}$ do not depend on $x,t,$ hence the ratio 
$$(\widehat{\mathbb{F}})_{x}(\widehat{\mathbb{F}})^{-1},\quad (\widehat{\mathbb{F}})_{t}(\widehat{\mathbb{F}})^{-1}$$
are analytic functions (do not have jumps). By the Liouville theorem, from the asymptotics at $\lambda\to\infty$ we find that
$$\mathbb{F}_{P,x} \mathbb{F}_P^{-1}=:\begin{pmatrix}0 & 1 \\ \lambda-2u_P(x,t) & 0\end{pmatrix},
\quad \mathbb{F}_{P,t} \mathbb{F}_P^{-1}=\begin{pmatrix}\dfrac{u_{P,x}}{6} & \dfrac{\lambda+u_P(x,t)}{-3} \\\\ -\frac{\lambda^2}{3}+\frac{u_P(x,t)\lambda}{3}+\frac{4u_P^2(x,t)+u_{P,xx}}{6}
& \dfrac{-u_{P,x}}{6} \end{pmatrix},$$
where $$u_P(x,t)=2a_1^{\mathbb{F}}(x,t)-(b_1^{\mathbb{F}})^2(x,t)=-\partial_x(b_1^{\mathbb{F}}(x,t)).$$
The consistency condition for the two above differential equations gives us that $u^{\mathbb{F}}(x,t)$ satisfies the KdV equation \eqref{KdV}.

The only delicate moment here is the possibility of expansion of $\widehat{\mathbb{F}}$ at $\lambda\to\infty,$ and the differentiability of $\mathbb{F}$ and its expansion. To address this issue, let us consider the RHP \ref{RHM}.
Its solution can be obtained by the formula 
\[M=\mathbf{1}+\mathcal{C}(M_-\cdot(J_M-I)),\]
where $\ M_--\mathbf{1}\in L_2$ is the solution of the singular integral equation
\[\left[Id-\mathcal{C}_-(.(J_M-\mathbf{1})\right]\circ (M_--\mathbf{1})=
\mathcal{C}_-(J_M-I).\]
Furthermore, we can write the derivative of $M$ as
\[M_x =\mathcal{C}(M_{-,x}\cdot(J_M-\mathbf{1})) + 
\mathcal{C}(M_{-}\cdot J_{M,x} ), \]
where $M_{-,x}$ is the solution of the singular integral equation
\[\left[Id-\mathcal{C}_-(.(J_M-\mathbf{1})\right]\circ M_{-,x}=
\mathcal{C}_-(M_{-}\cdot J_{M,x}).\]
We see that SIE for $M_-$ and for $M_{-,x}$ have the same operator, and hence the derivative $M_{-,x}\in L_2$ exists provided that the r.h.s. is in $L_2,$ i.e. $J_{M,x}$ vanishes sufficiently fast as $\lambda\to-\infty.$ Now, every time we differentiate $J_M$ w.r.t. $x,$ we gain $\lambda^{1/2}$ under the symbol of the integral with the integration path $\lambda =-\infty$ to $\lambda=0,$ and each time we differentiate w.r.t. $t,$ we gain $\lambda^{3/2}.$ 
Furthermore, the rate of convergence of the jump matrix $J_M$ to the identity matrix $\mathbf{1}$ as $\lambda\to-\infty$ depend on the rate of convergence of $r_u(\lambda)$ to 0 as $\lambda\to-\infty.$
In particular, with $r_u(\lambda)=\mathcal{O}(\lambda^{-N})$ for any $N$, we can differentiate $M$ infinitely many times w.r.t. $x,t.$

\end{proof}

\subsection{Characterization of compactly supported perturbations}
\begin{theorem}\label{teor_charact}
The function $u(x,t)$ defined above in Theorem \ref{teor_u_F} possesses the following property: 
$$u(x,t_0)=U(x,t_0)\quad \mbox{ for } x<A \quad \mbox{ and } x>B.$$
\end{theorem}

\begin{remark}It is not a complete characterization, since we do not show that $u(x,t_0)$ obtained from the solution of the RH problem \ref{RHF} is a function of bounded variation for $t=t_0,$ $A\leq x\leq B.$
\end{remark}

\begin{proof}
Let us consider RHP \ref{RHF} with functions $r_u, r_d, a_u, a_d$ satisfying Properties \ref{prop_abr}. 
First we divide $\mathbb{F}$ by $\e^{\theta\sigma_3}:$
$$\Lambda=\mathbb{F}\cdot \e^{-\theta\sigma_3}.$$
Now we consider 2 cases.
\\\noindent\textbf{Case $t=t_0,$ $x>B.$}
 
We apply the transformation
$$\Lambda^{(1)}=\Lambda\cdot\begin{cases}
                             \begin{pmatrix}
                              \frac{1}{1+\i r_u} & 0 \\ -r_u\e^{-2\theta} & 1+\i r_u
                             \end{pmatrix}, \Omega_{II},
				  \qquad
				 \begin{pmatrix}   
				  1 & 0\\-r_u\e^{-2\theta} & 1
				 \end{pmatrix}, \Omega_{I},
\\
\begin{pmatrix}
 \frac{1}{1-\i r_d} & 0\\ -r_d\e^{-2\theta} & 1-\i r_d
\end{pmatrix}, \Omega_{III},\qquad
\begin{pmatrix}
 1 & 0 \\ -r_d\e^{-2\theta} & 1
\end{pmatrix}, \Omega_{IV}.
                            \end{cases}
$$
Function $\Lambda^{(1)}$ is regular at the points $\lambda^*\in\Omega_{II}$ with $r_u(\lambda^*)=\i,$ and the 
jumps of $\Lambda^{(1)}$ are exactly the same as those of $\mathbb{E}$ (RHP \ref{RHE}).
The only issue is the asymptotics as $\lambda\to\infty.$
Since $$r_u(\lambda)\e^{-2\theta(x,t;\lambda)}=\mathcal{O}(\frac{1}{\sqrt{\lambda}})
\cdot\e^{\frac{2(t-t_0)}{3}\lambda^{3/2}+2(B-x)\lambda^{1/2}},$$
the asymptotics for $\Lambda^{(1)}(x,t_0;\lambda)\e^{\theta(x,t_0;\lambda)\sigma_3}$ 
is still of the form \eqref{eq_3a_E_asymp} only when $t=t_0,$ $x>B,$ and hence
$$\Lambda^{(1)}(x,t_0;\lambda)\e^{\theta(x,t_0;\lambda)\sigma_3}=\mathbb{E}(x,t_0;\lambda), \quad x>B.$$

\noindent\textbf{Case $t=t_0,$ $x<A.$}

In this case we apply the transformation
$$\Lambda^{(2)}=\Lambda\cdot
\begin{cases}
 \begin{pmatrix}
  a_u & \(-\i a_d+\frac{\i}{a_u+\i b_u}=\frac{a_ub_d}{a_u+\i b_u}\)\e^{2\theta}\\ 0 & \frac{1}{a_u}
 \end{pmatrix}, \Omega_{II}, \qquad
\begin{pmatrix}
 a_u & -\i(a_d-a_u)\e^{2\theta}\\0&\frac{1}{a_u}
\end{pmatrix}, \Omega_{I},
\\
\begin{pmatrix}
 a_d & \(\i a_u-\frac{\i}{a_d-\i b_d}=\frac{a_db_u}{a_d-\i b_d}\)\e^{2\theta}\\ 0 & \frac{1}{a_d}
\end{pmatrix}, \Omega_{III},
\qquad
\begin{pmatrix}
 a_d & -\i (a_d-a_u)\e^{2\theta}\\0& \frac{1}{a_d}
\end{pmatrix}, \Omega_{IV}.
\end{cases}
$$
Function $\Lambda^{(2)}$ is regular at the points $\lambda^*\in\Omega_{II}$ with $r_u(\lambda^*)=\i,$ and the  
jumps for $\Lambda^{(2)}$ are again exactly the same as for $\mathbb{E}$ (RHP \ref{RHE}).
As to the asymptotics of $\Lambda^{(2)}$ as $\lambda\to\infty,$ we notice that
$$(a_d(\lambda)-a_u(\lambda))\e^{2\theta(x,t;\lambda)}=\mathcal{O}(\frac{1}{\sqrt{\lambda}})\e^{\frac{2(t_0-t)}{3}\lambda^{3/2}+2(x-A)\lambda^{1/2}},$$
and hence we do not spoil the asymptotics when $t=t_0,$ $x<A,$ and hence 
$$\Lambda^{(2)}(x,t_0;\lambda)\e^{\theta(x,t_0;\lambda)\sigma_3}=\mathbb{E}(x,t_0;\lambda)\quad \mbox{ for }\quad x<A.$$
\end{proof}

\subsection{Uniqueness of solution of the Cauchy problem for KdV equation}
\begin{lem}\label{uniqueness_KdV}
Suppose that $v_1(x,t)$ and $v_2(x,t)$ are two solutions of the KdV equation
$$v_t+vv_x+\frac{1}{12}v_{xxx}=0$$
with the same initial data $$v_1(x,t_0)=v_2(x,t_0).$$
Assume that $v_j(x,t)$ is 3 times differentiable in $x$ and 1 time differentiable in $t$ for $t>t_0,$ and that $v_j(x,t_0)$ is 3 times differentiable in $x.$
Suppose also that, writing $\omega(x,t)=v_1(x,t)-v_2(x,t),$
\begin{itemize}
\item $\forall t\geq t_0\quad  \omega(x,t)\to 0,\ \omega\omega_{xx}(x,t)\to 0,\ \omega_x(x,t)\to 0\quad $ as $\quad x\to\pm\infty,$
\item $\forall t\geq t_0 \quad \exists\int\limits_{-\infty}^{+\infty}\omega^2(x,t)\d x<\infty,$
\item $\forall T>t_0\quad \exists\sup\limits_{T\geq t\geq t_0}\sup\limits_{x\in\mathbb{R}}|v_{1,x}(x,t)|=M(T)<\infty.$
\end{itemize}
Then $v_1(x,t)\equiv v_2(x,t)$ for all $t\geq t_0.$
\end{lem}
\begin{proof}
The proof mimicks the one from \cite{Jakovleva}.
Suppose that $v_1(x,t)$ and $v_2(x,t)$ are 2 solutions of KdV. Then their difference $\omega(x,t)=v_1(x,t)-v_2(x,t)$ satisfies the following equation:
$$\omega_t+\frac{1}{12}\omega_{xxx}-\omega\omega_x+2\omega v_{1x}=0.$$

Multiplying the above expression by $\omega$ and integrating it from $A$ to $B,$ we obtain
$$\dsfrac{\d}{\d t}\int\limits_{A}^{B}\frac{\omega^2}{2}\d x+\int\limits_{A}^{B}2\omega^2 v_{1x}\d x+\left.\(\frac{1}{12}\omega_{xx}\omega-\frac{1}{24}\omega_x^2-\frac13\omega^3\)\right|_A^B=0.$$
Now we take the limit $A\to-\infty,$ $B\to+\infty,$ which is possible due to our assumptions on $\omega,$ $v_1,$ and we get 
$$\dsfrac{\d}{\d t}\int\limits_{-\infty}^{\infty}\frac{\omega^2}{2}\d x+\int\limits_{-\infty}^{\infty}2\omega^2 v_{1x}\d x=0.$$
Denote $$q(t)=\int\limits_{-\infty}^{\infty}\frac{\omega^2}{2}\d x,\qquad M(t)=4\sup\limits_{x\in\mathbb{R}} |v_{1,x}(x,t)|\leq M(T),$$
then
$$\dot{q}(t)\leq M(T)q(t),\quad \Rightarrow \dsfrac{\d}{\d t}\(\e^{-M(T)\cdot t}q(t)\)\leq 0\quad \Rightarrow q(t)\leq q(t_0)\e^{M(T)\cdot (t-t_0)},$$
and since $q(t_0)=0,$ then $q(t)=0$ for all $t\geq t_0,$ hence $\omega(x,t)=0$ for all $x\in\mathbb{R},$ $t\geq t_0.$
\end{proof}

\begin{remark}
We do not make any assumptions on the decay at $x\to\pm\infty$ of $v_1(x,t).$ For example, we may take $v_1(x,t)=U(x,t)$ or $v_1(x,t)=\frac{-x}{t}.$ In the first case the derivative $v_{1,x}$ is decaying, in the second case it is bounded.
\end{remark}
\begin{remark}
The discontinuous initial datum $$v_1(x,t_0)=\begin{cases}U(x,t),\ x<A,\\c,\ A<x<B, \\ U(x,t),\ x>B,\end{cases}\qquad \textrm{ or }\qquad
v_1(x,t_0)=\begin{cases}\dsfrac{-x}{t},\ x<A,\\c,\ A<x<B, \\ \frac{-x}{t},\ x>B,\end{cases}$$
where $c$ is a constant, is not within the scope of the lemma.
\end{remark}

\color{black}

\begin{appendices}

\section{Derivation of the inverse Fourier Transform and the Parseval identity.}
\begin{theorem}
For any compactly supported locally integrable function $f(y),$ the following inverse Fourier transform relation is hold:
\[
\frac{1}{2\pi}\int\limits_{-\infty}^{+\infty} e_r(x,t;\lambda)\int\limits_{-\infty}^{+\infty}f(y)e_r(y,t;\lambda)\d\lambda = f(x).
\]
\end{theorem}
\begin{proof}

We follow the well-known ideas from [Titchmarsh(1960)],
[Levitan(1987)]. 
%
%
%
\noindent For a given fixed $t,$ denote 
\[\varphi(x,\lambda) = e_{r}(x,t;\lambda), \quad \psi(x,\lambda) = e_{l}(x,t;\lambda).\]
Furthemore, 
for a given function $f(x),$ define the function
\[\Phi(x,\lambda) = \varphi(x,\lambda)\int\limits^{x}_{-\infty}f(y)\psi(y,\lambda)dy +\psi(x,\lambda)\int_x^{+\infty}f(y)\varphi(y,\lambda)dy.\]
We will integrate it over some contour in the variable
$k=i\sqrt{\lambda}$ from $-N+iM$ to $N+iM,$ $0<M<N,$ in two different ways: in the first case the contour
$C$ consists of intervals $[-N+iM, -N+iN]$, $[-N+iN, N+iN]$,
$[N+iN,N+iM].$ The second contour ${C}_2$ is the union of the intervals
$-\rho_{N,M}^{-}:=[-N+iM,-0]$,
$\gamma_{0,N}^{+}:=[-0,iN-0]$,$-\gamma_{0,N}^{-}:=[iN+0,+0]$,
$\rho_{N,M}^{+}:=[+0,N+iM].$

We take $N$, $M$ to be large, with the constant relation
$\frac{M}{N}<1.$ We notice that $$d\lambda = -2kdk,$$
\begin{equation}\label{phipsi_asymp_1}\varphi(x,\lambda)\psi(y,\lambda)=\frac{e^{\sqrt{\lambda}(-x+y)}}{2\sqrt{\lambda}}(1+\mathrm{o}(1))=
\frac{e^{ik(x-y)}}{-2ik}(1+\mathrm{o}(1)).
\end{equation}
\begin{equation}\label{phipsi_asymp_2}\psi(x,\lambda)\varphi(y,\lambda)=\frac{e^{\sqrt{\lambda}(x-y)}}{2\sqrt{\lambda}}(1+\mathrm{o}(1))=
\frac{e^{-ik(x-y)}}{-2ik}(1+\mathrm{o}(1)),\end{equation}
\begin{equation}\nonumber\int\limits_{C}\Phi(x,\lambda)d\lambda=\int\limits_{C}\Phi(x,\lambda)d\lambda=
\int\limits_{C}\left(\varphi(x,\lambda)\int\limits^{x}_{-\infty}f(y)\psi(y,\lambda)dy+\psi(x,\lambda)\int_x^{+\infty}f(y)\varphi(y,\lambda)dy\right)d\lambda,
\end{equation}
%
Here we can switch the order of integration (for example, we can
take $f$ with a compact support), and then we can use asymptotics
(\ref{phipsi_asymp_1})-(\ref{phipsi_asymp_2}) (the issue whether
those asymptotics are uniform  w.r.t.  $y$ does not arise
here, since again we can take $f$ with compact support).
Then the integral becomes
\[\begin{split}&\int\limits_{-\infty}^x dy f(y)\int_{C}\varphi(x,\lambda)\psi(y,\lambda)d\lambda+\int\limits^{+\infty}_x dy f(y)\int_{C}\psi(x,\lambda)\varphi(y,\lambda)d\lambda=
\\
&=\int\limits_{-\infty}^x dy f(y)\int_{C}(-i)e^{ik (x-y)}(1+\mathrm{o}(1))d k+\int\limits^{+\infty}_x dy f(y)\int_{C}(-i)e^{ik(y-x)}(1+\mathrm{o}(1))d\lambda=
\\
&
=\int\limits_{-\infty}^{+\infty}(-i)f(y)\frac{e^{i|x-y|(N+iM)}-e^{i|x-y|(-N+iM)}}{i|x-y|}(1+\mathrm{o}(1))dy=
\\&
=\int\limits_{-\infty}^{+\infty}(-i)f(y)\frac{e^{-M|x-y|}2i\sin(|x-y|N)}{i|x-y|}(1+\mathrm{o}(1))dy=
\\
&=-2i\int\limits_{-\infty}^{+\infty}f(y)\frac{e^{-M|x-y|}\sin(|x-y|N)}{|x-y|}(1+\mathrm{o}(1))dy.\end{split}\]
Now we make the change of variable $y=x+\frac{\tilde y}{N}.$ Then
the integral becomes
\[=-2i\int\limits_{-\infty}^{+\infty}f\left(x+\frac{\tilde y}{N}\right)\frac{e^{-\frac{M}{N}|\tilde y|}\sin(|\tilde y|)}{|\tilde y|}(1+\mathrm{o}(1))d\tilde y\]
%
When $N\to\infty$, $\frac{M}{N}=const,$ the above integral tends
to
\begin{equation}\label{Pars_contrib_C}=-2if(x)\int\limits_{-\infty}^{+\infty}\frac{e^{-\frac{M}{N}|\tilde y|}\sin|\tilde y|}{|\tilde y|}d\tilde y=-4i\arctan\frac{N}{M}\ f(x).\end{equation}

\noindent\textbf{Integrating over $C_2.$ }Next we integrate over ${C}_2$. We split the contribution of
integration over $\gamma_{0,N}^+, -\gamma_{0,N}^-$, and over
$\rho_{N,M}^+$, $-\rho_{N,M}^-.$
The integral over $\gamma_{0,N}^+, -\gamma_{0,N}^-$ is
\[\begin{split}&\int\limits_{\gamma_{0,N}^+}\varphi(x,\lambda)\psi(y,\lambda)d\lambda + \int\limits_{-\gamma_{0,N}^-}\varphi(x,\lambda)\psi(y,\lambda)d\lambda = 
\int\limits_{\gamma_{0,N}}(\varphi(x,\lambda)\psi_+(y,\lambda)-\varphi(x,\lambda)\psi_-(y,\lambda))d\lambda =
\\&=\int\limits_{\gamma_{0,N}}(\varphi(x,\lambda)\psi_+(y,\lambda)-\varphi(x,\lambda)\psi_-(y,\lambda))d\lambda
=
\int\limits_{\gamma_{0,N}}(\varphi(x,\lambda)(\psi_+(y,\lambda)-\psi_-(y,\lambda))d\lambda 
\\&=
-i\int\limits_{\gamma_{0,N}}\varphi(x,\lambda)\varphi(y,\lambda)d\lambda
\
\begin{array}{ccc}\longrightarrow\\N\to\infty\end{array}\ 
-i\int\limits_{\gamma_{0}}\varphi(x,\lambda)\varphi(y,\lambda)d\lambda.
\end{split}\] Hence, the contributuion from the integrals over
$\gamma_{0,N}^+, -\gamma_{0,N}^-$ is
\[-i\int\limits_{\gamma_{0}}\varphi(x,\lambda)\int\limits_{-\infty}^{+\infty}f(y)\varphi(y,\lambda)dy d\lambda.\]
The integrals over $\rho_{N,M}^+$ $\rho_{N,M}^+$ are
\[\hskip-4cm\int_{-\infty}^xdyf(y)\int\limits_{-\rho_{N,M}^-}\varphi(x,\lambda)\psi(y,\lambda)d\lambda+
\int_{-\infty}^xdyf(y)\int\limits_{\rho_{N,M}^+}\varphi(x,\lambda)\psi(y,\lambda)d\lambda
+\]\[+\int^{+\infty}_xdyf(y)\int\limits_{-\rho_{N,M}^-}\psi(x,\lambda)\varphi(y,\lambda)d\lambda+
\int^{+\infty}_xdyf(y)\int\limits_{\rho_{N,M}^+}\psi(x,\lambda)\varphi(y,\lambda)d\lambda.\]
Since the asymptotics for $\varphi$ is not uniform up to
$\rho^{\pm}$, we will use the identities
$$\varphi(\lambda) = i(\psi(\lambda)-\ol{\psi(\ol{\lambda})}),\ \Im\lambda>0,\qquad \varphi(\lambda) = -i(\psi(\lambda)-\ol{\psi(\ol{\lambda})}),\ \Im\lambda<0.$$
Hence the above sum of integrals can be written in the form
  \begin{equation}\begin{split}&\int\limits_{-\infty}^xdyf(y)\hskip-2mm\int\limits_{-\rho_{N,M}^-}i(\psi(x,\lambda)-\ol{\psi(x,\ol\lambda)})\psi(y,\lambda)d\lambda+
\int\limits_{-\infty}^xdyf(y)\hskip-2mm\int\limits_{\rho_{N,M}^+}(-i)(\psi(x,\lambda)-\ol{\psi(x,\ol\lambda)})\psi(y,\lambda)d\lambda
+\\&
+\int\limits^{+\infty}_xdyf(y)\hskip-2mm\int\limits_{-\rho_{N,M}^-}i(\psi(y,\lambda)-\ol{\psi(y,\ol\lambda)})\psi(x,\lambda)d\lambda+
\int\limits^{+\infty}_xdyf(y)\hskip-2mm\int\limits_{\rho_{N,M}^+}(-i)(\psi(y,\lambda)-\ol{\psi(y,\ol\lambda)})\psi(x,\lambda)d\lambda.\end{split}\label{sum_Parseval_aux}\end{equation}
Now we notice that due to the decay properties at infinity, we have
\[\int\limits_{-\rho_{N,M}^{-}}\psi(x,\lambda)\psi(y,\lambda)d\lambda=\int\limits_{-\rho^{-}}\psi(x,\lambda)\psi(y,\lambda)d\lambda,
\qquad\int\limits_{\rho_{N,M}^{+}}\psi(x,\lambda)\psi(y,\lambda)d\lambda=\int\limits_{\rho^{+}}\psi(x,\lambda)\psi(y,\lambda)d\lambda.\]
Then the input from the terms containing
$\psi(x,\lambda)\psi(y,\lambda)$ is
\[\begin{split}&\int\limits_{-\infty}^x dy f(y)\int\limits_{-\rho^-}i\psi(x,\lambda)\psi(y,\lambda)d\lambda
+ \int\limits^{\infty}_x dy
f(y)\int\limits_{-\rho^-}i\psi(x,\lambda)\psi(y,\lambda)d\lambda +
\\&
\int\limits_{-\infty}^xdyf(y)\int\limits_{\rho^+}-i\psi(x,\lambda)\psi(y,\lambda)d\lambda+
\int\limits^{\infty}_xdyf(y)\int\limits_{\rho^+}-i\psi(x,\lambda)\psi(y,\lambda)d\lambda=
\\
&
=\int\limits_{-\infty}^{\infty} dy f(y)\int\limits_{\rho^-}(-i)\psi(x,\lambda)\psi(y,\lambda)d\lambda+
\int\limits_{-\infty}^{\infty} dy
f(y)\int\limits_{\rho^+}(-i)\psi(x,\lambda)\psi(y,\lambda)d\lambda.\end{split}\]
Now, since $\psi(x,\lambda-i0)=\ol{\psi(x,\lambda+i0)}$ for
$\lambda\in\mathbb{R},$ the above expression becomes
\[=-i\int\limits_{-\infty}^{\infty} dy f(y)\int\limits_{\rho^-}(\psi(x,\lambda)\psi(y,\lambda)+\ol{\psi(x,\lambda)}\ol{\psi(y,\lambda)})
d\lambda.\]
Integrals in (\ref{sum_Parseval_aux}), which contain
$\ol{\psi(x,\ol{\lambda})}\psi(y,\lambda),\ $ $
\ol{\psi(y,\ol{\lambda})}\psi(x,\lambda),$ can be treated in the
following way. Consider, for example,
\begin{equation}\label{sum_Parseval_aux_2}-i\int\limits_{-\infty}^x\hskip-2mmdyf(y)\hskip-3mm\int\limits_{-\rho_{N,M}^-}\hskip-3mm\ol{\psi(x,\ol{\lambda})}\psi(y,\lambda)d\lambda=
-i\int\limits_{-\infty}^xdyf(y)\left(\int\limits_{-\rho_{N,0}^-}\hskip-2mm\ol{\psi(x,\ol{\lambda})}\psi(y,\lambda)d\lambda
-\hskip-2mm\int\limits_{-N}^{-N+iM}\hskip-2mm\ol{\psi(x,\ol{\lambda})}\psi(y,\lambda)d\lambda
\right).
\end{equation}
Here the first integral is already in the form which is OK for us,
and to compute the second integral, we use the asymptotics
$$\psi(x,\lambda)=\frac{\exp\left\{\frac{1}{105}\lambda^{7/2}-\frac{t}{3}\lambda^{3/2}+x{\lambda^{1/2}}\right\}}{\sqrt{2}\sqrt[4]{\lambda}}(1+o(1)),
\quad\lambda\in\mathbb{C}\setminus\mathbb{R},$$
$$\ol{\psi(x,\ol{\lambda})}=i\,\mathrm{sign}(\Im\lambda)\frac{\exp\left\{\frac{-1}{105}\lambda^{7/2}+\frac{t}{3}\lambda^{3/2}-x{\lambda^{1/2}}\right\}}{\sqrt{2}\sqrt[4]{\lambda}},
\quad\arg\lambda\in(\pi-\varepsilon,\pi-0]\cup[-\pi+0,-\pi+\varepsilon),$$

Hence, for $x>y$ we have
\[\begin{split}&\int\limits_{-N}^{-N+iM}\ol{\psi(x,\ol{\lambda})}\psi(y,\lambda)d\lambda=
\int\limits_{-N}^{-N+iM}\frac{ie^{(y-x)\sqrt{\lambda}}}{2\sqrt{\lambda}}(1+\mathrm{o}(1))d\lambda=
\int\limits_{-N}^{-N+iM}\frac{ie^{-i(y-x)k}}{-2ik}(-2k)dk=
\\&=\int\limits_{-N}^{-N+iM}e^{ik(x-y)}dk
=\frac{e^{-i(x-y)N}\(e^{-M(x-y)}-1\)}{i(x-y)}(1+\mathrm{o}(1)),\end{split}\]
then the second sum in the integral (\ref{sum_Parseval_aux_2})
becomes
\[i\int\limits_{-\infty}^xdyf(y)\frac{e^{-i(x-y)N}\(e^{-M(x-y)}-1\)}{i(x-y)},\]
and after change of variables $ y=x+\frac{\tilde y}{N}$ we get
\[\begin{split}-\int\limits_{-\infty}^0d\tilde y\ f\(x+\frac{\tilde y}{N}\)\frac{e^{i\tilde y}\(e^{\frac{M}{N}\tilde y}-1\)}{\tilde y}&\to
-f(x)\int_{-\infty}^0\frac{e^{i\tilde y}\(e^{\frac{M}{N}\tilde
y}-1\)}{\tilde y} d\tilde y=
\\&=f(x)\left[i\(\frac{\pi}{2}-\arctan\frac{N}{M}\)-\frac{1}{2}\ln\(\frac{M^2}{N^2}+1\)\right]\end{split}\]
Hence, that  term tends to 
\[\begin{split}-i\int\limits_{-\infty}^xdyf(y)\int\limits_{-\rho_{N,M}^-}\ol{\psi(x,\ol{\lambda})}\psi(y,\lambda)d\lambda&\to
-i\int\limits_{-\infty}^xdyf(y)\int\limits_{-\rho_{N,0}^-}\ol{\psi(x,\ol{\lambda})}\psi(y,\lambda)d\lambda
\\&+f(x)\left[i\(\frac{\pi}{2}-\arctan\frac{N}{M}\)-\frac{1}{2}\ln\(\frac{M^2}{N^2}+1\)\right].
\end{split}\]
Similarly,
\[\begin{split}\int\limits_{x}^{+\infty}dyf(y)\int\limits_{-\rho_{N,M}^-}(-i)\ol{\psi(y,\ol{\lambda})}\psi(x,\lambda)d\lambda\to&
\int\limits_{x}^{+\infty}dyf(y)\int\limits_{-\rho_{N,0}^-}(-i)\ol{\psi(y,\ol{\lambda})}\psi(x,\lambda)d\lambda
\\&+f(x)\left[i\left(\frac{\pi}{2}-\arctan\frac{N}{M}\right)-\frac{1}{2}\ln\(\frac{M^2}{N^2}+1\)\right],
\end{split}\]
\[\begin{split}\int\limits^{x}_{-\infty}dyf(y)\int\limits_{\rho_{N,M}^+}i\ol{\psi(x,\ol{\lambda})}\psi(y,\lambda)d\lambda&\to
\int\limits^{x}_{-\infty}dyf(y)\int\limits_{\rho_{N,0}^+}i\ol{\psi(x,\ol{\lambda})}\psi(y,\lambda)d\lambda
\\&+f(x)\left[i\left(\frac{\pi}{2}-\arctan\frac{N}{M}\right)+\frac{1}{2}\ln\(\frac{M^2}{N^2}+1\)\right],\end{split}\]
\[\begin{split}\int\limits_{x}^{+\infty}dyf(y)\int\limits_{\rho_{N,M}^+}i\,\ol{\psi(y,\ol{\lambda})}\psi(x,\lambda)d\lambda&\to
\int\limits_{x}^{+\infty}dyf(y)\int\limits_{\rho_{N,0}^+}i\,\ol{\psi(y,\ol{\lambda})}\psi(x,\lambda)d\lambda
+\\&f(x)\left[i\left(\frac{\pi}{2}-\arctan\frac{N}{M}\right)+\frac{1}{2}\ln\(\frac{M^2}{N^2}+1\)\right].\end{split}\]
Summing up the last four expressions, and taking into account the
property $\psi(x,\lambda-i0)=\ol{\psi(x,\lambda+i0)},$
$\lambda\in\mathbb{R},$ we can combine the integrals in the 2$^{nd}$
and the 3$^{rd}$ lines, and in the first and the fourth lines. We
get
\[-i\int\limits_{-\infty}^{\infty}dyf(y)\int\limits_{-\rho_{N,0}^-}\(\ol{\psi(x,\ol{\lambda})}\psi(y,\lambda)+
\ol{\psi(y,\ol{\lambda})}\psi(x,\lambda)\)d\lambda+f(x)i\(2\pi-4\arctan\frac{N}{M}\).\]
Summing up, the integral
$\int\limits_{\mathcal{C}_2}\Phi(x,t;\lambda)d\lambda$ equals
\[\begin{split}&f(x)i\(2\pi-4\arctan\frac{N}{M}\)-i\int\limits_{\gamma_0}\varphi(x,\lambda)\int\limits_{-\infty}^{+\infty}f(y)\varphi(y,\lambda)d\lambda-
\\
&-i\int\limits_{-\infty}^{+\infty}f(y)\int\limits_{\rho^-}\(\psi(x,\lambda)\psi(y,\lambda)+\ol{\psi(x,\lambda)} \ol{\psi(y,\lambda)}-
\ol{\psi(x,{\lambda})}\psi(y,\lambda)-
\ol{\psi(y,{\lambda})}\psi(x,\lambda)\)d\lambda,\end{split}\] and, using the
relation $\varphi(\lambda)=i(\psi(\lambda)-\ol{\psi(\lambda)}),$
$\lambda\in\rho_-,$ we get
\[\hskip-0cmf(x)i\(2\pi-4\arctan\frac{N}{M}\)-i\int\limits_{\gamma_0}\varphi(x,\lambda)\int\limits_{-\infty}^{+\infty}f(y)\varphi(y,\lambda)d\lambda+
i\int\limits_{-\infty}^{+\infty}f(y)\int\limits_{\rho^-}\varphi(x,\lambda)\varphi(y,\lambda)d\lambda,\]
and, since $\rho^-$ has the direction from $0$ to $-\infty,$ the
last expression is equal to
\begin{equation}\label{Pars_contrib_C2}\hskip-3cmf(x)i\(2\pi-4\arctan\frac{N}{M}\)-i\int\limits_{-\infty}^{\infty}\varphi(x,\lambda)\int\limits_{-\infty}^{+\infty}f(y)\varphi(y,\lambda)d\lambda.
\end{equation}
Comparing $\int\limits_C\Phi(x,t;\lambda)d\lambda$
(\ref{Pars_contrib_C}) and
$\int\limits_{\mathcal{C}_2}\Phi(x,t;\lambda)d\lambda$
(\ref{Pars_contrib_C2}),
 we come to conclusion that the following inverse Fourier transform formula holds true:
\begin{equation}\label{Inverse_Fourier_identity}
 f(x)=\frac{1}{2\pi}\int\limits_{-\infty}^{\infty}\varphi(x,\lambda)\int\limits_{-\infty}^{+\infty}f(y)\varphi(y,\lambda)dyd\lambda.
\end{equation}
\end{proof}

\begin{remark} Let us notice, that when we wrote integrals of
$\int\limits_{x}^{+\infty}\psi(y,\lambda)$, we must take $f(y)=0$,
$y\to\infty$, but in the final formula
(\ref{Inverse_Fourier_identity}) we don't have this restriction.

\noindent In the symbolic form (\ref{Inverse_Fourier_identity}) can be
written as
\[ \delta(x-y)=\frac{1}{2\pi}\int\limits_{-\infty}^{\infty}\varphi(x,\lambda)\varphi(y,\lambda)d\lambda.\]
Multiplying (\ref{Inverse_Fourier_identity}) by a function
$\ol{g(x)}$ and integrating over $\mathbb{R}dx,$ we come to the
Parseval identity
\begin{equation}\label{Parseval_identity}
 \int\limits_{-\infty}^{+\infty}f(x)\ol{g(x)}dx=\frac{1}{2\pi}\int\limits_{-\infty}^{\infty}d\lambda\int\limits_{-\infty}^{\infty}\ol{g(x)}\varphi(x,\lambda)dx
\int\limits_{-\infty}^{+\infty}f(y)\varphi(y,\lambda)dy .
\end{equation}
\end{remark}

\begin{remark}The analogy with
(\ref{Inverse_Fourier_identity}) with the potential $x$ instead of
$U(x,t)$ is
$$\delta(x-y)=\frac{1}{2\pi}\int\limits_{-\infty}^{+\infty} \sqrt{2\pi}Ai(x+\lambda)\ \sqrt{2\pi}Ai(y+\lambda)d\lambda.$$
\end{remark}

\end{appendices}

\textbf{Acknowledgments.}
A.M. would like to thank Alberto Maspero and Pieter Roffelsen for reading a version of this manuscript and giving useful suggestions, and also to Tom Claeys, who pointed out at the reference \cite[p.18]{IKO}, where a statement similar to Proposition \ref{prop_vanish} is proven.
A.M. acknowledges the support of the H2020-MSCA-RISE-2017 PROJECT No. 778010 IPADEGAN.



\begin{thebibliography}{}
\bibitem[Bona Smith'75]{BS}Bona, J. L.; Smith, R. The initial-value problem for the Korteweg--de Vries equation. Philos. Trans. Roy. Soc. London Ser. A 278 (1975), no. 1287, 555–601.

\bibitem[Bondareva Shubin'82] {BSh} Bondareva, I. N.; Shubin, M. A. Growing asymptotic solutions of the Korteweg-de Vries equation and of its higher analogues. (Russian) Dokl. Akad. Nauk SSSR 267 (1982), no. 5, 1035–1038.

\bibitem[B\"{o}ttcher Karlovich'97]{BK97} B\"{o}ttcher Albrecht, Karlovich Yuri I. Carleson curves, Muckenhoupt weights, and Toeplitz operators. \href{https://doi.org/10.1007/978-3-0348-8922-3}{Progress in Mathematics}, 154. Birkhäuser Verlag, Basel, 1997.

\bibitem[BET'08]{BET08} A. Boutet de Monvel, I. Egorova, G. Teschl. Inverse scattering theory for one-dimensional Schrödinger operators with steplike finite-gap potentials, J. d'Analyse Math. 106:1, 271-316 (2008).


\bibitem[BMP'90]{BrezinMarinariParisi} Brezin E, Marinari E, Parisi G, A nonperturbative ambiguity free solution of a string model. Phys. Lett. B 242, 35--38, 1990.

\bibitem[Claeys Grava'09]{CG09}
Claeys T, Grava T. Universality of the break-up profile for the KdV equation in the small dispersion limit using the Riemann-Hilbert approach. Comm. Math. Phys. 286 (2009), no. 3, 979--1009.

 \bibitem[DKMVZ'99]{DKMVZ} Deift P, Kriecherbauer T, McLaughlin K T-R, Venakides S and Zhou X 1999 \textrm{Uniform asymptotics for polynomials orthogonal  w.r.t.  varying exponential weights and applications to universality questions in random matrix theory}
  Commun. Pure Appl Math 52, 1335--1425, 1999.
 \bibitem[Dubrovin'06]{Dubrovin} Dubrovin, Boris. On Hamiltonian perturbations of hyperbolic systems of conservation laws. II. Universality of critical behaviour. Comm. Math. Phys. 267 (2006), no. 1, 117–139.
 
 
 \bibitem[Egorova Teschl'11]{Egorova_Teschl} Egorova, Iryna; Teschl, Gerald On the Cauchy problem for the Korteweg--de Vries equation with steplike finite-gap initial data II. Perturbations with finite moments. J. Anal. Math. 115 (2011), 71–101.
 \bibitem[Claeys'10]{Claeys10}Claeys, T. Asymptotics for a special solution to the second member of the Painlevé I hierarchy. J. Phys. A 43 (2010), no. 43, 434012, 18 pp.
 \bibitem[Claeys'12]{Claeys12} Claeys, Tom. Pole-free solutions of the first Painlevé hierarchy and non-generic critical behavior for the KdV equation. Phys. D 241 (2012), no. 23-24, 2226–2236.
 \bibitem[Claeys Vanlessen'07]{Claeys Vanlessen} Claeys, T.; Vanlessen, M. The existence of a real pole-free solution of the fourth order analogue of the Painlevé I equation. Nonlinearity 20 (2007), no. 5, 1163–1184.
 \bibitem[Kapaev'95]{Kapaev}  Kapaev, A. A. Weakly nonlinear solutions of the equation $P_I^2.$ (Russian) Zap. Nauchn. Sem. Leningrad. Otdel. Mat. Inst. Steklov. (LOMI) 187 (1991), Differentsialnaya Geom. Gruppy Li i Mekh. 12, 88--109, 172--173, 175; translation in J. Math. Sci. 73 (1995), no. 4, 468--481.
 \bibitem[Grava Kapaev Klein'15]{Grava Kapaev Klein} Grava, Tamara; Kapaev, Andrei; Klein, Christian. On the tritronquée solutions of P2I. Constr. Approx. 41 (2015), no. 3, 425–466.
 \bibitem[CG]{Claeys Grava} Claeys, T.; Grava, T. Universality of the break-up profile for the KdV equation in the small dispersion limit using the Riemann-Hilbert approach. Comm. Math. Phys. 286 (2009), no. 3, 979–1009.
\bibitem[Its Kuijlaars $\ddot{\mathrm{O}}$stensson'08]{IKO} A.R. Its, A.~Kuijlaars and J. Östensson. Critical edge behavior in unitary random matrix ensembles and the thirty fourth Painlevé transcendent, Int. Math. Res. Not. 2008 (2008), article ID rnn017, 67 pages. (arXiv:0704.1972)

\bibitem[Jakovleva]{Jakovleva} Jakovleva, Master thesis, 1970.

\bibitem[Kadets'06]{Kadets}  Kadets, V. M. A course in functional analysis, (Russian) Kharkovskii Natsionalnyi Universitet imeni V. N. Karazina, Kharkiv, 2006. 608 pp. ISBN: 966-623-199-9.
\bibitem[Kenig Ponce Vega'97]{KPV}Kenig, Carlos E.; Ponce, Gustavo; Vega, Luis. Global solutions for the KdV equation with unbounded data. J. Differential Equations 139 (1997), no. 2, 339–364.

\bibitem[Lenells'18]{L18}Lenells Jonatan. Matrix Riemann-Hilbert problems with jumps across Carleson contours. Monatsh. Math. 186 (2018), no. 1, 111–152.


\bibitem[Maspero Schaad'16]{MaS}Maspero, Alberto; Schaad, Beat One smoothing property of the scattering map of the KdV on R. Discrete Contin. Dyn. Syst. 36 (2016), no. 3, 1493--1537.
\bibitem[Menikoff'72]{Menikoff} Menikoff, A. The existence of unbounded solutions of the Korteweg--de Vries equation. Comm. Pure Appl. Math. 25 (1972), 407–432. 
\bibitem[Pedersen'89]{Pedersen}Pedersen, Gert K., Analysis now. Graduate Texts in Mathematics, 118. Springer-Verlag, New York, 1989. xiv+277 pp. ISBN: 0-387-96788-5
\bibitem[Reed Simon'78]{RS} Reed M and Simon B. 1978 \textit{Methods of Modern Mathematical Physics IV} (New York: Academic)
\bibitem[Rybkin'11]{Rybkin} Rybkin, Alexei. The Hirota $\tau$-function and well-posedness of the KdV equation with an arbitrary step-like initial profile decaying on the right half line. Nonlinearity 24 (2011), no. 10, 2953–2990. 
\bibitem[Suleimanov'13]{Sul} Suleimanov, B. I. Asymptotics of the Gurevich-Pitaevskii universal special solution of the Korteweg--de Vries equation as |x|→∞. Proc. Steklov Inst. Math. 281 (2013), suppl. 1, S137–S145.
\bibitem[Zhou'89]{Zhou89} Zhou, X. The Riemann-Hilbert problem and inverse scattering. SIAM. J. Math. Anal., Vol.20, N 4, pp.966--986, 1989.
\end{thebibliography}
\end{document}